\documentclass[11pt]{article}

\def\showauthornotes{0}

\def\showkeys{0}
\def\showdraftbox{0}
\def\showcolorlinks{1}
\def\usemicrotype{1}
\def\showfixme{0}

\usepackage{etex}

\usepackage[l2tabu, orthodox]{nag}

\usepackage{xspace,enumerate}

\usepackage[dvipsnames]{xcolor}

\usepackage[american]{babel}

\usepackage{mathtools}

\usepackage{amsthm}

\newtheorem{theorem}{Theorem}[section]
\newtheorem*{theorem*}{Theorem}

\newtheorem{proposition}[theorem]{Proposition}
\newtheorem*{proposition*}{Proposition}
\newtheorem{lemma}[theorem]{Lemma}
\newtheorem*{lemma*}{Lemma}

\newtheorem*{conjecture*}{Conjecture}
\newtheorem{fact}[theorem]{Fact}
\newtheorem*{fact*}{Fact}

\newtheorem*{hypothesis*}{Hypothesis}

\theoremstyle{definition}
\newtheorem{definition}[theorem]{Definition}
\newtheorem*{definition*}{Definition}

\newtheorem{algorithm}[theorem]{Algorithm}
\newtheorem{problem}[theorem]{Problem}
\newtheorem*{problem*}{Problem}

\theoremstyle{remark}
\newtheorem{claim}[theorem]{Claim}
\newtheorem*{claim*}{Claim}
\newtheorem{remark}[theorem]{Remark}
\newtheorem*{remark*}{Remark}

\newtheorem*{observation*}{Observation}

\usepackage[
letterpaper,
top=0.9in,
bottom=0.9in,
left=0.8in,
right=0.8in]{geometry}

\usepackage{amsmath,amsfonts,amssymb}

\ifnum\showkeys=1
\usepackage[color]{showkeys}
\fi

\ifnum\showcolorlinks=1
\usepackage[
pagebackref,
colorlinks=true,
urlcolor=blue,
linkcolor=blue,
citecolor=OliveGreen,
]{hyperref}
\fi

\ifnum\showcolorlinks=0
\usepackage[
pagebackref,
colorlinks=false,
pdfborder={0 0 0}
]{hyperref}
\fi

\usepackage{prettyref}
\newcommand{\pref}{\prettyref}

\newcommand{\savehyperref}[2]{\texorpdfstring{\hyperref[#1]{#2}}{#2}}

\newrefformat{eq}{\savehyperref{#1}{\textup{(\ref*{#1})}}}
\newrefformat{eqn}{\savehyperref{#1}{\textup{(\ref*{#1})}}}
\newrefformat{lem}{\savehyperref{#1}{Lemma~\ref*{#1}}}
\newrefformat{def}{\savehyperref{#1}{Definition~\ref*{#1}}}
\newrefformat{thm}{\savehyperref{#1}{Theorem~\ref*{#1}}}
\newrefformat{cor}{\savehyperref{#1}{Corollary~\ref*{#1}}}
\newrefformat{cha}{\savehyperref{#1}{Chapter~\ref*{#1}}}
\newrefformat{sec}{\savehyperref{#1}{Section~\ref*{#1}}}
\newrefformat{app}{\savehyperref{#1}{Appendix~\ref*{#1}}}
\newrefformat{tab}{\savehyperref{#1}{Table~\ref*{#1}}}
\newrefformat{fig}{\savehyperref{#1}{Figure~\ref*{#1}}}
\newrefformat{hyp}{\savehyperref{#1}{Hypothesis~\ref*{#1}}}
\newrefformat{alg}{\savehyperref{#1}{Algorithm~\ref*{#1}}}
\newrefformat{rem}{\savehyperref{#1}{Remark~\ref*{#1}}}
\newrefformat{item}{\savehyperref{#1}{Item~\ref*{#1}}}
\newrefformat{step}{\savehyperref{#1}{step~\ref*{#1}}}
\newrefformat{conj}{\savehyperref{#1}{Conjecture~\ref*{#1}}}
\newrefformat{fact}{\savehyperref{#1}{Fact~\ref*{#1}}}
\newrefformat{prop}{\savehyperref{#1}{Proposition~\ref*{#1}}}
\newrefformat{prob}{\savehyperref{#1}{Problem~\ref*{#1}}}
\newrefformat{claim}{\savehyperref{#1}{Claim~\ref*{#1}}}
\newrefformat{relax}{\savehyperref{#1}{Relaxation~\ref*{#1}}}
\newrefformat{red}{\savehyperref{#1}{Reduction~\ref*{#1}}}
\newrefformat{part}{\savehyperref{#1}{Part~\ref*{#1}}}
\newrefformat{obs}{\savehyperref{#1}{Observation~\ref*{#1}}}
\newrefformat{corr}{\savehyperref{#1}{Corollary~\ref*{#1}}}

\newcommand{\Sref}[1]{\hyperref[#1]{\S\ref*{#1}}}

\usepackage{nicefrac}
\newcommand{\half}{\nicefrac12}

\ifnum\usemicrotype=1
\usepackage{microtype}
\fi

\ifnum\showauthornotes=1
\newcommand{\Authornote}[2]{{\sffamily\small\color{red}{[#1: #2]}}}
\newcommand{\Authornotecolored}[3]{{\sffamily\small\color{#1}{[#2: #3]}}}
\newcommand{\Authorcomment}[2]{{\sffamily\small\color{gray}{[#1: #2]}}}
\newcommand{\Authorstartcomment}[1]{\sffamily\small\color{gray}[#1: }

\newcommand{\Authorfnote}[2]{\footnote{\color{red}{#1: #2}}}
\newcommand{\Authorfixme}[1]{\Authornote{#1}{\textbf{??}}}
\newcommand{\Authormarginmark}[1]{\marginpar{\textcolor{red}{\fbox{\Large #1:!}}}}
\else
\newcommand{\Authornote}[2]{}
\newcommand{\Authornotecolored}[3]{}
\newcommand{\Authorcomment}[2]{}
\newcommand{\Authorstartcomment}[1]{}

\newcommand{\Authorfnote}[2]{}
\newcommand{\Authorfixme}[1]{}
\newcommand{\Authormarginmark}[1]{}
\fi

\ifnum\showfixme=0

\fi

\usepackage{boxedminipage}

\newcommand{\Paren}[1]{\left(#1\right)}

\newcommand{\card}[1]{\lvert#1\rvert}

\newcommand{\norm}[1]{\lVert#1\rVert}
\newcommand{\Norm}[1]{\left\lVert#1\right\rVert}

\newcommand{\iprod}[1]{\langle#1\rangle}
\newcommand{\Iprod}[1]{\left\langle#1\right\rangle}

\newcommand{\Esymb}{\mathbb{E}}
\newcommand{\Psymb}{\mathbb{P}}
\newcommand{\Vsymb}{\mathbb{V}}

\DeclareMathOperator*{\E}{\Esymb}
\DeclareMathOperator*{\Var}{\Vsymb}
\DeclareMathOperator*{\ProbOp}{\Psymb}
\DeclareMathOperator*{\pE}{{\tilde\Esymb}}

\renewcommand{\Pr}{\ProbOp}

\newcommand{\tensor}{\otimes}

\newcommand{\textparen}[1]{\text{(#1)}}

\ifx\because\undefined
\newcommand{\because}[1]{\textparen{because #1}}
\else
\renewcommand{\because}[1]{\textparen{because #1}}
\fi

\newcommand{\sbits}{\{\pm1\}}

\newcommand{\defeq}{\stackrel{\mathrm{def}}=}

\newcommand{\mper}{\,.}
\newcommand{\mcom}{\,,}

\newcommand\bdot\bullet

\ifx\mathds\undefined % use double stroke fonts if available
\DeclareMathOperator{\Ind}{\mathbb{I}}
\else
\DeclareMathOperator{\Ind}{\mathds 1}}
\fi

\DeclareMathOperator{\Tr}{Tr}

\DeclareMathOperator{\opt}{opt}

\DeclareMathOperator{\polylog}{polylog}

\newcommand{\etal}{et al.\xspace}

\newcommand{\N}{\mathbb N}
\newcommand{\R}{\mathbb R}

\newcommand{\cA}{\mathcal A}

\newcommand{\cC}{\mathcal C}
\newcommand{\cD}{\mathcal D}
\newcommand{\cE}{\mathcal E}
\newcommand{\cF}{\mathcal F}
\newcommand{\cG}{\mathcal G}
\newcommand{\cH}{\mathcal H}
\newcommand{\cI}{\mathcal I}

\newcommand{\cM}{\mathcal M}

\newcommand{\cP}{\mathcal P}

\newcommand{\cS}{\mathcal S}

\newcommand{\cU}{\mathcal U}
\newcommand{\cV}{\mathcal V}

\newcommand{\cX}{\mathcal X}

\ifnum\showdraftbox=1
\newcommand{\draftbox}{\begin{center}
  \fbox{%
    \begin{minipage}{2in}%
      \begin{center}%
          \Large\textsc{Working Draft}\\%
        Please do not distribute%
      \end{center}%
    \end{minipage}%
  }%
\end{center}
\vspace{0.2cm}}
\else
\newcommand{\draftbox}{}
\fi

\let\epsilon=\varepsilon

\numberwithin{equation}{section}

\newcommand\MYcurrentlabel{xxx}

\newcommand{\MYstore}[2]{%
  \global\expandafter \def \csname MYMEMORY #1 \endcsname{#2}%
}

\newcommand{\MYload}[1]{%
  \csname MYMEMORY #1 \endcsname%
}

\newcommand{\MYnewlabel}[1]{%
  \renewcommand\MYcurrentlabel{#1}%
  \MYoldlabel{#1}%
}

\newcommand{\MYdummylabel}[1]{}

\newcommand{\torestate}[1]{%
  \let\MYoldlabel\label%
  \let\label\MYnewlabel%
  #1%
  \MYstore{\MYcurrentlabel}{#1}%
  \let\label\MYoldlabel%
}

\newcommand{\restatetheorem}[1]{%
  \let\MYoldlabel\label
  \let\label\MYdummylabel
  \begin{theorem*}[Restatement of \prettyref{#1}]
    \MYload{#1}
  \end{theorem*}
  \let\label\MYoldlabel
}

\newcommand{\restatedef}[1]{%
  \let\MYoldlabel\label
  \let\label\MYdummylabel
  \begin{definition*}[Restatement of \prettyref{#1}]
    \MYload{#1}
  \end{definition*}
  \let\label\MYoldlabel
}

\newcommand{\restatelemma}[1]{%
  \let\MYoldlabel\label
  \let\label\MYdummylabel
  \begin{lemma*}[Restatement of \prettyref{#1}]
    \MYload{#1}
  \end{lemma*}
  \let\label\MYoldlabel
}

\newcommand{\restateprop}[1]{%
  \let\MYoldlabel\label
  \let\label\MYdummylabel
  \begin{proposition*}[Restatement of \prettyref{#1}]
    \MYload{#1}
  \end{proposition*}
  \let\label\MYoldlabel
}

\newcommand{\restatefact}[1]{%
  \let\MYoldlabel\label
  \let\label\MYdummylabel
  \begin{fact*}[Restatement of \prettyref{#1}]
    \MYload{#1}
  \end{fact*}
  \let\label\MYoldlabel
}

\newcommand{\restateobs}[1]{%
  \let\MYoldlabel\label
  \let\label\MYdummylabel
  \begin{observation*}[Restatement of \prettyref{#1}]
    \MYload{#1}
  \end{observation*}
  \let\label\MYoldlabel
}
\newcommand{\restate}[1]{%
  \let\MYoldlabel\label
  \let\label\MYdummylabel
  \MYload{#1}
  \let\label\MYoldlabel
}

\newcommand{\addreferencesection}{
  \phantomsection
  \addcontentsline{toc}{section}{References}
}

\let\origparagraph\paragraph
\renewcommand{\paragraph}[1]{\origparagraph{#1.}}

\allowdisplaybreaks

\newcommand{\cclassmacro}[1]{\texorpdfstring{\textbf{#1}}{#1}\xspace}

\newcommand{\np}{\cclassmacro{NP}}

\usepackage[newenum,newitem,increaseonly]{paralist}
\usepackage{enumitem}

\usepackage{comment}

\let\citet\cite

\theoremstyle{definition}

\DeclareUrlCommand\email{}

\newcommand{\restateproblem}[2]{%
  \let\MYoldlabel\label
  \let\label\MYdummylabel
  \begin{problem*}[Restatement of \prettyref{#1}, {#2}]
    \MYload{#1}
    \end{problem*}
  \let\label\MYoldlabel
}
\usepackage{bm}

\usepackage{tikz}
\newcommand{\bT}{{\bf T}}
\newcommand{\bA}{{\bf A}}

\newcommand{\squares}{\mathop{\textrm{squares}}}

\newcommand{\nsurp}{\mathrm{\# new}}
\newcommand{\nrecyc}{\mathrm{\# reused}}
\newcommand{\nhigh}{\mathrm{\# high}}
\newcommand{\nunf}{\mathrm{\# unforced}}
\newcommand{\nshare}{\mathrm{\# share}}

\newcommand{\nret}{\mathrm{\# return}}

\newcommand{\ty}{t}

\newcommand{\sos}{{SoS}\xspace }

\newcommand{\hght}{h}
\newcommand{\wdth}{w}

\newcommand{\tO}{\widetilde{O}}
\newcommand{\tOm}{\widetilde{\Omega}}

\newcommand{\onorm}[1]{\lVert#1\rVert_{\oplus}}
\newcommand{\dnorm}[1]{\lVert#1\rVert_{0}}

\title{Strongly Refuting Random CSPs Below the Spectral Threshold}

\author{%
\normalsize
Prasad Raghavendra \thanks{UC Berkeley,
  \protect\email{prasad@cs.berkeley.edu}. Supported by NSF Career
  Award, NSF CCF-1407779 and the Alfred. P. Sloan Fellowship. }
\and
\normalsize
Satish Rao\thanks{UC Berkeley, \protect\email{satishr@cs.berkeley.edu}.}
\and
\normalsize
Tselil Schramm\thanks{UC Berkeley, \protect\email{tschramm@cs.berkeley.edu}.
Supported by an NSF Graduate Research Fellowship (NSF award no 1106400).}
}

\date{}

\begin{document}

\maketitle

\draftbox

\thispagestyle{empty}

\begin{abstract}
    Random constraint satisfaction problems (CSPs) are known to exhibit threshold phenomena: given a uniformly random instance of a CSP with $n$ variables and $m$ clauses, there is a value of $m = \Omega(n)$ beyond which the CSP will be unsatisfiable with high probability.
    Strong refutation is the problem of certifying that no variable assignment satisfies more than a constant fraction of clauses; this is the natural algorithmic problem in the unsatisfiable regime (when $m/n = \omega(1)$).

    Intuitively, strong refutation should become easier as the clause density $m/n$ grows, because the contradictions introduced by the random clauses become more locally apparent.
    For CSPs such as $k$-SAT and $k$-XOR, there is a long-standing gap between the clause density at which efficient strong refutation algorithms are known, $m/n \ge \widetilde O(n^{k/2-1})$, and the clause density at which instances become unsatisfiable with high probability, $m/n = \omega (1)$.

    In this paper, we give spectral and sum-of-squares algorithms for strongly refuting random $k$-XOR instances with clause density $m/n \ge \widetilde O(n^{(k/2-1)(1-\delta)})$ in time $\exp(\widetilde O(n^{\delta}))$ or in $\widetilde O(n^{\delta})$ rounds of the sum-of-squares hierarchy, for any $\delta \in [0,1)$ and any integer $k \ge 3$.
    Our algorithms provide a smooth transition between the clause density at which polynomial-time algorithms are known at $\delta = 0$, and brute-force refutation at the satisfiability threshold when $\delta = 1$.
We also leverage our $k$-XOR results to obtain strong refutation algorithms for SAT (or any other Boolean CSP) at similar clause densities.

Our algorithms match the known sum-of-squares lower bounds due to Grigoriev and Schonebeck, up to logarithmic factors.

\end{abstract}

\clearpage
\thispagestyle{empty}

\tableofcontents

\clearpage
\setcounter{page}{1}
\section{Introduction}

Random instances of constraint satisfaction problems (CSPs) have been a subject of intense study in computer science, mathematics and statistical physics.
Even if we restrict our attention to random $k$-SAT, there is already a vast body of work across various communities--see \cite{Achli09} for a survey.
In this paper, our focus is on {\it refuting} random CSPs: the task of algorithmically proving that a random instance of a CSP is unsatisfiable.
Refutation is a well-studied problem with connections to myriad areas of theoretical computer science including proof complexity \cite{BB02}, inapproximability \cite{Feige02}, SAT solvers, cryptography \cite{ABW10}, learning theory \cite{DLSS14}, statistical physics \cite{CLP02} and complexity theory \cite{BKS13}.

For the sake of concreteness, we will for a moment restrict our attention to $k$-SAT, the most well-studied random CSP.
In the random $k$-SAT model, we choose a $k$-uniform CNF formula $\Phi$ over $n$ variables by drawing $m$ clauses independently and uniformly at random.
The density of $\Phi$ is given by the ratio $\alpha = m/n$.
It is conjectured that for each $k$, there is a critical value $\alpha_k$ such that $\Phi$ is satisfiable with high probability if $\alpha < \alpha_k$, and unsatisfiable with high probability for $\alpha > \alpha_k$.
Such phase transition phenomena are conjectured to occur for all nontrivial random CSPs; for the specific case case of $k$-SAT, it was only recently rigorously established for all sufficiently large $k$ \cite{DSS15}.

In the unsatisfiable regime, when $\alpha > \alpha_k$, the natural algorithmic problem we associate with random $k$-SAT formulas is the problem of {\it refutation}.
We define the notion of a refutation algorithm formally:
\begin{definition} (Refutation Algorithm)
An algorithm $\cA$ is a {\it refutation} algorithm for random $k$-SAT at density $\alpha$, if given a random instance $\Phi$ of $k$-SAT with density $\alpha$, the algorithm $\cA$:
\begin{compactitem}

\item Outputs YES with probability at least $\frac{1}{2}$ over the choice of $\Phi$.\footnote{The choice of the fraction $\frac{1}{2}$ here is arbitrary, and one could potentially consider any fixed constant.}
\item Outputs NO if $\Phi$ is satisfiable.
\end{compactitem}
\end{definition}
\noindent Note that if the algorithm $\cA$ outputs YES on an instance $\Phi$, it certifies that the instance $\Phi$ is unsatisfiable.

Refuting random $k$-SAT is a seemingly intractable problem in that the best polynomial-time algorithms require density  $\alpha > \tilde{O}(n^{k/2 - 1}) \gg \tO(1)$.
We survey the prior work on refuting CSPs in \pref{sec:prior}.

At densities far exceeding the unsatisfiability threshold, i.e., $\alpha \gg \alpha_k$, a simple union bound argument can be used to show that a random instance $\Phi$ has no assignment satisfying more than a $1 - \frac{1}{2^k} + \delta(\alpha)$ fraction of constraints, where $\delta(\alpha) \to 0$ as $\alpha \to \infty$.  In this regime, a natural algorithmic task is {\it strong refutation}:

\begin{definition} (Strong Refutation)
An algorithm $\cA$ is a {\it strong refutation} algorithm for random $k$-SAT at density $\alpha$, if for a fixed constant $\delta > 0$, given a random instance $\Phi$ of $k$-SAT with density $\alpha$, the algorithm $\cA$:
\begin{compactitem}
\item Outputs YES with probability at least $\frac{1}{2}$ over the choice of $\Phi$.
\item Outputs NO if $\Phi$ has an assignment satisfying at least a $(1-\delta)$-fraction of clauses.
\end{compactitem}
\end{definition}

An important conjecture in complexity theory is Feige's ``R3SAT hypothesis,'' which states that for any $\delta>0$, there exists some constant $c$ such that there is no polynomial-time algorithm that can certify that a random $3$-SAT instance has value at most $1-\delta$ (that is, strongly refute $3$-SAT) at clause density $m/n = c$.
Feige exhibited hardness of approximation results based on the hypothesis for a class of otherwise elusive problems such as densest-$k$ subgraph and min-bisection \cite{Feige02}.
This hypothesis has subsequently been used as the starting point in a variety of reductions (see e.g. \cite{AAMMW11,BKS13,DLS13}).

The problem of strong refutation is non-trivial even for polynomial-time solvable CSPs such as $k$-XOR.\footnote{
    The {\it weak} refutation problem for $k$-XOR can be easily solved using Gaussian elimination.}
A random $k$-XOR instance $\Phi$ on $n$ variables $x_1,\ldots,x_n \in \{\pm 1\}$ consists of $m$ equations of the form $x_{i_1} \cdot x_{i_2} \cdots x_{i_k} = \pm 1$
By a simple union bound, one can show that at all super-linear densities $m/n = \omega(1)$, with high probability, no assignment satisfies more than $\frac{1}{2}+ o(1)$-fraction of the equations.\footnote{Random $k$-XOR can also be equivalently defined in terms of equations of the form $x_{i_1} \oplus \cdots x_{i_k} = 0/1$.  The equivalence follows by mapping $0 \to 1$, $1 \to -1$, and $\oplus \to \cdot $.}
The problem of strong refutation for random $k$-XOR amounts to certifying that no assignment satisfies more than $1-\delta$ fraction of equations for some constant $\delta > 0$.
A natural spectral algorithm can efficiently strongly refute $k$-XOR at densities $m/n \ge n^{k/2-1}$ \cite{COGL04,COGL07,AOW15,BM15}.
However, strong refutation at any lower density is widely believed to be an intractable problem \cite{ABW10, BM15, DLSS14, D16}.
We refer the reader to \cite{D16} for a survey of the evidence pointing to the intractability of the problem.

To expose the stark difficulty of strongly refuting random $k$-XOR, consider the easier task of distinguishing random $k$-XOR instances from those generated from the following distribution:
first, sample a satisfiable instance of $k$-XOR uniformly at random, by sampling a planted solution $z \in \{\pm 1\}^n$ and randomly choosing $m$ equations, each on $k$ variables, satisfied by $z$.
Then, corrupt each of the $m$ equations (so that $z$ does not satisfy it) with probability $\delta$.
Equivalently, this problem can be described as { \it learning parity with noise}, wherein $z \in \{\pm 1\}^n$ defines the unknown parity and each equation $C_i$ is an {\it example} to the learning algorithm.
An algorithm to learn parity from noisy examples can be used to distinguish the planted instances sampled as described above from uniformly random instances of $k$-XOR.
There is no known distinguishing algorithm at any density $m/n < n^{k/2-1}$, and the computational intractability of this problem has recently been used to obtain lower bounds for improper learning \cite{D16}.

\paragraph{Sum-of-Squares Refutations}
A natural proof system for strong refutation is the sum-of-squares (\sos) proof  system.
Given an instance $\Phi$ of a Boolean $k$-CSP,  the fraction of constraints satisfied by an assignment $x$ can be written as a polynomial $P_{\Phi}(x)$ of degree at most $k$ in $x$. Let $\opt(\Phi)$ denote the largest fraction of constraints satisfied by any assignment to the variables, i.e.,
\[
    \opt(\Phi) \defeq  \max_{x \in \{\pm 1\}^n} P_{\Phi}(x) \mper
\]
Therefore, certifying an upper bound $c$ on $\opt(\Phi)$ reduces to certifying that $\max_{x \in \{\pm 1\}^n} P_{\Phi}(x) < c$.

A degree-$d$ {\it sum-of-squares proof}  for this fact is a polynomial identity of the form,
\[
    c - P_{\Phi}(x) = \sum_{i} q_i^2(x) \mod \mathcal{I}\mcom
\]
where $\deg(q_i^2) \leq d$ and $\mathcal{I}$ is the ideal generated by
the polynomials $\{x_j^2 - 1\}$ that define the variety
$\{\pm 1\}^n$.

The size of a degree-$d$ \sos proof is at most $n^{O(d)}$, assuming the coefficients have a bit-complexity of at most $n^{O(d)}$.
Moreover, finding a degree-$d$ \sos proof can be formulated as a semidefinite program, also known as the degree-$d$ sum-of-squares hierarchy or the $d$-round Lasserre/Parrillo SDP hierarchy \cite{Las, Par}.
Therefore, if there is a degree-$d$ \sos proof with bit complexity $n^{O(d)}$, then one can be found in time $n^{O(d)}$.

\sos proof systems are very powerful in that they capture both local arguments, such as resolution-based proofs, and global methods like spectral techniques.
Furthermore, these proof systems subsume various linear programming and SDP hierarchies such as the Sherali-Adams, Lov\'{a}sz Schrijver (LS) and LS+ hierarchies.
In the recent past, the \sos SDP hierarchy has received considerable attention due to its ability to certify the objective value on many candidate hard instances for the unique games problem \cite{BBHKSZ12}.

Unfortunately, the lower bounds of Grigoriev \cite{Gri01} and Schoenebeck \cite{Sch08} rule out efficient strong \sos refutations for random $k$-XOR and random $k$-SAT at densities significantly smaller than $m/n < n^{k/2-1}$.
Specifically, Schonebeck's result implies that with high probability over $k$-XOR instances $\Phi$ with clause density $m/n < O(n^{(k/2 -1)(1-\delta)})$, the \sos hierarchy cannot refute $\Phi$ at degree $O(n^{\delta})$.

Note that this leaves open the possibility that random $k$-XOR and random $k$-SAT admit subexponential-sized strong refutations well-below the $n^{k/2-1}$ threshold.
This sets the stage for our main result.
\medskip

\begin{theorem} \label{thm:maindlin}
    For all $\delta \in [0,1)$ given a random $k$-XOR instance $\Phi$ on $n$ variables, with high probability over $\Phi$, the degree $O(n^{\delta})$ sum-of-squares hierarchy can strongly refute $\Phi$, certifying that
	\[
	    \opt(\Phi) \leq \frac{1}{2}+ \epsilon \mcom
	\]
	for any constant $\epsilon>0$ as long as $\Phi$ has clause density  $m/n \ge \tO (n^{(k/2-1)(1-\delta)})$, where the $\tO$ notation hides logarithmic factors and a dependence on $\epsilon$ and $k$.
	Further, there is a spectral algorithm achieving the same guarantees by computing the eigenvalue of an $2^{\tO(n^{\delta})} \times 2^{\tO(n^{\delta})}$ matrix.
\end{theorem}

\begin{remark}
The algorithm from \prettyref{thm:maindlin} yields {\it tight} refutations--certifying a tight upper bound of $\opt(\Phi) + \epsilon$ for any constant $\epsilon>0$.
\end{remark}
Notice that the result establishes a smooth trade-off between the clause density of $\Phi$ and the running time of the refutation algorithm.
Specifically for all $\delta \in [0,1)$, the algorithm strongly refutes at density $m/n = \tO(n^{(k/2 -1)(1-\delta)})$ in time $\exp(\tO(n^{\delta}))$, so that when $\delta = 0$ the result matches the performance of the best known polynomial-time algorithms, and at $\delta = 1$, the algorithm refutes instances just above the threshold of satisfiability in exponential time.
Moreover, the degree of the sum-of-squares refutations matches the degree lower bounds of \cite{Gri01,Sch08} up to polylogarithmic factors.

Feige \cite{Feige02} introduced a connection between the refutation of random XOR instances and the refutation of other CSPs, and this connection was later used in several other works (e.g. \cite{FKO06,AOW15,BM15}).
Using the machinery developed by Allen \etal \cite{AOW15}, we apply our algorithm for $k$-XOR to refute other random CSPs involving arbitrary Boolean predicates $P$; for example to $k$-SAT.
\begin{theorem}\label{thm:anycsp}
    Let $P:\{\pm 1\}^k \to \{0,1\}$ be a predicate with expected value $\E[P]$ over a random assignment in $\{\pm 1\}^k$.
For all $\delta  \in (0,1]$, given an instance $\Phi$ of a random $k$-CSP with predicate $P$ on $n$ variables, the degree $O(n^{\delta})$ \sos hierarchy strongly refutes $\Phi$ with high probability, certifying that
\[
    \opt(\Phi) \le \E[P] + \epsilon\mcom
\]
for any constant $\epsilon > 0$ so long as $\Phi$ has density at least $m/n \ge \tO (n^{(k/2 - 1)(1 - \delta)})$, where the $\tO$ hides a dependence on a polylog factor, $k$ and $\epsilon$.
Further, there is a spectral algorithm achieving the same guarantees.
\end{theorem}

We can extend \pref{thm:anycsp} so that the density/runtime trade-off depends on the \emph{independence parameter} of the predicate $P$ as defined by \cite{AOW15}--we defer the details to \pref{sec:sat}.

\subsection*{Injective tensor norm}

The proof techniques we develop are applicable beyond strongly refuting random $k$-XOR, to the problem of certifying upper bounds on the injective tensor norm of random tensors.

The injective tensor norm generalizes the matrix operator norm, in the following sense.
For an order-$k$ symmetric tensor with all dimensions equal to $n$, the injective tensor norm is defined as
\[
    \|\bT\|_{inj} \defeq \max_{\substack{x\in \R^n\\ \|x\| = 1}} \left| \Iprod{\bT, x^{\tensor k}}\right|,
\]
where by $x^{\tensor k}$ we mean the symmetric rank-1 tensor of order $k$ given by tensoring $x$ with itself, and by the inner product we mean the entry-wise sum of the products of the entries of $\bT$ and $x^{\tensor k}$, as is standard.

When $k=2$, computing $\|\bT\|_{inj}$ is equivalent to computing the matrix operator norm.
Yet when $k \ge 3$, the injective tensor norm is hard to compute.
The hardness of approximating the injective tensor norm is not fully understood, but we do know that, assuming the exponential-time hypothesis, the injective tensor norm requires quasipolynomial time to approximate, even within super-constant factors \cite{BBHKSZ12}.
There are also reductions to the problem from a variety of problems such as Planted Clique \cite{BV09} and Small-Set Expansion \cite{BBHKSZ12}.

The problem is nontrivial even when the tensor has i.i.d. random entries.
It is well-known that the norm of a tensor with i.i.d. symmetric subgaussian entries is of the same order as the norm of a random matrix:
\begin{theorem}[\cite{tomioka2014}]
    If $k\in \N$ is constant and $\bT$ is a symmetric order-$k$ tensor of dimension $n$ with i.i.d. symmetrically distributed subgaussian entries, then with probability at least $1-o(1)$,
  $
       \|\bT\|_{inj} \le \tO(\sqrt{n}).
   $
\end{theorem}
\noindent So the question arises naturally: is it easy to certify tensor norm bounds under distributional assumptions on the entries?
The current known polynomial-time algorithms fall short of the bound $\tO(\sqrt{n})$, and can only certify bounds of $\|\bT\|_{inj}\le \tO(n^{k/4})$ for tensors of order $k$ \cite{RM14,HSS15,HSSS15}.
The algorithm of Hopkins et al. \cite{HSS15} is based on the degree-$k$ \sos relaxation for the tensor norm problem.
They also give a lower bound for the \sos relaxation for the order-$3$ tensor at degree $4$, proving that the relaxation has value $\tOm(n^{3/4})$, which implies that their analysis is tight for the \sos hierarchy at degree $4$.

By applying our techniques for random $k$-XOR refutations to the problem of certifying bounds on tensor norms, we have the following result:

\begin{theorem}\label{thm:inj-tensor-informal}
    For any $\delta \in [0,1/120)$, given a symmetric order-$k$ tensor $\bT$ with i.i.d. standard Gaussian entries, with high probability over the choice of $\bT$, the degree $O(n^{\delta})$ \sos hierarchy relaxation certifies that
    \[
	\|\bT\|_{inj} \le \tO( n^{1/2 + (k-2)(1-\delta)/4  + 3k^2\delta^2}) \mcom
    \]
    where the $\tO$ notation hides a polylogarithmic factor and a dependence on $k$.
    Furthermore, there is a spectral algorithm that computes the eigenvalues of a $2^{\tO(n^{\delta})}\times 2^{\tO(n^{\delta})}$ matrix that certifies the same bound.
\end{theorem}

We remark that the above theorem also holds, up to constants, for symmetric tensors with i.i.d. entries from any symmetric distribution $\cD$ over $\R$ with subgaussian tails.
Strong refutation for $k$-XOR instances can be thought of as a special case of the problem of bounding the norm of a random tensor--we elaborate on the connection at the start of \pref{sec:tech-over}.
However, the underlying distribution for random $k$-XOR yields tensors which are extremely {\em sparse}, which poses several additional technical challenges.

In an independent work, Bhattiprolu et al. \cite{BGL16} have obtained a result similar to \pref{thm:inj-tensor-informal} for low-order tensors (for order $k = 3,4$); at $k=3,4$, they certify tighter bounds.
They also obtain a tight lower bound on the integrality gap of degree-$k$ \sos relaxations for $k$-tensor norms.

\paragraph{Techniques}
We give an overview of our techniques in \pref{sec:tech-over}, but here we give a brief teaser.
For both the $k$-XOR problem and the injective tensor norm problem, we want to bound the norm of a random tensor with independent entries (under different distributional assumptions, and only for Boolean vectors in the case of CSPs).
A canonical upper bound for injective tensor norm is to take the maximum eigenvalue of its natural matrix flattening.
This bound is loose, because the top eigenvector of the matrix flattening is not restricted to be a rank-$1$ Kronecker power of a vector in $\R^n$--in a sense, the top eigenvector of the matrix flattening does not have enough symmetry.
Our algorithmic strategy is to exploit this lack of symmetry.

The sum-of-squares SDP at degree $d$ provides a linear map $\pE:\R[x]^{\le d} \to \R$, from the space of polynomials of degree $\le d$ over $x \in \R^{n}$ to the reals, with the property that $\pE[p^2(x)] \ge 0$ for any polynomial $p$ of degree $\le d/2$.
Each monomial $\prod_{i\in S} x_i$ is identified with an SDP variable $X_{S} \defeq \pE[\prod_{i\in S}x_i]$.

At a high level, the sum-of-squares SDP relaxations for CSPs contain two classes of constraints.
First, the \emph{symmetry constraints}, such as the constraint that the SDP variable that we identify with the monomial $X_{ijk} = \pE[x_i x_j x_k]$ be equal to the variable identified with the variable corresponding to the permuted monomial $X_{jki} = \pE[x_j x_k x_i]$.
Second, the \emph{Booleanness} constraints, which enforce $\pE[x_i^2] = 1$ for each $i$.

Roughly speaking, our spectral algorithm relies mainly on the symmetry constraints in the SDP.
In \pref{sec:tech-over}, we describe how to harness these symmetry constraints to obtain better approximations with increasing \sos degree (by showing how the symmetry constraints can be used to improve a spectral algorithm).
This technique adds to the arsenal of tools for algorithm design via the \sos SDP hierarchy, and is the main technical contribution of this work.

\subsection{Related work}
\label{sec:prior}
We briefly survey the prior work on refuting random CSPs--we refer the reader to \cite{AOW15} for a thorough survey on the topic.
Work on refuting random CSPs began with Chv\'{a}tal and Szemer\'{e}di  \cite{CS88}, who showed that a random $k$-SAT instance with clause density $\alpha > c$ (for $c$ constant) with high probability requires Resolution refutations of exponential size.
This lower bound was later complemented by the works of \cite{Fu96,BKPS98}, which show that at clause density $\alpha \ge O(n^{k-1})$, polynomial-sized resolution proofs exist and can be found efficiently.
At the turn of the century, Goerdt and Krivelevich \cite{GK01} pioneered the spectral approach to refuting CSPs, showing that a natural spectral algorithm gives refutations for $k$-SAT in polynomial time when $\alpha = m/n \ge n^{\lceil k/2\rceil - 1}$.
A series of improvements followed, first achieving bounds for $\alpha \ge O(n^{1/2+\epsilon})$ for any constant $\epsilon$ for the special case of $3$-SAT \cite{FG01,FGK05}, then achieving strong refutation at densities $\alpha \ge \tO(n^{\lceil k/2\rceil-1})$ \cite{COGL04,COCF10}.
Finally, the works of Allen et al. and Barak and Moitra gave spectral algorithms for strongly refuting $k$-XOR and $k$-SAT for any $\alpha \ge \tO(n^{k/2-1})$ \cite{AOW15,BM15}, and Allen et al. also give a reduction from any CSP which is far from supporting a $t$-wise independent distribution to $t$-XOR.
These spectral algorithms are the algorithmic frontier for efficient refutations of random CSPs.

Though not algorithmic, the work of \cite{FKO06} is worth mentioning as well.
Feige et al. show that, at clause density $\alpha = m/n \ge \tO(n^{0.4})$, there exists a polynomial-sized (weak) refutation for random $3$-SAT given by a subset of $O(n^{0.2})$ unsatisfiable clauses.
Understanding whether polynomial-sized weak refutations exist for smaller $\alpha$ is an intriguing open problem.

In a concurrent and independent work, Bhattiprolu et al. \cite{BGL16} obtained a result similar to \pref{thm:inj-tensor-informal} for bounding the norms of tensors of order $3$ and $4$.
The bounds obtained in \cite{BGL16} are tighter for $k = 3$ and $k = 4$.
The results of \cite{BGL16} do not imply new results for refutation even for  CSPs of arity $3$ and $4$, since their upper bound is too weak on sparse tensors--a regime that poses additional technical hurdles.

\subsection{Organization}
In \pref{sec:tech-over}, we illustrate our core ideas via a detailed exposition of our proof for certifying bounds on the norm of order-$4$ tensors, and explain how these techniques can be built upon to strongly refute CSPs.
\pref{sec:randombds} contains the full proof our tensor norm results.
\pref{sec:lin} contains our results for refuting $k$-XOR instances.
In \pref{sec:sat}, we combine our $k$-XOR refutation algorithms with the framework of \cite{AOW15} to refute other CSPs.
Finally, in \pref{sec:sos} we argue that our spectral algorithms give \sos proofs.

\subsection{Preliminaries}
We represent tensors by boldface letters such as $\bT$.
We refer to the map from a tensor $\bT$ with entries indexed by $[n]^{\tensor 2k}$ to a matrix $T$ indexed by $[n]^{k} \times [n]^{k}$ as the ``natural flattening'' of $\bT$.
For a matrix or vector $M \in \R^{n\times m}$, the notation $M^{\tensor d}$ refers both to the $n^d \times m^d$ $d$-wise Kronecker power of $M$, or to the $n \times \cdots \times m$ tensor given by the $d$-wise cross-product of $M$ with itself.
For matrices, tensors, or vectors $A,B$ whose entries are identified with the same space, $\iprod{A,B}$ denotes the sum of the entrywise products of $A$ and $B$.

\section{Main Ideas: Proof for Random $4$-Tensors}\label{sec:tech-over}

In this section, we will survey the main ideas in our paper by proving \pref{thm:inj-tensor-informal} (our tensor norm certification algorithm) for the case of random $4$-tensors.
This specific case yields the simplest proof, while encapsulating the core ideas of our techniques for both injective tensor norm and $k$-XOR.
We formally state the injective tensor norm problem here.

\begin{problem}[Certifying injective tensor norm]\label{prob:tnorm}
Given an order-$k$ tensor $\bT$ with dimension $n$, certify that for
all $x\in \R^n$ with $\|x\|=1$, $|\iprod{\bT,x^{\tensor k}}| \le
\|\bT\|_{inj}\le \tau$ for some upper bound $\tau$.
\end{problem}

\paragraph{From $k$-XOR to tensor norms}
First, we briefly outline the connection between $k$-XOR refutation and certifying bounds on tensor norms.
Let $\Phi$ be a random $k$-XOR formula on $x \in \{\pm 1\}^n$  with $m
\approx pn^k$ clauses, sampled as follows: for each $S \subset [n]^k$
independently with probability $p$, add the constraint that
$\prod_{i\in S} x_i = \eta_S$ where $\eta_S$ is a uniform bit $\pm 1$,
and with probability $1-p$, add no constraint.
We can form an order-$k$ tensor $\bT$ so that for each $S \in [n]^k$, $\bT_S = 0$ if there is no constraint, and otherwise $\bT_S = \eta_S$.

For any assignment $x \in \{\pm 1\}^n$, the inner product $\langle
\bT, x^{\otimes k}\rangle $ is equal to the difference in the number
of $\Phi$'s constraints that $x$ does and does not satisfy.  Since
$\Phi$ has $m$ constraints in all, certifying that $\max_{x \in \sbits^n}
|\iprod{\bT,x^{\otimes k}}| \leq o(m)$ is equivalent to certifying that $\opt(\Phi)
\le \tfrac{1}{2}+o(1)$.
On the other hand, certifying the injective tensor norm amounts to
exhibiting an upper bound on $\max_{\norm{y} \leq 1} |\iprod{\bT,y^{\otimes k}}|$ where the maximization is over all unit vectors $y$.
Every Boolean vector $x \in \sbits^n$ is of length $\norm{x} =
\sqrt{n}$, which implies that $\max_{x \in \sbits^n}
|\iprod{\bT,x^{\otimes k}}| \leq n^{k/2} \cdot \norm{\bT}_{inj}$.

While the above reduction from certifying that $\opt(\Phi) \leq \frac{1}{2} +
o(1)$ to certifying a bound on $\norm{\bT}_{inj}$ exposes the
connection between the two problems, it is too lossy to be
useful.
In fact, for $p < 1/n^{k/2}$, the sparsity of the tensor $\bT$ implies
that the form $\iprod{y^{\otimes k},\bT}$ is maximized by sparse real-valued
vectors $y$ that are completely unlike Boolean vectors.  In other
words, almost surely there exists a sparse $y \in \R^n$, $\|y\| =
\sqrt{n}$ with $|\iprod{\bT, y^{\otimes k}}| \gg \max_{x \in \sbits^n} |\iprod{\bT, x^{\otimes k}}|$.
As a result, our refutation algorithm for $k$-XOR is more involved
than the certification algorithm for injective tensor norm in two ways.
First, it crucially uses the {\it non-sparseness} of Boolean
vectors and second, the sparsity of the tensor $\bT$ calls for more nuanced concentration arguments.
We give a short overview of these differences in \pref{sec:xordet} after presenting the broad strokes of the proof, via our algorithm for certifying tensor norms, and full details in \pref{sec:lin}.

\paragraph{Certifying injective tensor norm}
In what follows, we will give a \emph{spectral algorithm} for \pref{prob:tnorm} for random $4$-tensors with i.i.d. subgaussian entries.
We will show that this spectral algorithm is subsumed by a \sos relaxation of appropriate degree in \pref{sec:sos}.
The rest of this section is organized as follows.
\begin{compactenum}
\item We first {\bf describe the matrix} whose maximum eigenvalue provides the upper bound on the injective tensor norm.
Rather than writing down the matrix immediately, we will build up our intuition by first considering a simple spectral approach, and then seeing how we can improve.
\item We then {\bf obtain bounds on the eigenvalues} of the matrix, which will hold with high probability for tensors with i.i.d. subgaussian entries--this is the step in which we analyze the performance of our algorithm.
    Because our matrix is somewhat complicated and not amenable to the application of black-box matrix concentration inequalities, we will apply the {\em trace power method}.
    This amounts to bounding the expected trace of a large power of
    our matrix, a goal which we split in to two steps.
    \begin{compactenum}
    \item First, we {\em reduce computing the expected trace to a hypergraph counting problem}.
    \item Then, we simplify the counting by {\em analyzing a particular hypergraph sampling process}.
    \end{compactenum}
\end{compactenum}

\subsection{Improving on the natural spectral algorithm with higher-order symmetries}
A natural spectral algorithm for \pref{prob:tnorm} is to flatten the tensor to a matrix, and then compute the operator norm of the matrix.
This is a valid relaxation because, given an order-$4$ tensor $\bA$ with symmetric i.i.d. standard normal entries, if we take $A$ to be the natural $n^2 \times n^2$ matrix flattening of $\bA$,
\begin{equation}
 \|\bA\|_{inj}
    = \max_{\substack{x \in \R^n : \|x\|=1}} \left| (x \tensor x)^\top A (x \tensor x) \right|
    \le \max_{\substack{y \in \R^{n^2}:\|y\|=1} } \left|y^\top A y\right|
    =\|A\|_{op}\mper\label{eq:relax}
\end{equation}
So $\|A\|_{op} $ gives a valid upper bound for $ \|\bA\|_{inj}$.
This is great--on the left, we have a program that we cannot efficiently optimize, and on the right we have a relaxation which we can compute in polynomial time.

On the other hand this bound is quite loose--classical results from random matrix theory assert that with high probability, $\|A\|_{op} = \widetilde\Theta(n)$ whereas with high probability $\|\bA\|_{inj} \le O(\sqrt{n})$.
The issue is that the relaxation in \pref{eq:relax} is too lenient--the large eigenvalues of $A$ correspond to eigenvectors $y \in \R^{n^2}$, that are far from vectors of the form $x \tensor x : x\in \R^n$.
We want to decrease the spectrum of $A$ along these asymmetric {\it non-tensor product} directions.

A tensored vector of the form $x \otimes x$ satisfies the symmetry that $(x \otimes x)_{ij} = (x
\otimes x)_{ji} = x_ix_j$.  Therefore, a natural approach to decrease
the spectrum of $A$ along the {\it non-tensor product} directions is
to average the matrix $A$, along these symmetries.  Specifically, for
each $(i,j)$, we would average the $ij^{th}$ and $ji^{th}$ rows, and
then repeat the same operation on columns.  Formally, the averaged
matrix $A'$ is given by,
\[
	A' = \E_{\Sigma, \Pi \in \hat\cS_{2}}\left[\Sigma A\Pi\right]
\]
where $\hat\cS_{2}$ is the set of matrices which perform the permutations corresponding to the symmetric group on $2$ elements on the rows and columns of matrices indexed by $[n]^2$.
Unfortunately, for a symmetric $4$-tensor $\bA$, the matrix $A$ is
also symmetric with respect to these operations, so that $A' = A$.

To better exploit the symmetries of tensored vectors $x \otimes x$, we
will work with higher powers of the injective tensor norm.
For any $d \in \N$, we can write the $d^{th}$-power of
$\norm{\bA}_{inj}$ as
\begin{align}
    \|\bA\|_{inj}^d
    & = \max_{x \in \R^n,\|x\|=1} \left|\langle x^{\otimes 4}, \bA \rangle^{d}\right|
    = \max_{x \in \R^{n}, \|x\|=1} \left|(x^{\otimes 2d})^{\top} A^{\otimes d}x^{\otimes 2d} \right|,\nonumber
    \intertext{where $A^{\otimes d}$ is the natural $n^{2d} \times
    n^{2d}$ matrix flattening of $\bA^{\otimes d}$.
    The symmetric vector $x^{\tensor 2d}$ is fixed by averaging over any permutation of the indices, so averaging over such permutations does not change the maximum:}
    &= \max_{x \in \R^{n}, \|x\|=1} \left|\E_{\Pi,\Sigma\in\hat\cS_{2d}}\left[(\Pi x^{\otimes 2d})^\top A^{\otimes d} (\Sigma x^{\otimes 2d})\right]\right|\mcom\nonumber
    \intertext{and by linearity of expectation,}
    &= \max_{x \in \R^{n}, \|x\|=1} \left|(x^{\otimes 2d})^\top\left(\E_{\Pi,\Sigma\in\hat\cS_{2d}}\left[ \Pi^{\top} A^{\otimes d} \Sigma \right]\right)x^{\otimes 2d}\right|
     \quad \le~~ \left\|\E_{\Pi,\Sigma\in\hat\cS_{2d}}\left[ \Pi^{\top} A^{\otimes d} \Sigma \right]\right\|_{op}\mper \label{eq:properrelax}
\end{align}
The operator norm of the above described matrix will certify our upper bounds:
\begin{proposition} \label{prop:properrelax}
Let $k \in \N$ be even.
    Let $\hat \cS_{kd/2}$ be the set of matrices performing the permutations of $\cS_{kd/2}$ on matrices with rows and columns indexed by $[n]^{kd/2}$.
    For any order-$k$ tensor $\bA$ with matrix flattening $A$,
    \[
	\|\bA\|_{inj}
	\le \left(\left\|\E_{\Pi,\Sigma\in\hat\cS_{kd/2}}\left[ \Pi A^{\otimes d} \Sigma \right]\right\|_{op}\right)^{1/d}\mper
	\]
\end{proposition}
\begin{proof}
    The sequence of calculations culminating in \pref{eq:properrelax} gives the proof.
\end{proof}
\noindent Now, how can this give an improved upper bound over $\|A^{\otimes d}\|_{op} = \|A\|_{op}^d$?
The reason is that although $\bA$ had $4$-wise symmetry, the tensor $\bA^{\otimes d}$ {\em does not have} $4d$-wise symmetry.
For $I,J \in [n]^{2d}$, $I = (i_1,i'_1),\ldots, (i_d, i'_{d})$ and $J = (j_1,j'_1),\ldots,(j_d,j'_d)$ and for permutations $\pi,\sigma$ on $2d$ elements,
\[
(A^{\tensor d})_{I,J}
~=~
\prod^d_{\substack{\ell = 1\\(i_\ell,i'_{\ell}) \in I\\(j_\ell,j'_{\ell}) \in J}} A_{i_\ell,i'_{\ell},j_{\ell},j'_{\ell}}
~\bm{\neq}~
\prod^d_{\substack{\ell = 1\\(a_\ell,a'_\ell) \in \pi(I)\\(b_\ell,b'_\ell) \in \sigma(J)}} A_{a_\ell,a'_\ell,b_\ell,b'_\ell}
~=~ (A^{\tensor d})_{\pi(I),\sigma(J)}\mcom
\]
because the identity of the base variables in the expression may change under the permutation of the indices $I$ and $J$.
Thus, the typical entry of $\E_{\Pi,\Sigma \in \hat\cS_{2d}} \left[ \Pi(A^{\tensor d})\Sigma\right]$ is an average of $(d/2!)^2$ random variables, which are not independent, but also not identical.
Since the entries of $A$ are distributed symmetrically about zero, we expect the magnitude of the typical entry to drop after this averaging.
If we indulge the heuristic assumption that the entries of $\E_{\Pi,\Sigma \in \hat\cS_{2d}} \left[ \Pi(A^{\tensor d})\Sigma\right]$ are averages of $d^{\Omega (d)}$ {\em independent} random symmetric variables of constant variance, then the magnitude of the typical entry should be $\approx \frac{1}{d^{\Omega (d)}}$.
So heuristically, we have that
\[
    \left\|\E_{\Pi,\Sigma \in \hat\cS_{2d}} \left[ \Pi(A^{\tensor d})\Sigma\right]\right\|_F \le \frac{1}{d^{\Omega (d)}}\cdot\|A^{\otimes d}\|_F.
\]
By Wigner's semicircle law, matrices with independent entries have eigenvalues that are all roughly of the same magnitude.
Because our matrix has roughly independent entries, we may hope that the semicircle law holds for us, so that from the above heuristic calculations and from \pref{eq:properrelax},
\begin{equation*}
    \|\bA\|_{inj}
    \le \left(\left\|\E_{\Pi,\Sigma \in \hat\cS_{2d}} \left[ \Pi(A^{\tensor d})\Sigma\right]\right\|_{op}\right)^{1/d}
    \le \left(\frac{1}{d^{\Omega (d)}}\cdot\|A^{\otimes d}\|_{op}\right)^{1/d}
    \le \frac{n}{d^{\Omega(1)}}\mper
\end{equation*}
Thus, we expect that as we increase $d$, and therefore increase the symmetry of the {\it tensored vectors} $x \otimes x$ relative to the ``noisy'' non-tensor product eigenvectors of $A$, we can certify a tighter upper bound on $\|\bA\|_{inj}$.
Of course, since our certificate is the eigenvalue of a $n^{2d} \times n^{2d}$ matrix, the running time the refutation algorithm grows exponentially in the choice of $d$.

\subsection{Matrix concentration for the certificate}
Our algorithm is now clear: we form our matrix certificate by averaging overrows and columns corresponding to permutations of row and column indices in $A^{\otimes d}$, then use the certificate matrix's eigenvalues to upper bound $\|\bA\|^d_{inj}$ (by \pref{prop:properrelax}).
\begin{theorem}\label{thm:4-upper}
    Let $n,d\in \N$.
Let $\bA$ be an order-$4$ tensor with independent entries, distributed according to subgaussian distribution symmetric about $0$.
Then if $d\log n \ll n$, with high probability over $\bA$,
    \[
	\left\| \E_{\Pi,\Sigma}\left[\Pi A^{\tensor d} \Sigma\right]\right\|^{1/d}
	\le \tO\left(\frac{n}{d^{1/2}}d^{12\frac{\log d}{\log n}}\right)
    \mper
    \]
\end{theorem}
\noindent As a corollary of \pref{thm:4-upper} and \pref{prop:properrelax}, we get \pref{thm:inj-tensor-informal} for the case of order-$4$ tensors.

At the end of the previous subsection we gave a heuristic argument that a statement along the lines of \pref{thm:4-upper} should be true.
While the heuristic argument is plausible, it is very far from a formal proof; we need to prove that the eigenvalues of $\E[\Pi A^{\tensor d} \Sigma]$ are bounded by $\approx \tO(n/\sqrt{d})^d$  with high probability.
But the matrix $\E[\Pi A^{\tensor d} \Sigma]$ is not a sum of independent random matrices, and it does not have independent entries, so sophisticated matrix concentration tools (like the semicircle law or matrix Chernoff bounds) do not apply.
For tasks of this sort, the trace power method, or the method of moments, is the tool of choice:
\begin{proposition}[Trace power method]
    \torestate{
	\label{prop:tpm}
	Let $n,\ell \in \N$, let $c \in \R$, and let $M$ be an $n \times n$ random matrix.
	Then
	\[
	    \E_M[\Tr((MM^\top)^\ell)] \le \beta \implies
	    \Pr\Paren{\|M\| \ge c \cdot \beta^{1/2\ell} } \ge 1-c^{-2\ell} \mper
\]
}
\end{proposition}
\noindent The proof is essentially an application of Markov's inequality; we give it in \pref{app:useful}.

\subsubsection{From bounding the expected trace to a hypergraph counting problem}
A classic way to apply the trace power method is to reduce to a graph counting problem.
For example, let $M$ be a symmetric $n \times n$ random matrix with independent Rademacher entries.
We can view the row/column index set $[n]$ as a set of ``vertices,'' and the entry $M_{i,j}$ as an ``edge'' variable between vertices $i$ and $j$.
The trace $\Tr(M^{\ell})$ is the sum over products of edge variables along closed walks of length $\ell$ in the graph defined by $M$.
When we take $\E_{M}[\Tr(M^{\ell})]$, any closed walk in which an edge appears with odd multiplicity does not contribute to the sum, since $\E[M_{i,j}^m] = 0$ for odd $m$.
Therefore, $\E[\Tr(M^{\ell})]$ is equal to the number of closed walks
of length $\ell$ in which every edge appears with even multiplicity,
within the complete graph $K_n$, and bounding $\E[\Tr(M^\ell)]$ becomes a counting problem.

We make a similar reduction for our matrix $C \defeq \E_{\Sigma,\Pi}[\Pi A^{\otimes d} \Sigma]$.
The rows and columns of $A$ are indexed by pairs $[n]^2$, we interpret each variable $A_{ij,k\ell}$ as a (multi)hyperedge between the vertices $(i,j)$ and $(k,\ell)$ corresponding to the row and column indices respectively.
In the Kronecker power $A^{\otimes d}$, the rows and columns are indexed by vertex multisets $I,J \in [n]^{2d}$, $I = (i_1,i_1',\ldots,i_d,i_{d}'),\ J = (j_1,j_1',\ldots,j_{d},j_d')$, and the entry $(A^{\otimes d})_{I,J}$ is the product of the hyperedges $\prod_{k = 1}^d A_{i_k i_k',j_k j_k'}$.
We view this as a hyperedge matching between $I,J$, in which the vertices $(i_{k},i_{k}')$ are matched with the vertices $(j_{k},j_{k}')$ for each $k \in [d]$ (see \pref{fig:hypic}).

Now to obtain our matrix $C$, we average over row and column symmetries, so that $C_{I,J} = \E_{\pi,\sigma \in \cS_{2d}} [(A^{\otimes d})_{\pi(I),\sigma(J)}]$.
In each entry of $C$, we average over the permutations of the left and right vertex sets, which is the same as averaging over all perfect hypergraph matchings from $I$ to $J$ (again see \pref{fig:hypic}).

\begin{figure}[t!]
    \center
    \includegraphics[width=0.9\textwidth]{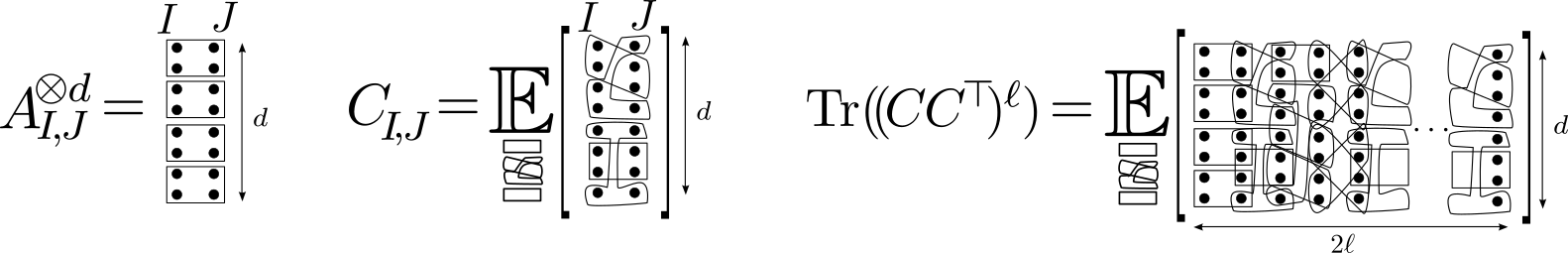}
    \caption{Hypergraph interpretations of the entries of $A^{\otimes d}$, $C_d$, and $\Tr((CC^{\top})^{\ell})$.}\label{fig:hypic}
\end{figure}

Just as in the case of the simple random matrix $M$, we can interpret $\Tr((CC^\top)^\ell)$ as the sum over all closed walks of length $2\ell$ on the complete graph (with self-loops) on the vertex set $[n]^{2d}$, where the edge variable between $I,J$ is the {\em average} over all possible hyperedge matchings between $I$ and $J$.
When we take the expectation over $\bA$, $\E_{\bA}[\Tr((CC^{\top})^\ell)]$, any {\em hyperedge} appearing with odd multiplicity will cause the contribution of the closed walk to be $0$, since the entries of $\bA$ are distributed symmetrically about $0$.

Our reduction is now complete.
Because we will be dealing with subgaussian random variables, the entries of $\bA$ will concentrate well enough for us to reduce to the Rademacher case.
\begin{lemma}\label{lem:red}
    Let $\bA$ be an order-$4$ tensor with i.i.d. Rademacher entries, and let $A$ be its matrix flattening.
    Let $C_d \defeq \E_{\Sigma,\Pi \in \hat \cS_{2d}}[\Pi A^{\otimes d} \Sigma]$.
    For the $2\ell$ multisets of vertices $I_1,\ldots,I_{2\ell} \in [n]^{2d}$, let $\cH$ be the set of all sequences of perfect hyperedge matchings between each $I_j$ and $I_{j+1 \mod 2\ell}$, so that each hyperedge has 2 vertices from $I_j$ and 2 vertices from $I_{j+1}$.
    For a fixed sequence of hyperedge matchings $H \in \cH$, let
    $\cE_{I_1,\ldots,I_{2\ell}}(H \text{ even})$ be the event that
    every hyperedge appears with even multiplicity.
    Then
    \[
	\E_{\bA}\left[\Tr\left((C_dC_d^{\top})^{\ell}\right)\right]
	= \sum_{I_1,\ldots,I_{2\ell} \in [n]^{2d}} \Pr_{H\sim\cH}\left[ \cE_{I_1,\ldots,I_{2\ell}}(H \text{ even}) \right]
    \]
\end{lemma}
\begin{proof}
    Any product of Rademacher random variables has expectation $0$ if some variable appears with odd multiplicity, and $1$ otherwise.
This, along with the observations preceding the lemma statement, implies that each $I_1,\ldots,I_{2\ell}$ contributes exactly the probability that hyperedges chosen for it all have even multiplicity (where we get a probability since each entry $C_{I,J}$ is the average over hyperedge matchings from $I$ to $J$).
\end{proof}

\subsubsection{Bounding the probability of an even hypergraph}

From \pref{lem:red} and \pref{prop:tpm}, in order to prove \pref{thm:4-upper} it suffices for us to bound
    \begin{equation}\label{eq:prsum}
	\sum_{I_1,\ldots,I_{2\ell} \in [n]^{2d}} \Pr_{H\sim\cH}\left[ \cE_{I_1,\ldots,I_{2\ell}}(H \text{ even}) \right] \le \tO(n/d^{1/2})^{2d\ell},
    \end{equation}
for $\ell = \Omega(\log n)$.
Since each probability is bounded by $1$ and there are $n^{4d\ell}$ terms in the sum, \pref{eq:prsum} easily gives us an upper bound of $n^{4d\ell}$.
We need to improve upon this naive bound twofold: first, we need the dependence on $n$ to be $n^{2d\ell}$.
This would give a bound of $\|\E[\Sigma A^{\otimes d} \Pi]\| \le \tO(n^{d})$ w.h.p., but we can get this bound trivially by ignoring the symmetrization, as $\|A^{\otimes d}\| \le \tO (n^d)$ w.h.p.
To fully reap the rewards of symmetrization, we must improve by a factor of $\approx (\sqrt{d})^{-2d\ell}$.

At first, bounding \pref{eq:prsum} seems daunting--it is unclear how to count the number of such hypergraphs with even multiplicity, while simultaneously getting the correct dependence on $n$ and $d$.
It will be helpful to use the following two-step process for sampling hypergraphs: for a fixed vertex configuration $I_1,\ldots,I_{2\ell} \in [n]^{2d}$,
\begin{compactenum}
\item First, sample perfect simple edge matchings between $I_j,I_{j+1}$ for each $j \in [2\ell]$.\label{step:medge}
\item Next, pair up the edges between $I_j, I_{j+1}$ and merge each pair to form a hyperedge.\label{step:hedge}
\end{compactenum}
We will use \pref{step:medge} to bound the dependence on $n$, and \pref{step:hedge} to bound the dependence on $d$.
In particular, our arguments from \pref{lem:red} give us the following lemma almost immediately:
\begin{lemma} \label{lem:simpleeven}
    Let $\cM$ be the set of all possible choices of edge sets sampled in \pref{step:medge}.
    Let $\cE_{I_1,\ldots,I_{2\ell}}(E \text{ even})$ be the event that the graph given by the edges $E \in \cM$ on $I_1,\ldots,I_{2\ell}$ has every edge appearing with even multiplicity. Then
    \begin{equation}
	\sum_{I_1,\ldots,I_{2\ell} \in [n]^{2d}} \Pr_{E \sim
	\cM}\left[\cE_{I_1,\ldots,I_{2\ell}}(E \text{ even})\right] =
	\E_{M}\left[\Tr\left((BB^{\top})^{\ell}\right)\right]
	\mcom\label{eq:simplesum}
    \end{equation}
	where $M$ is an $n \times n$ matrix with i.i.d. Rademacher entries, and $B \defeq \E_{\Pi,\Sigma \in \hat\cS_{2d}}[\Pi M^{\otimes 2d}\Sigma]$.
\end{lemma}

\pref{lem:simpleeven} lets us relate the probability that we sample a
perfect matching in which every edge appears with even multiplicity in
\pref{step:medge} to the norm of a matrix $M$ with i.i.d. Rademacher entries, which is an object we understand well: with very high probability, $\|M\| \le O(\sqrt{n})$, and because of the connection between the expected trace and the norm of a matrix, we can then bound \pref{eq:simplesum} by the desired $\tO(n^{1/2})^{4d\ell}$.

To use \pref{eq:simplesum}, we need to relate the probability that the edges sampled in \pref{step:medge} have even multiplicity to the probability that the hyperedges sampled in \pref{step:hedge} have even multiplicity.

\begin{lemma}[somewhat informal statement]\label{lem:graph-from-hgraph-nocount}
    Let $I_1,\ldots,I_{2\ell} \in [n]^{2d}$, and suppose we have sampled hyperedges $H \in \cH$ by first sampling simple edges $E \in \cM$ as in \pref{step:medge} and then grouping them as in $\pref{step:hedge}$.
Then
\[
    \Pr(\cE_{I_1,\ldots,I_{2\ell}}(E \text{ even})~|~ \cE_{I_1,\ldots,I_{2\ell}}(H \text{ even})) \ge \left(\frac{1}{2}\right)^{2d\ell}\mper
\]
\end{lemma}
\begin{proof}[Proof (sketch).] For any given hyperedge $(i,j,k,\ell) \in H$, with $i,j \in I_a$ and $k,\ell \in I_{a+1}$, there are only two ways it could have been sampled as pairs of edges, either as a merge of $(i,k),(j,\ell) \in E$ or of $(i,\ell),(j,k) \in E$.
If all copies of a hyperedge of even multiplicity $m$ are sampled the same way, then the corresponding edges also have even multiplicity.\footnote{In the formal proof, we'll have to take care to start with an asymmetric tensor, with $\bA_{ijk\ell} \neq \bA_{\pi(ijk\ell)}$ for permutations $\pi$, so that no hyperedge can appear with even multiplicity by being grouped from the edges $(i,k), (j,\ell)$ and also $(i,j), (k,\ell)$.}
For a hyperedge of multiplicity $m$, every copy of the hyperedge is sampled in the same way with probability at least $(1/2)^{m}$,  which becomes $(1/2)^{2d\ell}$ for the $2d\ell$ hyperedges in the graph.
\end{proof}

Now, using the shorthand $\cE(\cdot\text{ even}) \defeq\cE_{I_1,\ldots,I_{2\ell}}(\cdot \text{ even})$, we already have that
\begin{align*}
    \Pr_{H\sim\cH}\left[ \cE(H \text{ even}) \right]
    &=
    \frac{\Pr\left[ ~\cE(H \text{ even}) , \cE(E \text{ even})~\right] }{\Pr[\cE(E \text{ even})~|~ \cE(H \text{ even}))}
    \le
    2^{2d\ell}\cdot \Pr\left[ \cE(H \text{ even}) ~|~ \cE(E \text{ even})\right] \cdot \Pr\left[\cE(E \text{ even})\right]\mper
\end{align*}
Further, we have our bound from \pref{lem:simpleeven}, so if we could bound $\max_{I_1,\ldots,I_{2\ell}} \Pr[\cE(H \text{ even}) ~|~ \cE(E \text{ even})] \le d^{-2k\ell}$, we would be done.
But this conditional probability is not always small--for example, there is the case when $I_1=\cdots=I_{2\ell}$ are all multisets containing the same vertex $i \in [n]$ with multiplicity $2d$.
In this case, the probability that we sample an even hypergraph is $1$.

Still, so long as there are sufficiently many different vertices in $I_1,\ldots,I_{2\ell}$, we can prove that this conditional probability is small enough:
\begin{lemma}
	\label{lem:matching-columns}
    Let $E_1,\ldots, E_{2\ell} \in [n\times n]^{2d}$ be multisets of edges such that
    every edge is present in the union at least twice, and
    the number of distinct edges in the union is at least $(1-\beta) 2d\ell$, i.e.,  $|\cup_{i=1}^{2\ell} E_i| \geq (1-\beta) 2d\ell$.

Let $P_i$ denote a uniformly random pairing of elements within $E_{i}$ sampled independently for each $i \in[2\ell]$.
Then there exists a constant $c_\beta$ depending only on $\beta$ such that
\[
    \Pr[ \cup_i P_i \text{ has every pair with even multiplicity} ]  \leq \left( \frac{c_\beta}{d} \right)^{(1 - 10\beta)d\ell} \mper
\]
\end{lemma}
\begin{proof}[Proof (sketch, details in proof of \pref{lem:even-matching}).]
    Suppose we make our pairing decisions one multiset at a time.
We must pair the last copy of each edge correctly, so that all its pairs have even multiplicity.
There are $2d$ edges per matching, so the probability that we make this last decision correctly is $\approx \Omega(d)^{-1}$.
We make $d$ pairing decisions per matching, and we make the ``last'' decision about half of the time since every edge appears close to twice on average--this gives the probability to be roughly $\Omega(d)^{-d\ell}$.
\end{proof}

Now, as there are only $\approx n^{(1-\alpha) \cdot 2d\ell}$ choices of sets $I_1,\ldots,I_{2\ell}$ which could have at most $(1-\alpha) \cdot 2d\ell$ different edges, these sets contribute negligibly to the sum, and we have that
\begin{align*}
    \sum_{I_1,\ldots,I_{2\ell}}
     \Pr_{H\sim\cH}\left[ \cE_{I_1,\ldots,I_{2\ell}}(H \text{ even}) \right]
     &\le 2^{2d\ell} \cdot \left(n^{(1-\alpha) \cdot 2d\ell} + \Pr[\cE(H \text{ even}) ~|~ \cE(E \text{ even})] \sum_{I_1,\ldots,I_{2\ell}} \Pr\left[\cE_{I_1,\ldots,I_{2\ell}}(E \text{ even})\right]\right)\\
     &\le 2^{2d\ell} \cdot \left(n^{(1-\alpha) \cdot 2d\ell} + \left(\frac{c_{\alpha}}{\sqrt{d}}\right)^{(1-10\alpha)2d\ell} \cdot n^{2d\ell}\right)
\end{align*}
\noindent Balancing the terms concludes the proof; we will fill in the few remaining details in \pref{sec:tensor-D-even}.

\subsection{From $k$-XOR to tensor norms, and odd-order tensors}
\label{sec:xordet}

The proof of \pref{thm:4-upper} generalizes to tensors of all even orders $k$ almost immediately.
For odd $k$ we need an extra idea or two, since all natural flattenings of the tensor to a matrix result in a non-square matrix.
We give the details for even and odd $k$ in \pref{sec:tensor-D-even} and \pref{sec:tensor-D-odd} respectively.

As hinted earlier, to apply these ideas to strongly refute $k$-XOR we need to overcome two main hurdles.
First, as the number of clauses is small, $m \approx p \cdot n^{k} < n^{k/2}$, the tensor corresponding to the instance is sparse enough that the injective tensor norm $\max_{y \in \R^n} | \iprod{\bT,y^{\otimes k}}|$ is maximized by {\em sparse} vectors $y$.
Sparse vectors $y \in \R^n$ are too far from the solutions of interest, namely Boolean vectors $x \in \sbits^n$, which are in a sense maximally dense.

To address this issue, we will consider a sub-matrix of the tensored matrix $A^{\otimes d}$.
Again, let us consider the case of $k = 4$.
Recall that,
$ \max_{x \in \sbits^n} \iprod{\bA, x^{\otimes 4}}^d = \max_{x \in \sbits^n} \left| x^{\otimes 2d} A^{\otimes d} x^{\otimes 2d} \right|$.
The rows and columns of $A^{\otimes d}$ are indexed by $I,J \in [n]^{2d}$.
We refer to a tuple $I \in [n]^{2d}$ as {\it high multiplicity} if there is some $i \in [n]$ which has multiplicity greater than $100\log n$ in $I$ (since we are interested in the case when $d = n^{\delta} \gg \log n$).
The rows and columns of $A^{\otimes d}$ corresponding to such tuples will be referred to as high-multiplicity rows and columns.
Let $\Gamma$ denote the projection on to the low-multiplicity indices, $(\Gamma x)_I = x_I \cdot \Ind[I \text{ not high-multiplicity}]$.

The key idea is that for a Boolean vector $x \in \sbits^{n}$, almost all of the $\ell_2$-norm of $x^{\otimes 2d}$ is concentrated within the low-multiplicity indices, i.e., $\norm{\Gamma x^{\otimes 2d}} \approx \norm{ x^{\otimes 2d}}$.
However, for a sparse vector $y \in \R^n$, $\norm{\Gamma y^{\otimes 2d}} \ll \norm{y^{\otimes 2d}}$.
Therefore, we eliminate the sparse maxima of the polynomial, by restricting the matrix to the low-multiplicity rows and columns, and then apply the averaging over row and column permutations.
Specifically, the spectral upper bound used by the refutation algorithm is,
\begin{align*}
    \max_{x \in \sbits^n} \left|\iprod{\bA, x^{\otimes 4}}\right|^d  =  \max_{x \in \sbits^n} \left| (x^{\otimes 2d})^{\top} A^{\otimes d} x^{\otimes 2d} \right|
    & \approx  \max_{x \in \sbits^n} \left| (x^{\otimes 2d})^{\top} \left(\Gamma A^{\otimes d}\Gamma^T \right)x^{\otimes 2d} \right|\\
 & \leq n^{2d} \cdot \Norm{\E_{\Pi,\Sigma \in \hat{\cS}_{2d}} \left[ \Pi \left(\Gamma A^{\otimes d}\Gamma^T \right) \Sigma \right]}\mper
\end{align*}

The second challenge is that, in the sparse regime where $p \le 1/n^{k/2}$, the entries of the random matrix $A$ are ill-behaved.
Specifically, the entries of $A$ have distributions with unusually large higher moments, completely unlike Gaussian or Rademacher random variables.
For example, the $2r^{th}$ moment of an entry $\E[A_{ijk\ell}^{2r}] = p \gg (\E[A_{ijk\ell}^2])^r = p^r$.
In the trace calculation we outlined earlier, each term of the sum was either $0$ if any variable had odd multiplicity, and otherwise $1$.
In the sparse regime, different terms in the trace contribute vastly different amounts, depending on the multiplicities involved.
So we must count our hypergraphs precisely, taking into account the multiplicity of each hyperedge, rather than the just the parity.
We use the {\it encoding technique} to count the number of hypergraph structures accurately, in a way reminiscent of similar arguments in random matrix theory (e.g. \cite{FK81}).
Although the counting argument involved is more subtle than the case of random $4$-tensors (see \pref{sec:pr-hgraph}), we are still able to use the same $2$-step hyperedge sampling process to simplify the counting.

\section{Injective Tensor Norm for Subgaussian Random Tensors}\label{sec:randombds}

In this section, we show how to certify bounds on the norm of a random tensor, building on our proof of the order-$4$ case in \pref{sec:tech-over}.
We handle the even-order and odd-order cases separately, as the odd-order case contains some additional intricacies.

\pref{sec:tensor-D-even} contains the proof for even tensors.
\pref{sec:tensor-D-odd} contains the proof for odd tensors.
In \pref{sec:combinatorial-lemma}, we prove a combinatorial lemma that we rely upon in both proofs.

\subsection{Even-order tensors}\label{sec:tensor-D-even}
The case of order-$k$ tensors when $k$ is even is almost completely outlined in \pref{sec:tech-over}, in the proof overview of \pref{thm:4-upper}.
Some of the statements from the overview need additional proof, and some need generalization for $k > 4$.
We briefly fill in the gaps.

Recall that in our setting, we are given a symmetric order-$k$ tensor $\bA$ with i.i.d. standard Gaussian entries, where $k$ is even.
Our algorithm consists of computing the operator norm of a certificate matrix; though we described this certificate ion \pref{sec:tech-over}, we will require one small twist to make our proofs easier:
\begin{algorithm}[Certifying even $k$-tensor norms]\label{alg:eten}~

    {\bf Input:} An order-$k$ dimension-$n$ tensor $\bA$, for even $k$.
    \begin{compactenum}
    \item Form the asymmetric tensor $\bA'$ from $\bA$ as follows.
	For each $S \in [n]^k$,
	\begin{compactenum}
	\item if $S$ is lexicographically first among all permutations of $S$, set $\bA'_S = \sum_{\pi\in \cS_k}\bA_{\pi(S)}$.
	\item otherwise, set $\bA'_S = 0$.
	\end{compactenum}
    \item Take the natural $n^{k/2} \times n^{k/2}$ matrix flattening $A$ of $\bA'$, and form $A^{\tensor d}$.
    \item Letting $\hat \cS_{dk/2}$ be the set of all permutation matrices that perform the index permutations corresponding to $\cS_{dk/2}$ on the rows and columns of $A^{\otimes d}$, form
	\[
	    C_d \defeq \E_{\Pi,\Sigma \in \hat\cS_{dk/2}}\left[\Pi A^{\tensor d} \Sigma \right]\mper
	\]
    \end{compactenum}
    {\bf Output:} $\|C_d\|^{1/d}$ as a bound on the objective value.
\end{algorithm}

First, we verify the completeness of the certificate:
\begin{lemma}\label{lem:D-lower}
    Let $\bA$ be a symmetric order-$k$ tensor for even $k$, and let $A$ be the natural matrix flattening of $\bA'$ the asymmetrization of $\bA$ described in \pref{alg:eten}.
	Let $\cS_{dk/2}$ be the symmetric group on $dk/2$ elements, and further let $\hat \cS_{dk/2}$ be the set of $n^{dk/2} \times n^{dk/2}$ matrices that apply the permutations of $\cS_{dk}$ to matrices whose rows and columns are identified with multisets in $[n]^{dk}$.
    Then
    \[
	\left\|\E_{\Pi,\Sigma \in \hat\cS_{dk}}\left[\Pi(A^{\tensor d})\Sigma\right]\right\|^{1/d} \ge \|\bA\|_{inj}\mper
\]
\end{lemma}
\begin{proof}
    The proof is identical to that of \pref{prop:properrelax}, up to noticing that $\langle \bA, x^{\otimes k}\rangle = \langle \bA', x^{\otimes k}\rangle$.
\end{proof}

Now, we will prove that in the case that $\bA$ is a random tensor with i.i.d. subgaussian entries, our certification algorithm improves smoothly upon the simple spectral algorithm as we invest more computational resources.
\begin{theorem}\label{thm:D-upper}
    Let $n,k,d\in \N$, with even $k$.
    Let $\bA$ be a symmetric order-$k$ tensor with independent entries distributed symmetrically about $0$.
    Let $A$ be the matrix flattening of $\bA'$, the asymmetrization of $\bA$ described in \pref{alg:eten}.
    Then if $d \ll n^{1/3k^2}$, there exists a constant $c$ such that with high probability over $\bA$,
    \[
	\left\| \E_{\Pi,\Sigma}\left[\Pi A^{\tensor d} \Sigma\right]\right\|^{1/d}
    \le (c \log^2 n )^k \cdot d^{\frac{k^2\log d}{4\log n}} \cdot \frac{n^{k/4}}{d^{(k-2)/4}}\mper
    \]
\end{theorem}
\noindent The proof is nearly identical to the $k = 4$ case from \pref{sec:tech-over}, so we will be brief.
\begin{proof}
    We will assume that each entry of $\bA'$ is bounded in absolute value by $\gamma = O(\sqrt{d \log n}))$, as by the subgaussian assumption this is true with high probability, even after symmetrization.
    This assumption preserves the symmetry of the distribution.

    As in the proof of the $k=4$ case from \pref{sec:tech-over}, we will use the trace power method.
    For shorthand, let $C \defeq \E_{\Pi,\Sigma\in\cS_{dk}}\left[\Pi A^{\tensor k} \Sigma\right]$.
Let $\cH$ be the set of all hyperedge configurations possible (the set of all possible length-$2\ell$ sequences of hypergraph matchings on two sets of $dk/2$ vertices).
    Let $\cV$ be the set of all vertex configurations possible (the set of all possible length-$2\ell$ sequences of vertex multisets $I_1,\ldots,I_{2\ell} \in [n]^{dk/2}$).
    We note now that there are not many vertex configurations which use few vertices in $[n]$:
\begin{fact}\label{fact:few-configs}
    Let $\cV_{\alpha}$ be the set of vertex configurations on $dk\ell$ vertices containing fewer than $\alpha dk \ell/2$ distinct vertices from $[n]$.
    Then
    \[
	|V_{\alpha}| \le (\alpha d k \ell/2)^{(1-\alpha/2)dk\ell} \cdot n^{\alpha d k \ell/2}.
    \]
\end{fact}
\begin{proof}
    There are only $n^{\alpha d k \ell/2}$ choices for vertex labels, and then $(\alpha d k \ell/2)^{(1-\alpha/2)dk\ell}$ choices for the rest.
\end{proof}

For $H \in \cH$ and $V \in \cV$, we let $w_{\bA}(V,H)$ denote the product of all hyperedge weights in the hyperedge cycle $(V,H)$ when the weights are given by entries of the tensor $\bA$.
Because the entries are distributed symmetrically about $0$, we have that
\begin{align}
    \E_{\bA}\left[ \Tr((CC^\top)^\ell)\right]
    &= \sum_{V \in \cV} \E_{H\in \cH}\left[\E_{\bA}[w_{\bA}(V,H)]\right]\nonumber\\
    &\le \sum_{V \in \cV} \E_{H\in \cH}\left[\gamma^{2d\ell} \cdot \Ind[ (V,H) ~\text{even,}\neq 0~]\right]
    ~=~ \gamma^{2d\ell} \cdot \sum_{V \in \cV} \Pr_{H\in \cH}\left[ (V,H) ~\text{even,}\neq 0~\right],\nonumber
    \intertext{where $\Ind[\cdot]$ is the $0-1$ indicator for an event.
    Notice that now, evenness is not enough to ensure that we have nonzero contribution--because we asymmetrized $\bA$, every hyperedge also has to be lexicographically first, meaning it appears either as $\bA'_{S,T}$ or $\bA'_{T,S}$ depending on whether it comes from a $C$ or $C^\top$ term.
    Using \pref{fact:few-configs} to argue that the number of vertex configurations with fewer than $(1-\beta)dk\ell/2$ distinct vertices (the number of $V \in \cV_{(1-\beta)}$) cannot be too large,}
    &\le \gamma^{2d\ell} \left( \left((\frac{1+\beta}{2}dk\ell)^{(1+\beta)}n^{(1-\beta)}\right)^{dk\ell/2} + \sum_{V \not\in \cV_{(1-\beta)}} \Pr_{H\in \cH}\left[ (V,H) ~\text{even,}\neq 0~\right]\right)\mper \label{eq:left-off}
\end{align}
So for a fixed $V \in \cV$, we will bound $\Pr_{H} [(V,H)~ \text{even,} \neq 0]$.

To do this, we repeat our argument from \pref{sec:tech-over}.
Fixing a vertex configuration $V = I_1,\ldots,I_{2\ell}$, we sample $H \sim \cH$ uniformly in two steps:
\begin{compactenum}
\item Sample a random perfect matching (of edges, not hyperedges) between every two consecutive vertex sets $I_i, I_{i+1}$, letting the configuration of edges we chose be $E$ from the set of all such possible configurations $\cM$. \label{step:medge2}
\item Group the edges between $I_i$ and $I_{i+1}$ into groups of size $k/2$, and merge every group into a hyperedge (of order $k$).\label{step:hedge2}
\end{compactenum}

Let $(V,E)$ be the intermediate graph in this process that produces the hypergraph $(V,H)$.
Notice that now,
    We restate, then prove, a more precise version of \pref{lem:graph-from-hgraph-nocount}
    \begin{lemma}[formal version of \pref{lem:graph-from-hgraph-nocount}]\label{lem:graph-from-hgraph-nocount2}
    Let $V = I_1,\ldots,I_{2\ell} \in [n]^{dk/2}$, and suppose we have sampled hyperedges $H \in \cH$ by first sampling simple edges $E \in \cM$ as in \pref{step:medge2} and then grouping them into groups of $k/2$ as in $\pref{step:hedge2}$.
Then
\[
    \Pr((V,E) \text{ even}~|~ (V,H) \text{ even}, \neq 0) \ge \left(\frac{1}{\frac{k}{2}!}\right)^{2d\ell}\mper
\]
\end{lemma}
\begin{proof}
    Suppose every hyperedge in $H$ is lexicographically first and has even multiplicity.
    Each hyperedge $h$ in $H$, $h$ was sampled from one of the $(k/2)!$ matchings of its left-hand vertices to its right-hand vertices with equal probability.
Let $h_1,\ldots,h_m$ be the distinct labeled hyperedges of our hypergraph.
    Since all our hyperedges are lexicographically first, the same bipartition of vertices is common to every appearance of $h_i$ for all $i \in [m]$.
    Thus, if we choose a uniformly random perfect matching of simple edges in each hyperedge of the hypergraph, we choose the same simple matching for all copies of $h_i$ with probability at least $(\frac{k}{2}!)^{-\# h_i}$.
    It follows that if all hyperedges appear in $(V,H)$ with even multiplicity, then with probability at least $(\frac{k}{2}!)^{-2d\ell}$ all simple edges in $(V,E)$ appear with even multiplicities.
\end{proof}

    Applying \pref{lem:graph-from-hgraph-nocount2},
\begin{align}
    \Pr((V,H)~\text{even}, \neq 0)
    &= \frac{\Pr((V,H) ~\text{even},\neq 0 ~\&~ ~(V,E) ~\text{even})}{\Pr((V,E) ~\text{even} ~|~ (V,H) ~\text{even},\neq 0)}\nonumber \\
    &\le \left(\frac{k}{2}!\right)^{2d\ell} \cdot\Pr((V,E) ~\text{even})\cdot \Pr((V,H) ~\text{even}~|~(V,E) ~\text{even})\mper \label{eq:cond}
\end{align}

We now relate the quantity $\sum_{V \in \cV} \Pr_{E}[(V,E)~\text{even}]$ to a matrix quantity we can control well.
Letting $B$ be an $n\times n$ matrix with symmetric i.i.d. entries uniform from $\{\pm 1\}$, and letting $C' = \E[\Pi B^{\tensor dk/2} \Sigma]$,
\[
    \E\left[\Tr((C' C'^\top)^\ell)\right]
   =
    \sum_{V \in \cV} \Pr_{E}[(V,E)~\text{even}]\mper
\]
We now prove and apply the following proposition, which is a restatement of \pref{lem:simpleeven} for arbitrary $k$:
\begin{proposition}\label{prop:simple-graph-bound-1}
    Let $n,d,k,\ell \in \N$ so that $dk\ell\log n \ll n$.
    Let $C' = \E_{\Pi,\Sigma \in \cS_{dk/2}} \left[\Pi B^{\tensor dk/2}\Sigma\right]$, for an $n \times n$ matrix $B$ with i.i.d. Rademacher entries.
    Then
    \[
    \E\left[\Tr((C' C'^\top)^\ell)\right]
	\le 2^{4dk\ell+1}n^{dk\ell/2 + dk/2}\mper
    \]
\end{proposition}
\begin{proof}
    Let $B$ be an $n \times n$ matrix with i.i.d. Rademacher entries, and let $d,\ell \in \N$.
    We have that for any $N \times N$ PSD matrix $P$,
	$\Tr\left(P^\ell\right)
	\le N \cdot \left\|P^\ell\right\|$,
    and because $C'C'^\top$ is PSD it follows that
    \begin{align}
	\Tr\left((C'C'^\top)^\ell\right)
	&\le n^{dk/2}\cdot \left\|(C'C'^\top)^\ell\right\|\mper\label{eq:trfact}
    \end{align}
    We will get a bound on $\E\|(C'C'^\top)^\ell\|$.
    Because $C'$ is symmetric, $C'C'^\top = (C')^2$.
    Thus, a bound on $\E\|C'^{2\ell}\|$ will suffice.
    We apply the triangle inequality and the submultiplicativity of the norm to deduce that for any $B$,
    \begin{align*}
	\|C'^{2\ell}\|
	~=~ \left\|\left(\E_{\Pi,\Sigma \in \hat\cS_{2d}}\left[ \Pi (B^{\tensor dk/2}) \Sigma\right] \right)^{2\ell}\right\|
	&\le \left(\E_{\Pi,\Sigma \in \hat\cS_{2d}}\left[\| \Pi\| \cdot \| (B^{\tensor dk/2}) \| \cdot \|\Sigma \|\right]\right)^{2\ell}
	~\le~ \| B \|^{dk\ell}\mcom
    \end{align*}
    and now, we can use standard arguments from random matrix theory to get tail bounds on $\|B\|$.
    From \pref{thm:rademacher}, we have that $\Pr[\|B\| -12 n^{1/2}\ge s] \le \exp(-s^2/16)$, and we also have that $\|B\| \le \|B\|_F \le n$, and thus it follows that
    \begin{align*}
	\E\left[\|C'^{2\ell}\|\right]
	~\le~ \E\left[\|B\|^{dk\ell}\right]
	&\le \Pr[\|B\| \le 16 \sqrt{n}] \cdot (16\sqrt{n})^{dk\ell} + \Pr[\|B\| > 16 \sqrt{n}]\cdot n^{dk\ell}\\
	&\le (1-\exp(-n)) \cdot (16\sqrt{n})^{dk\ell} + \exp(dk\ell \log n -n) ~\le~ 2(16\sqrt{n})^{dk\ell}\mcom
    \end{align*}
    and the conclusion follows from combining the above with \pref{eq:trfact}.
\end{proof}

We thus have
\begin{align}
    \sum_{V \in \cV_{(1-\beta)}} \Pr_{E}[(V,E)~\text{even}]
    \le
    \sum_{V \in \cV} \Pr_{E}[(V,E)~\text{even}]
    \le
    \E\left[\Tr((C' C'^\top)^\ell)\right]
    \le 2^{4dk\ell+1}n^{dk\ell/2 + dk/2}\mper\label{eq:left-off-2}
\end{align}

Now, from \pref{eq:cond} we are left to bound $\Pr[(V,H) ~\text{even}~|~(V,E) ~\text{even},~V \in \cV_{(1-\beta})]$.
We apply the following lemma:
\begin{lemma}
    \torestate{
    \label{lem:vtx-edge}
    Let $m,n,c \in \N$, and let $G$ be a graph which is a union of at most $c$ disjoint cycles.
    Suppose furthermore that each vertex receives labels from the set $[n]$, that every labeled edge appears with even multiplicity, and that there are exactly $m$ distinct labeled edges.
    Then letting $L$ be the number of distinct vertex labels, we have
    \[
	L \le m + c.
    \]
}
\end{lemma}
\noindent The proof of \pref{lem:vtx-edge} proceeds by a cute inductive argument, which we will reserve for \pref{sec:extra-lemmas}.

\pref{lem:vtx-edge} implies that if $(V,E)$ has at least $(1-\beta)dk\ell/2$ distinct vertices, then it must have at least $(1-\beta)dk\ell/2 - dk/2$ distinct edges.
Let $E_1,\ldots, E_{2\ell}$ be the matchings in $E$ so that $E_i$ gives the edges between $I_i, I_{i+1}$.
We invoke and prove a generalization of \pref{lem:matching-columns}:
\begin{lemma}
    \torestate{
    \label{lem:even-matching}
Fix $M,r, \ell, N \in \N$ and $\beta \in  (0,1)$.
Let $E_1,\ldots,E_M \in [N]^{r \cdot c}$ be multisets of elements such that
    the number of distinct elements in the union $\cup_{i\in[M]} E_i$ is at least $(1-\beta)M\cdot r \cdot c/2$.
Let $G_i$ denote a uniformly random $r$-grouping of elements within $E_i$, sampled independently for each $i\in[M]$.
Let $\bigoplus_i G_i$ denote the set of $r$-groups $(a_1,\ldots,a_r) \in [N]^r$ that appear an odd number of times within $\cup_i G_i$.
Then for any $0 < \delta < 3.5\beta$,
\[
    \Pr[ | \oplus_i G_i | \leq \delta Mc ]  \leq  \left( \frac{112} {\beta c} \right)^{ (1-(4r+1)\beta-2\delta) (r-1) Mc/2}
\]
}
\end{lemma}
\noindent We will prove \pref{lem:even-matching} in \pref{sec:combinatorial-lemma}.

We apply \pref{lem:even-matching} to the multisets $E_i$ with parameters $M\leftarrow 2\ell$, $c\leftarrow d$, $r\leftarrow k/2$, to conclude that if the $S_i$ are each grouped into matchings of hyperedges with $d$ edges each, then
\begin{align*}
    \Pr\left[(V,H) ~\text{even}~|~(V,E) ~\text{has}\ge (1-\beta)dk\ell/2~\text{edges}~\right]
    &\le \left(\frac{112}{\beta d}\right)^{(1-(2k+1)\beta)(k/2-1)d\ell}\\
    &\le c_\beta^{dk\ell/2}\left(\frac{1}{d}\right)^{(1-3k\beta)(k/2-1)d\ell}\mper
\end{align*}
for some constant $c_\beta$ depending only on $\beta$.
Putting this together with \pref{eq:left-off},\pref{eq:cond}, and \pref{eq:left-off-2},
\begin{align*}
    &\E_{\bA}\left[ \Tr((CC^\top)^\ell)\right]\\
    ~~&\le \gamma^{2d\ell} \left( \left(((1+\beta)dk\ell/2)^{(1+\beta)}\cdot n^{(1-\beta)}\right)^{dk\ell/2} + \left(\frac{k}{2}!\right)^{2d\ell}\cdot (2^8n)^{dk\ell/2 + dk/2} \cdot c_{\beta}^{dk\ell/2} \left(\frac{1}{d}\right)^{(1-3k\beta)(k/2-1)d\ell}\right)\\
    ~~&\le (c'_\beta \cdot k^k\gamma^2\ell^{k})^{d\ell} n^{dk/2}
    \left( (d^{1+\beta}n^{1-\beta})^{dk\ell/2} + \left(\frac{1}{d}\right)^{(1-3k\beta)(k/2-1)d\ell}n^{dk\ell/2}\right)
\end{align*}
for some constant $c'_\beta$.
Choosing $\beta = \frac{2(k-1)\log d}{k(3k-7)\log d + \log n}$ balances the terms, so for smaller $\beta$ we have
\[
    \E_{\bA}\left[ \Tr((CC^\top)^\ell)\right]
    ~~\le 2(c'_\beta \cdot k^k\gamma^2\ell^{k})^{d\ell}n^{dk/2} \cdot \left(\frac{n^{k/2}}{d^{k/2-1}}\right)^{d\ell} \cdot d^{\beta(k/2-1)d\ell}\mper
\]
Now, requiring that $d \le n^{1/3k^2}$ and choosing $\beta \leftarrow (k-1)\frac{\log d}{\log n}$, we have that
\[
    \E\left[ \Tr((CC^\top)^\ell)\right]^{1/2\ell}
    \le 2(c'_\beta \cdot k^k\gamma^2\ell^{k})^{d/2}n^{dk/4\ell} \cdot \left(\frac{n^{k/2}}{d^{k/2-1}}\right)^{d/2} \cdot d^{\frac{k^2\log d}{2\log n}\cdot d/2}
\]
Taking $\ell = O(\log n)$ and applying \pref{prop:tpm}, the conclusion of \pref{thm:D-upper} follows.
\end{proof}

\subsection{Odd-order tensors}
\label{sec:tensor-D-odd}

In this section, we give our algorithm for certifying bounds on the injective tensor norm of random odd-order tensors.
Because there is no canonical way to flatten an odd-order tensor to a square matrix, the algorithm includes an additional step, similar to the one we employ for $k$-XOR instances when $k$ is odd (\pref{sec:lin-odd}).

We begin with a brief high-level overview of our algorithm.
To begin with, let $\bA \in \R^{[n]^{k}}$ be an order-$k$ symmetric tensor of dimension $n$.
For convenience, we define an integer $\kappa$ such that $k = 2\kappa+1$.
For the rest of this section, we will use $A_i$ to denote the $[n]^{\kappa} \times [n]^\kappa$ matrix obtained by flattening the $i^{th}$ slice of $\bA$, i.e.,
\[
    A_i (I,J) \defeq \bA_{(i,I,J)} \qquad \qquad \forall I,J \in [n]^\kappa \mper
\]
Using the Cauchy-Schwarz inequality, we can bound the injective norm in terms of the matrices $A_i$,
\begin{align}
\iprod{x^{\otimes 2\kappa+1}, \bA} & = \sum_{i} x_i \cdot \iprod{x^{\otimes
  \kappa},A_i x^{\otimes \kappa}}  \nonumber \\
& \leq \left(  \sum_{i } x_i^2 \right)^{1/2} \cdot \left(  \sum_i \iprod{x^{\otimes \kappa}, A_i x^{\otimes \kappa}}^2 \right)^{1/2}
~ =  \left( \iprod{x^{\otimes 2\kappa}, \left(\sum_i A_i \otimes A_i \right) x^{\otimes 2\kappa}}
  \right)^{1/2}\mper \label{eq:somename0}
\end{align}
Therefore, in order to bound $\|\bA\|_{inj}$, it is sufficient to bound the following quantity.
\begin{align} \label{eq:oddtensor1}
\max_{\norm{x} \leq 1} \Iprod{x^{\otimes 2\kappa}, \left( \sum_{i} A_i \otimes A_i \right) x^{\otimes 2\kappa}
  }
\end{align}

For a tensor $\bA$ whose entries are i.i.d. subgaussian variables, we bound the value of the maximization problem \pref{eq:oddtensor1}.

The matrix $\sum_{i} A_i \tensor A_i$ has large diagonal entries.
However, our tensoring and symmetrizing algorithm requires a matrix with eigenvalues roughly symmetric about $0$ (see the heuristic explanation in \pref{sec:tech-over}).
Thus, we will work with a diagonal-free version of the matrix.
Define the matrix $N \in \R^{[n]^{2\kappa} \times [n]^{2\kappa}}$ as follows:
$$ N_{i} ((a,b),(c,d)) =  A_i(a,c) \cdot A_i(b,d) \cdot \Ind[(a,c) \neq (b,d)]
\qquad \qquad \forall a,b,c,d \in[n]^\kappa $$

We can rewrite the polynomial in \eqref{eq:oddtensor1} as,
\begin{align*}
 \Iprod{x^{\otimes 2\kappa}, \left( \sum_{i} A_i \otimes A_i \right) x^{\otimes 2\kappa}
  } & =  \Iprod{x^{\otimes 2\kappa}, \left( \sum_i N_i \right) x^{\otimes 2\kappa}
 } + \sum_{i \in [n],a,b \in [n]^\kappa} x_a^2 x_b^2\bA^2_{i} (a,b)
 \end{align*}
 And we can upper bound the latter term by
 \begin{align}
     \sum_{i \in [n],a,b \in [n]^\kappa} x_a^2 x_b^2\bA^2_{i} (a,b)
 ~\le~ \sum_{a,b \in [n]^\kappa} x_a^2 x_b^2 \left(\sum_{i}\bA^2_{i}(a,b)\right)
 ~\leq~  \max_{a,b} \left(\sum_{i}\bA^2_{i} (a,b)\right)
  \label{eq:somename1}
\end{align}
where we have used the fact that $\|x\|^2 = 1$.
Bounding the norm of tensor $\bA$ thus reduces to upper bounding $\iprod{x^{\otimes 2\kappa}, \left( \sum_i N_i \right) x^{\otimes 2\kappa}}$.
Now our strategy is as before--we take a $d$th tensor power of our matrix, then average over the symmetries of $x^{\otimes 2\kappa d}$.

Having discussed the differences between the even and odd cases, we are ready to give our algorithm.
\begin{algorithm}[Odd-order Injective tensor norm]\label{alg:oddtensor}~\\
      {\bf Input:} A random tensor $\bA$ of dimension $n$ and odd order $k = 2\kappa+1$, and a parameter $d$.
      \begin{compactenum}
      \item Form the asymmetric tensor $\bA'$ as described in \pref{alg:eten}, so that $\iprod{x^{\otimes k}, \bA}=\iprod{x^{\otimes k}, \bA'}$ but only lexicographically first entries are nonzero.
	\item Let $A_i$ be the $n^\kappa \times n^\kappa$ matrix flattening of the $i$th slice of $\bA'$, and form the matrix
	    \[
		M := \sum_{i\in[n]} A_i \tensor A_i
	    \]
	\item Zero out all entries of the matrix corresponding to $(I_1,I_2),(J_1,J_2) \in [n]^{2\kappa}$ such that $(I_1,J_1) = (I_2,J_2)$, forming a new matrix $N$:
	    \[
		N_{(I_1,I_2),(J_1,J_2)} := M_{(I_1,I_2),(J_1,J_2)}\cdot \Ind((I_1,J_1) \neq (I_2,J_2))\mper
	    \]
	\item Take the $d$th tensor power of $N$,
	\[
	    N \to N^{\tensor d}\mper
	\]
    \item Symmetrize the rows and columns of $N^{\tensor d}$ according to the symmetries of $\cS_{2d\kappa}$ to obtain the matrix $C$,
	\[
	    C_d \defeq \E_{\Pi,\Sigma \in \hat\cS_{2d\kappa}} \left[\Pi(N^{\tensor d}) \Sigma\right]\mper
	\]
      \end{compactenum}
      {\bf Output:} The quantity $\left(\|C_d\|^{1/d} + \max_{a,b \in [n]^\kappa} \sum_{i \in [n]} A_i(a,b)^2\right)^{1/2}$ as an upper bound on $\|\bA\|_{inj}$.
  \end{algorithm}

  \begin{proposition}
      For any symmetric tensor $\bA$, \pref{alg:oddtensor} outputs a valid upper bound on $\|\bA\|_{inj}$.
  \end{proposition}
  \begin{proof}
      Our asymmetrization in step 1 ensures that $\iprod{x^{\otimes k},\bA} = \iprod{x^{\otimes k},\bA'}$.
      The proof then follows from the calculations above, beginning at \pref{eq:somename0} and ending at \pref{eq:somename1}, and then using that the symmetrization step fixes vectors of the form $x^{\tensor 2d\kappa}$.
  \end{proof}

  We prove that when $\bA$ has subgaussian, centered, independent entries, \pref{alg:oddtensor} improves over the basic spectral algorithm.
  \begin{theorem}\label{thm:odd-tensor}
      For any symmetric tensor $\bA$ with independent subgaussian centered entries, with high probability over the choice of $\bA$, \pref{alg:oddtensor} certifies that
      \[
	  \|\bA\|_{inj}\le
	\tilde O\left( \frac{n^{k/4}}{k^{(k-2)/4}} \cdot
	d^{ \frac{k^2\log{d}}{2\log n}}\right)\mper
      \]
      so long as $d \log n \ll n^{1/120}$.
  \end{theorem}
  First, the very straightforward observation that subtracting the maximum element cannot have too strong of a negative effect:
  \begin{lemma}\label{lem:easy}
    If $\bA$ is an order-$D$ tensor with i.i.d. symmetric subgaussian entries, then
    \[
	\max_{a,b \in [n]^d} \sum_{i\in[n]} \bA(i,a,b)^2 \le O(n\log n),
    \]
    with high probability.
\end{lemma}
\begin{proof}
    The lemma follows from the fact that the variables are subgaussian, and by applying first a Chernoff bound and then a union bound over the indices.
\end{proof}

Now, we bound the norm of the matrix $\|C_d\|$.
\begin{theorem}\label{thm:odd-upper}
    So long as $kd\ell < 4n^{\beta/4}$,
there exists some absolute constant $c_\beta$ depending on $\beta$ such that with high probability over the choice of $\bA$,
\[
 \norm{C_d} \leq \left( c_\beta^d\log n \cdot
	\frac{n^{k/4}}{d^{(k-2)/4 - 6k\beta}} \cdot n^{1/2\ell}
\right)^{d} \mper
\]
\end{theorem}
\begin{proof}
    Because the entries of $\bA$ are subgaussian, with high probability all entries of the tensor are bounded in magnitude by $\gamma = O(\sqrt{\kappa\log n})$.
   We will assume this to be the case in the remainder of the proof.

    We bound the expected trace $\E[\Tr\left((CC^\top)^\ell\right)]$ over the choice of $\bA$, in order to apply the tensor power method.
Let $M:=(\sum_{i} A_i \tensor A_i)^{\tensor d}$ for convenience.
    The $(A,B),(C,D)$th entry of $M$ (for $A,B,C,D \in [n]^{d\kappa}$ with $A = a_1,\ldots,a_{d}$ with $a_i \in [n]^\kappa$, and with similar decompositions defined for $B,C,D$) has value
\begin{align*}
    M_{(A,B),(C,D)}
    &=  \prod_{i \in [d]} \left(\sum_{u\in[n]}\bA_{a_i,c_i,u} \cdot \bA_{b_i,d_i,u}\right)
    =  \sum_{U \in [n]^d} \prod_{i \in [d]} \left(\bA_{a_i,c_i,u_i} \cdot \bA_{b_i,d_i,u_i}\right)\mper
\end{align*}
Interpreting the variables $\bA_{a_i,c_i,u_i}$ as $k = (2\kappa +1)$-uniform hyperedges, we have that each entry is a sum over hypergraphs indexed by $U \in [n]^d$.
    \begin{figure}
	\centering
	\includegraphics[width=0.7\textwidth]{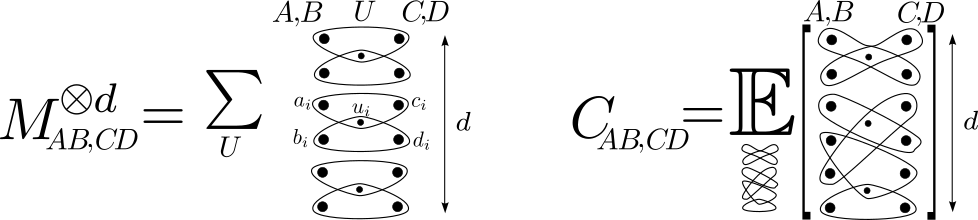}
	\caption{Hypergraphs corresponding to odd certificate entries.}\label{fig:oddten}
    \end{figure}
For each $U \in [n]^d$, we have a hypergraph on the following vertex configuration: on the left, we have the vertices from the multiset $A,B$.
On the right, we have the vertices from the multiset $C,D$.
In the center, we have the vertices from $U$.
On this vertex set, we have $2d$ hyperedges.
Of these hyperedges, $d$ form a tripartite matching on the vertices in $A,C,U$, with $\kappa$ vertices from each of $A,C$ and one vertex in $U$.
The other $d$ form a similar tripartite matching on the vertices in $B,D,U$.
Every hyperedge on $A,C,U$ shares exactly one vertex in $U$ with exactly one hyperedge from $B,D,U$.
See \pref{fig:oddten} for an illustration.

The subtraction of the square terms $\squares(A_u \tensor A_u)$ forces us to never have two hyperedges sharing a vertex in $U$ if they contain vertices of the same type in $[n]$: that is, we can never have $(a_i,c_i) = (b_i,d_i)$ as ordered multisets.
Then, the averaging operation $\E_{\Pi,\Sigma\in\hat \cS_{2d\kappa}}$ takes each such entry to an average over all allowed hyperedge configurations on the vertex set $(A,B), (C,D), U$.

When we take $\Tr(C_{(d)}C_{(d)}^\top)^\ell$, we are taking a sum over all ``cycles'' of length $2\ell$ in such hypergraphs, where the vertices in the cycle are given  by the $(A,B)$ multisets, and the edges are given by the average hyperedge configuration between $(A,B)$ and the next $(C,D)$, with the $U$ vertices in between.

    To this end, we describe an equivalent definition of the matrix $C_{(d)}$.
    Specifically, given $a,b \in [n]^{2\kappa d}$ the entry $C_{(d)}(a,b)$ can be evaluated as follows:
\begin{compactenum}
\item Sample a random matching $\cE = \{ e_1,\ldots,e_{2\kappa d} \}$ between the multisets $a$ and $b$.
\item Group the edges of $\cE$ in to $2d$ groups of size $\kappa$, to obtain $2d$ { \it blocks} $\cF = \{f_1,\ldots,f_{2d}\}$.
\item Pick a random matching $\cM$ between the blocks in $\cF$.
    Let $\cM$ be given by $d$ pairs $\{(h_i,h_i')\}_{i \in [d]}$.

\item For each choice of {\it ``pivot vertices''} $\sigma \in [n]^d$, we get a $(2\kappa+1)$-uniform hypergraph $\cH_{\sigma}$ with $2d$
  hyperedges given by
$$ \{ (\sigma_i, h_i),(\sigma_i,h_i') | i \in [d]\} \mper $$

\item Output the value $\sum_{\sigma \in [n]^d} \prod_{i \in [d]} A_{(\sigma_i,h_i)} \cdot A_{(\sigma_i,h_i')} \cdot \Ind[h_i \neq h_i']  $.
\end{compactenum}

The entries of the matrix $C$ are given by,

\[
    C(a,b) = \E_{\cE}\E_{\cF} \E_{\cM} \sum_{\sigma \in [n]^d} \left[ \prod_{i \in [d]} T_{(\sigma_i,f_i)} \cdot T_{(\sigma_i,g_i)} \cdot
\Ind[h_i \neq h_i']\right]
\]

Returning to the quantity $\Tr((CC^\top)^\ell)$, we can understand this as a sum over cycles in the entries of $C$, which gives us a sum over products of random variables corresponding to the edges in cyclic hypergraphs.
Since we have assumed the entries of $\bA$ are distributed symmetrically about $0$, each term in the sum $\E\Tr((CC^\top)^\ell)$ is non-zero only if every hyperedge appears with even multiplicity.

We can organize the terms in $ \Tr\left((CC^\top)^\ell\right)$ as follows:
\begin{compactitem}
\item For each vertex configuration $V = \{a_1,b_1,a_2,\ldots,b_\ell, a_1\} \in \cV \subset [n]^{2\kappa d}$
\begin{compactenum}
\item Sample matchings $\cE = \{\cE_1,\ldots,\cE_{2\ell}\}$
\item Group the edges in to { \it blocks} $\cF = \{\cF_1,\ldots,\cF_{2\ell}\}$
\item Pick random matchings $\cM = \{\cM_1,\ldots,\cM_{2\ell}\}$ between the
  blocks.

\item For each choice of {\it ``pivots''} $\sigma = \{\sigma_1,\ldots,\sigma_{2\ell}\} \subset [n]^d$ we get a
  $(2\kappa+1)$-uniform hypergraph $\cH_{\sigma}$ with $2d\ell$ hyperedges.

\end{compactenum}
\end{compactitem}

We will call the hypergraph $\cH_{\sigma}$ {\it diagonal-free} (or {\it d-free}) if
 there are no pairs of identical blocks matched with each other in
 $\cM$.
We will use the notation $\|\cdot \|_{\oplus}$ to denote the number of elements of odd multiplicity in a multiset, and similarly the notation $\|\cdot\|_{0}$ to denote the number of distinct elements in a multiset.
We will say $\cH_{\sigma}$ is {\it even} if the number of occurrences
 of each hyperedge is even.
Now, dividing by our upper bound on the absolute value of the maximum entry,
\begin{align}
    \gamma^{-2\kappa d\ell}&\E_T\left[ \Tr\left((CC^\top)^\ell\right)\right]\nonumber \\
    & \le \sum_{V \in \cV} \E_{\cE} \E_{\cF} \E_{\cM}  \left[ \sum_{\sigma}  \Ind[\cH_{\sigma} \text{ even \& d-free}] \right] \nonumber \\
& \leq (\kappa!)^{4d\ell} \cdot  \sum_{V \in \cV} \E_{\cE} \E_{\cF} \E_{\cM}  \left[ \sum_{\sigma}  \Ind[\cH_{\sigma} \text{ even \& d-free}] \cdot \Ind[ \|\cE\|_{\oplus} = 0 ] \right] \qquad (\text{by \pref{lem:graph-from-hgraph-nocount}})\nonumber \\
& = (\kappa!)^{4d\ell} \cdot  \sum_{V \in \cV} \E_{\cE} \Ind\left[\onorm{\cE}=0\right] \E_{\cF} \E_{\cM}  \left[ \sum_{\sigma}  \Ind\left[\cH_{\sigma} \text{ even \& d-free} \right] \right]\nonumber\\
& \leq  (\kappa!)^{4d\ell} \cdot  \sum_{V \in \cV} \E_{\cE} \Ind\big[\onorm{\cE}=0  \wedge \dnorm{\cE} \geq 2d\kappa\ell(1-\beta)\big] \E_{\cF} \E_{\cM}  \left[ \sum_{\sigma}  \Ind[\cH_{\sigma} \text{ even \& d-free} ] \right]  \label{eq:term1} \\
& + (\kappa!)^{4d\ell} \cdot  \sum_{V \in \cV} \E_{\cE} \Ind\big[\onorm{\cE} = 0 \wedge \dnorm{\cE} \leq 2d\kappa\ell(1-\beta)\big] \E_{\cF} \E_{\cM}  \left[ \sum_{\sigma}  \Ind[\cH_{\sigma} \text{ even \& d-free} ] \right] \label{eq:term2}
\end{align}

First we will bound the value of term in \eqref{eq:term2}.
Recall that by \prettyref{lem:vtx-edge}, if $\cE$ is even then the number of distinct labels in $V\in \cV$ is less than $\dnorm{\cE}$.
Therefore,
\[
    \card{\left\{\cE ~|~ \onorm{\cE} = 0 \wedge \dnorm{\cE} \leq 2d\kappa \ell
    (1-\beta)\right\}} <  n^{2d\kappa\ell(1-\beta)} (4d\kappa\ell)!
\]
Now, we will use the following claim:
\begin{claim} \label{claim:trivialbound}
For every choice of $V\in \cV,\cE,\cF,\cM$,
$$ \sum_{\sigma} \Ind [ \onorm{\cH_{\sigma}} = 0 \wedge \cH_{\sigma} \text{ is
diagonal-free}] \leq (2d\ell)! \cdot n^{d\ell}$$
\end{claim}
\begin{proof}
If $\cH_{\sigma}$ is diagonal-free and even, then we claim that each {\it
  pivot} value appears twice.  Suppose not, if $\sigma_i$ is such that
$\sigma_i \neq \sigma_j$ for all $j \neq i$.  Since $\cH_{\sigma}$ is diagonal free, the two
hyperedges involving $\sigma_i$ are distinct.  Since this is the unique
occurrence of these two hyperedges in $\cH_{\sigma}$, $\cH_{\sigma}$ cannot
be {\it even}--a contradiction.  With each {\it pivot} appearing
at least twice, the number of distinct choices of $\sigma$ is at most
$(2d\ell)! n^{d\ell}$.
\end{proof}

By \prettyref{claim:trivialbound}, for each $\cE$ the corresponding
term is at most $(2d\ell)! n^{d\ell}$.  In all, this shows that
\eqref{eq:term2} can be bounded as
\begin{align}
	\text{ \eqref{eq:term2}} &
\leq (\kappa!)^{4d\ell} \cdot  n^{2d\kappa\ell(1-\beta)} (4d\kappa\ell)!
                                   \cdot \left( (2d\ell)! \cdot n^{d \ell} \right)
 \leq \left( \frac{ (d\ell \kappa)^{5\kappa} \cdot n^{2\kappa+1}}{n^{2\kappa\beta}} \right)^{d\ell} \leq \left(
\frac{n^{2\kappa +1}}{d^{2\kappa-1}}\right)^{d\ell} \label{eq:bound100}
 \end{align}
 where the final simplification uses $d \kappa \ell < n^{\beta}{4}$.

 Now we bound \eqref{eq:term1}.
 Using \pref{prop:simple-graph-bound-1} and reasoning similar to that in the proof of \pref{thm:D-upper}, by making an analogy between the set of configurations with even $\cE$ and the norm of a random matrix under our tensoring and averaging operations, we know that
\[
    \sum_{V\in\cV} \E_{\cE} \Ind[\onorm{\cE} = 0 ] \leq c^{2\kappa d\ell} n^{2\kappa d\ell+\kappa d}
\]
for an absolute constant $c > 0$.
Moreover, conditioned on
$\dnorm{\cE} \geq 2d\kappa \ell(1-\beta)$,
we will show the following bound
	$$ \E_{\cF} \E_{\cM} \sum_{\sigma \in [n]^{d\ell}} \Ind[\onorm{\cH_{\sigma}} =
0] \leq \left( \ell c_{\kappa\beta} \cdot \frac{n}{d^{(2\kappa-1) - 8\kappa^2\beta}} \right)^{d\ell}
$$
for a constant $c_{\kappa\beta}$ depending only on $\kappa,\beta$ in
\prettyref{lem:bigone}.  By the preceding pair of inequalities, we get
that
\begin{align} \label{eq:bound200}
    \text{\eqref{eq:term1}} \leq \left(\ell {c'}_{\kappa\beta} \cdot
	\frac{n^{2\kappa+1}}{d^{2\kappa-1 - 8 \kappa^2\beta}} \cdot n^{\kappa/\ell}
	\right)^{d\ell}
\end{align}
From \eqref{eq:bound100} \& \eqref{eq:bound200}, we conclude that
\begin{align*}
	\left( \E_{T} \left[\Tr\left((CC^T)^{\ell}\right)\right]
	\right)^{1/2\ell} \leq  \left(\ell c_{\kappa\beta} \cdot
	\frac{n^{\kappa+\half}}{d^{\kappa-\half - 4 \kappa^2\beta}} \cdot n^{\kappa/2\ell}
	\right)^{d}
\end{align*}

By \prettyref{prop:tpm}, taking $\ell = O(\log n)$, we conclude that
$$ \Pr\left[ \norm{C} \leq \left( c_{\kappa\beta}' \log n \cdot
	\frac{n^{\kappa+\half}}{d^{\kappa-\half - 4 \kappa^2\beta}}
\right)^{d} \right] \geq 1- n^{-100}\mcom$$
\end{proof}

We can now put together the easy bound on the maximum diagonal entry with the bound on $\|C\|$ to prove \pref{thm:odd-tensor}.

\begin{proof}[Proof of \pref{thm:odd-tensor}]
    \pref{alg:oddtensor} returns the upper bound
    \[
	\left(\|C\|^{1/d} + \max_{I,J} \left(\sum_i \bA_{i,I,J}^2 \right)\right)^{1/2}\mper
    \]
    We combine \pref{lem:easy} with \pref{thm:odd-upper}, and we have that with high probability, for constants $c_\beta$ and $c_2$,
	\begin{align}
	\left(\|C\|^{1/d} + \max_{I,J} \left(\sum_i \bA_{i,I,J}^2 \right)\right)^{1/2}
	&\leq \left(\log n\cdot c_{\kappa\beta} \cdot \frac{n^{\kappa+\half}}{d^{\kappa-\half - 4\kappa^2\beta}} + c_2 n \log n\right)^{1/2}
\end{align}
By picking the best possible $\beta$ under the constraint $\beta <
1/30$ and $d\kappa \ell < n^{\beta/4}$, we have that the former term always dominates, and we get the bound:
\begin{align*}
	\norm{\bA}_{inj}
	\leq
	\tilde O\left( \frac{n^{(2\kappa+1)/4}}{d^{(2\kappa-1)/4}} \cdot
	d^{ 2\kappa^2\frac{\log{d}}{\log n}}\right)\mper
\end{align*}
This concludes the proof.
\end{proof}

Now, we prove some of the lemmas we have relied upon in the proof of \pref{thm:odd-upper}.
We begin with a lemma bounding the probability that the hyperedges we sample all have even multiplicity.

\begin{lemma}  \label{lem:bigone}
Suppose $k < n^{\beta/4}$ and $\beta < 1/30$.
Then conditioned on an $\cE$ such that $\dnorm{\cE} \geq 2kd\ell (1-\beta)$,
	$$ \E_{\cF} \E_{\cM} \sum_{\sigma \in [n]^{k\ell}} \Ind[\onorm{\cH_{\sigma}} =
0] \leq \left( \ell c_{d\beta} \cdot \frac{n}{k^{(2d-1) - 8d^2\beta}} \right)^{k\ell}
$$
where $c_{d\beta}$ is a constant depending on $\beta$ and $d$.
\end{lemma}

\begin{proof}
Note that $\onorm{\cH_{\sigma}} = 0$ implies that $\onorm{\cF} = 0$.
By applying \prettyref{lem:even-matching} with $r \leftarrow d$, $c \leftarrow 2k$, $M \leftarrow 2\ell$, and $E_i \leftarrow \cE_i$, we obtain the following bound over the choice of $\cF$.
\begin{align}
    \Pr_{\cF}[ \onorm{\cF} = 0 | \dnorm{\cE} \geq 2kd \ell(1-\beta)] \leq \left( \frac{112}{2\beta k} \right)^{2(d-1) k\ell (1-(4d+1)\beta)}
\label{eq:11}
\end{align}
Furthermore, if $\dnorm{\cE} \geq 2kd \ell(1-\beta)$ then clearly $\dnorm{\cF} \geq 2k\ell(1-\beta)$.
By \prettyref{lem:fixedF}, for every $\cF$ with $\dnorm{\cF}\geq 2k
\ell(1-\beta)$ we have,
\begin{align}
  \E_{\cM} \sum_{\sigma \in [n]^{k\ell}} \Ind[\onorm{\cH_{\sigma}} = 0]
  & \leq  
\left(4 \beta k \ell \right)! \cdot n^{k \ell} \cdot
  \left(\left(\frac{112}{\beta k}\right)^{k \ell(1-10 \beta)}   +
  n^{-\beta k \ell/3}\right) \nonumber\\
   & \leq  \left( k^{4\beta}  \ell^{4 \beta} \cdot n \cdot
 \left( \left(\frac{112}{\beta k}\right)^{(1-10 \beta)}  +
 n^{-\beta /3}\right) \right)^{k\ell} 
 \label{eq:22}
\end{align}
Using \eqref{eq:11} and \eqref{eq:22} we conclude that,
\begin{align*}
 \E_{\cF} \E_{\cM} \sum_{\sigma \in [n]^{k\ell}} \Ind[\onorm{\cH_{\sigma}} = 0]
 \leq \left( \left(\frac{112}{2\beta k}  \right)^{2(d-1)(1-(4d+1)\beta)} \cdot k^{4\beta}  \ell^{4 \beta} \cdot n \cdot \left( \left(\frac{112}{\beta k}\right)^{(1-10 \beta)}  + n^{-\beta /3}\right) \right)^{k\ell}
\end{align*}
Since $k < n^{\beta/4}$ and $\beta < 1/30$, we have that $k^{1-14 \beta} << n^{\beta/3}$, and so the first term in the latter parenthesis dominates.
This implies that,
\begin{align*}
 \E_{\cF} \E_{\cM} \sum_{\sigma \in [n]^{k\ell}} \Ind[\onorm{\cH_{\sigma}} =
0] \leq
\left( c_{d\beta}\ell \cdot \frac{n}{k^{(2d-1) - 8d^2\beta}} \right)^{k\ell}
\end{align*}
where $c_{d\beta}$ is a constant depending on $d$ and $\beta$, and where we have used the fact that $8d^2 \ge 8d^2 - 6d + 4$ for all $d \ge 1$.

\end{proof}

The following lemma we employ in bounding the probability that our blocks from $\cF$ are matched in a way that gives hyperedges with even multiplicity.
We do this via reducing the problem to counting the number of multigraphs with labeled edges in which every subgraph induced by a given label is Eulerian.

\begin{lemma} \label{lem:fixedF}
For every $\cF$ with $\dnorm{\cF} \geq 2k \ell(1-\beta)$, $$  \E_{\cM} \sum_{\sigma \in [n]^{k\ell}} \Ind[\onorm{\cH_{\sigma}} = 0] \leq  
\left(4 \beta k \ell \right)! \cdot n^{k \ell} \cdot
  \left(\left(\frac{112}{\beta k}\right)^{k \ell(1-10 \beta)}   +
  n^{-\beta k \ell/3}\right)$$
\end{lemma}
\begin{proof}
Define a multigraph $\cG$ as follows.
In the multigraph $\cG$, there is a vertex $v_f$ for each {\it distinct} block $f \in \cF$.
There is an edge in $\cG$ for each edge in the matchings $\cM$ between the blocks.
Every choice of {\it pivot} vertices $\sigma \in [n]^k$ corresponds to a labeling of the edges $\sigma : E(\cG) \to [n]$.
For each edge $e \in E(\cG)$ incident at a vertex $v_f \in V(\cG)$, there is a hyperedge in $\cH_{\sigma}$ corresponding to $(\sigma(e),v_f)$.
The hypergraph $\cH_{\sigma}$ is {\it even} if and only if for each pivot vertex $i
\in [n]$, and each vertex $v_f \in V(\cG)$, the number of edges labeled $i$
incident at $v_f$ is even.  This implies that $\sigma^{-1}(i)$ form an
Eulerian subgraph for each $i \in [n]$.  By \prettyref{lem:eulerian1},
the number of such labelings $\sigma : E(\cG) \to [n]$ is at most
$(2|E(\cG)|-
2|V(\cG)|)! \cdot n^{\nicefrac{E(\cG)}{2}- \nicefrac{E_{\oplus}(\cG)}{6}}$.

By definition of the graph $\cG$, $|V(\cG)| = \dnorm{\cF} \geq 2k
\ell(1-\beta)$ and $E(\cG) = 2k \ell$.  Moreover, by applying
\prettyref{lem:even-matching} with $r \leftarrow 2$, $c \leftarrow k$, $M \leftarrow 2\ell$, $\delta \leftarrow 1$, and $E_i \leftarrow \cF_i$, we conclude that the graph $\cG$ has many odd
multiedges with high probability over the choice of $\cM$.  Formally,
$$\Pr[|E_{\oplus}(\cG)| \leq 2\beta k \ell] \leq
\left(\frac{112}{\beta c}\right)^{k\ell(1-10\beta)}$$
Now we are ready to wrap up the proof of the lemma.
\begin{align*}
\E_{\cM} \sum_{\sigma \in [n]^{k\ell}} \Ind[\onorm{\cH_{\sigma}} = 0] 
& \leq \E_{\cM}(2E(\cG)-2V(\cG))!\cdot  n^{\nicefrac{E(\cG)}{2}-
 \nicefrac{E_{\oplus}(\cG)}{6}} \\
 & \leq \E_{\cM} \left(4 \beta k \ell \right)! \cdot n^{k \ell} \cdot
n^{-|E_{\oplus}(\cG)|/6} \\
 & =  \left(4 \beta k \ell \right)! \cdot n^{k \ell} \cdot \E_{\cM}
 n^{-|E_{\oplus}(\cG)|/6} \\ 
 & \leq  \left(4 \beta k \ell \right)! \cdot n^{k \ell} \cdot
 \left( \Pr[|E_\oplus(\cG)|\leq 2\beta k \ell] + n^{-2\beta k
 \ell/6}\right)\\ 
 & \leq  \left(4 \beta k \ell \right)! \cdot n^{k \ell} \cdot
 \left( \left(\frac{112}{\beta k}\right)^{k \ell(1-10 \beta)}  +
 n^{-2\beta k \ell/6}\right) \\
 & \leq  \left( k^{4\beta}  \ell^{4 \beta} \cdot n \cdot
 \left( \left(\frac{112}{\beta k}\right)^{(1-10 \beta)}  +
 n^{-\beta /3}\right) \right)^{k\ell} 
\end{align*}
\end{proof}

Our final lemma of this section is a bound on the number of labelings of a multigraph such that the subgraphs induced by all edge labels are Eulerian, given a bound on the number of multi-edges appearing with odd multiplicity.
\begin{lemma} \label{lem:eulerian1}
Given a multigraph $\cG$, a labeling of its edges $\sigma: E(\cG) \to [n]$ is said to be {\it even}, if the preimage of every label $i$ forms an
Eulerian subgraph (not necessarily connected) of $\cG$.  Specifically, the set of edges $\sigma^{-1}(i) \subseteq E(\cG)$ induce a subgraph where
the degree of every vertex is even.

$$ \card{ \{ \sigma: E(\cG) \to [n] | \sigma \text{ is even }\} } \leq  (2|E(\cG)| - 2
|V(\cG|))! \cdot n^{|E(\cG)|/2 - |E_{\oplus}(\cG)|/6} $$
where $|E_{\oplus}(\cG)|$ is the number of multi-edges with odd multiplicity
within $\cG$.
\end{lemma}

\begin{proof}
We will count the number of {\it even} labelings $\sigma$ as follows:
\begin{compactitem}
\item Pick a unordered partition of the edges of the graph in to
  Eulerian subgraphs.  By \prettyref{claim:eulerparts}, there
  are at most $(2|E(\cG)|-2|V(\cG)|)!$ of them.

\item Assign a label from $[n]$ to each Eulerian subgraph in the
  partition.  The number of labelings is clearly at most $n^t$ where
  $t$ is the number of subgraphs in the partition.  By
  \prettyref{claim:numpartitions}, there are at most $\frac{|E(\cG)|}{2} -
  \frac{|E_{\oplus}(\cG)|}{6}$ subgraphs in any partition. Hence, there are at
  most $n^{\frac{|E(\cG)|}{2} - \frac{|E_{\oplus}(\cG)|}{6}}$ labelings for each
  partition of $\cG$ in to Eulerian subgraphs.
\end{compactitem}
The lemma follows immediately from the
\prettyref{claim:eulerparts} and \prettyref{claim:numpartitions}
which we will show now.

\begin{claim} \label{claim:eulerparts}
The number of unordered partitions of the edges of the graph in to
Eulerian subgraphs is at most $(2|E(\cG)|-2|V(\cG)|)!$.
\end{claim}
\begin{proof}

Let $d_v$ denote the degree of vertex $v \in V(\cG)$.
We can specify a partition of the edges of $\cG$ in to Eulerian
subgraphs, by specifying a sequence of Eulerian traversals whose union
covers all the edges in the graph exactly once.

Consider a vertex $v$.  Any sequence of traversals induces a matching $M_v$
between the edges incident at $v$ -- where $e,e'$ are matched if one
of the traversals goes along $e \rightarrow v \rightarrow e'$.  Furthermore, given a set
of matchings $\{M_v | v\in V(\cG)\}$, it uniquely identifies a set of
  traversals.

Therefore the number of partitions of $E(\cG)$ in to Eulerian subgraphs
is at most

\begin{align*}
 \prod_{v \in V(\cG)} |\# \text{ matchings of edges incident at} v| & \leq  \prod_{v \in V(\cG)} \left( \frac{d_v !}{(d_v/2)!}
\cdot \frac{1}{2^{d_v}} \right) \\
& \leq \prod_{v \in V(\cG)} (d_v -2)! \\
& \leq (\sum_{v \in V(\cG)} (d_v - 2))! \leq \left( 2E(\cG)- 2V(\cG) \right)!
\end{align*}
\end{proof}

\begin{claim} \label{claim:numpartitions}
In any partition of $\cG$ in to Eulerian subgraphs, the number of
partitions is at most $\frac{|\cE(G)|}{2} - \frac{|\cE_{\oplus}(G)|}{6}$
\end{claim}

\begin{proof}
  Suppose $E(\cG) = \cup_{i =1}^{t} E_i(\cG)$ denote a partition of
  $E(\cG)$ into Eulerian subgraphs.  For each edge $e \in E_i(\cG)$ assign a weight
  $w_e = \frac{1}{|E_i(\cG)|} $.  By definition of the weights, we have
$$ \sum_{e \in E(\cG)} w_e = t \mper$$
Note that $w_e \leq \frac{1}{2}$ for all $e \in E(\cG)$, since each subset $E_i(\cG)$
contain at least two edges by virtue of being Eulerian.  Moreover,
$w_e = \frac{1}{2}$ if the edge $e$ belongs to an Eulerian subgraph $E_i
(\cG)$ with exactly two edges.  In particular, $E_i(\cG) = \{e,e'\}$ where
$e$ and $e'$ form a $2$-cycle.  For every multiedge $(a,b)$ with odd
multiplicity, at least one of its edges has $w_{e} \leq \frac{1}{3} = \frac{1}{2} -
\frac{1}{6}$.

Therefore we conclude that
\[
    t = \sum_{e \in E(\cG)} w_e \leq \sum_{e \in E(\cG)} \frac{1}{2} -
\sum_{(a,b) \in E_{\oplus}(\cG) } \frac{1}{6} = \frac{|E(\cG)|}{2} - \frac{|E_{\oplus}(\cG)|}{6}\mper
\]
\end{proof}
These claims together finish the proof.
\end{proof}
\subsection{Useful combinatorial lemmas}
\label{sec:combinatorial-lemma}
\label{sec:extra-lemmas}

Define an \emph{$r$-grouping} to be a partition of a set of size $c \cdot r$ into $c$ subsets of size $r$.
The following lemma bounds the probability that, given a multiset with many distinct elements, an $r$-grouping of the elements results in few $r$-sets with odd multiplicity.
We rely on this lemma in our injective tensor norm upper bounds, to bound the probability that a hypergraph sampled from a simple graph has the evenness property.

\restatelemma{lem:even-matching}
\begin{proof}
We will refer to each $s \in [N]$ as a ``type''.
Call a type $s \in[N]$ { \it infrequent} if the number of occurrences of $s$ within $\cup_{i} E_i$ is nonzero but at most $8$.

Suppose a type $s \in [N]$ appears exactly once in the sets $E_1,\ldots,E_M$,then irrespective of the choice of the grouping, the group involving $s$ appears exactly once.
If there are more than $r\delta M c$ types that appear exactly once then,
\[
    \Pr[|\oplus_{i} G_i| \leq \delta M c] = 0,
\]
and the lemma holds.
Henceforth, we assume that all but $r\delta Mc$ types appear at least twice.

Call a type to be {\it frequent} if it occurs more than $8$ times within $\cup_{i} E_i$.
Out of the $rMc$ elements, at most an $8 \beta$ fraction are occurrences of {\it frequent} types.
Otherwise, the number of distinct types would be less than $rMc\left( (1-8\beta)/2  + 8 \beta/ 8 + \delta \right) < \frac{rMc}{2} (1-\beta)$.

Moreover, this implies that the number of distinct {\it frequent} types is at most $ 8 \beta r M c/8 \leq \beta r M c$.
Finally, the number of distinct infrequent types is at least $\frac{rc M}{2}(1-\beta) - \beta rMc \geq \frac{rM c}{2} \cdot(1-3 \beta)$.

Let us sample uniform random $r$-groupings $\{G_i\}_{i \in [M]}$ one group at a time.
Specifically, we will sample groups $g_1, \ldots,g_{cM}$ where $G_{i} = \{g_{(i-1)c+1}, \ldots,g_{ic}\}$, one group at a time.
We sample the $i^{th}$ grouping $G_i$ as follows:
\begin{compactitem}
\item For $j = 1$ to $c$
\item  Pick the element $s$ with the smallest number of ungrouped
  occurrences left within $\cup_{j=i}^{M}E_j$ (breaking ties
  lexicographically).
\item Sample the group  $g_{(i-1)c+j}$ by picking the remaining $r-1$
  elements uniformly at random from ungrouped elements in $E_i$
\end{compactitem}
It is clear that the above sampling procedure picks a uniformly random
grouping $\{G_i\}_{i \in [M]}$.

We will refer to the groups picked at any stage to be {\it configuration}.  So, the configuration at the end of $i^{th}$ stage is $\cE_i \defeq \{g_1,\ldots,g_i\}$.
Given a current configuration $\cE_i$, there is a unique element $s(\cE_i)$ that will be grouped in the next step.
A configuration $\cE_i$ is said to be {\it critical} if
\begin{compactenum}
\item $s(\cE_i)$ is its final ungrouped occurrence of an {\it infrequent} type.
\item All previous occurrences of $s(\cE_i)$ has been grouped with {\it infrequent} types.
\item There are at least $ \beta rc$ ungrouped elements within the current multiset $E_j$ that is being grouped.
\end{compactenum}

\begin{claim} \label{claim:criticalonallpath}
    For every sequence of random choices, the sampling procedure encounters at least $\frac{c M}{2} \cdot (1-(4r+1)\beta)$ {\it critical} configurations.
\end{claim}
\begin{proof}
There are at most $8\beta rc M$ occurrences of { \it frequent} types.
This implies that among the $\frac{rMc}{2} (1-3\beta)$ infrequent labels, at least $\frac{rMc}{2} (1-3\beta) - (8\beta r M c)(r-1) \geq  \frac{rMc}{2}(1-(4r - 1)\beta)$ are grouped only with infrequent types.

For each of these $\frac{r Mc}{2}(1-(4r-1)\beta)$ types there is one final ungrouped occurrence.
Even assuming we match all these final occurrences among themselves,  both conditions (1) \& (2) are met at least $\frac{Mc}{2}(1-(4r-1)\beta)$ times during the sampling procedure.

Finally, there are at most $\beta c$ groups that are picked among the final $\beta rc$ elements within the sets $E_i$.
Therefore, for at least $\frac{Mc}{2}(1 - (4r-1)\beta) - \beta Mc \ge \frac{Mc}{2}(1 - (4r+1)\beta)$ steps,  $\cE_i$ is a critical configuration.
\end{proof}

Define random variables $\{Z_i\}_{i \in [m]}$ as follows:
\[
    Z_i \defeq \Ind[ g_i \text{ is final occurrence of an odd group in } \cup_{j} G_j] \mper
\]
By definition, we have
$$ \card{ \oplus_{j} G_j} = \sum_{i \in [cM]} Z_i$$

Set $\alpha = \left(\frac{56}{\beta c}\right)^{r-1}$.  In order to obtain concentration bounds on $\sum_{i \in
  [cM]}Z_i$ we will bound $ \E[\alpha^{\sum_{i} Z_i}]$.

\begin{claim} \label{claim:maxalpha}
For all $ \alpha \leq \left( \frac{56}{\beta c} \right)^{c-1}$, for all $t \in [cM]$ and all critical configurations $\cE_t$,
$$\E[\alpha^{\sum_{i=t}^{c M} Z_{i+1}} | \cE_t] \leq 2\alpha \cdot \max_{\cE_{t+1} | \cE_t} \E[\alpha^{\sum_{i=t+1}^{cM}Z_{i+1}} | \cE_{t+1}] $$
where the maximum is taken over all feasible configurations $\cE_{t+1}$ from $\cE_t$.
\end{claim}

\begin{proof}

At a critical configuration $\cE_t$, the next group is the last occurrence of $s(\cE_t)$.
Recall that $s(\cE_t)$ is infrequent in that it has at most $7$ previous occurrences.
Moreover, each of its previous occurrences is grouped to an infrequent type (appearing less than 8 times).

There are at least $\beta rc$ ungrouped elements from which the remaining $r-1$ elements of the group are chosen.
For all but  at most $(56)^{r-1}$ group choices, the group contains a type $s'$ such that this is the first occurrence of $s$ with $s'$ in a group.

Therefore, for all but at most $(56)^{r-1}$ choices, the group sampled is its first and only occurrence.
In particular, this implies that for a critical configuration $\cE_t$,
$$ \Pr[Z_{t+1} = 0 | \cE_t] \leq \frac{(56)^{r-1}}{\binom{\beta
    rc}{r-1}} \leq \left( \frac{56}{\beta c} \right)^{r-1}$$
Finally, we have
\begin{align*}
\E[\alpha^{\sum_{i = t+1}^{cM}}] & = \Pr[Z_{t+1} = 1| \cE_t] \cdot \alpha \cdot \E[\alpha^{\sum_{i=t+2}^{cM} Z_{i}} | \cE_{t}, Z_{t+1} = 1] + \Pr[Z_{t+1} = 0| \cE_t] \cdot
  \E[\alpha^{\sum_{i=t+2}^{cM} Z_{i}} |\cE_{t}, Z_{t+1} = 0 ] \\
& \leq \alpha \cdot  \E[\alpha^{\sum_{i=t+2}^{cM} Z_{i}} | \cE_{t}, Z_{t+1} = 1] + \alpha \cdot  \E[\alpha^{\sum_{i=t+2}^{cM} Z_{i}} |\cE_{t}, Z_{t+1} = 0 ] \\
& \leq 2\alpha  \cdot \max_{\cE_{t+1} | \cE_t} \E[\alpha^{\sum_{i=t+1}^{cM}Z_{i+1}} | \cE_{t+1}]
\end{align*}

\end{proof}

Combining \prettyref{claim:maxalpha} and \prettyref{claim:criticalonallpath}, we have that

$$ \E[\alpha^{\sum_{i \in [cM]}z_i}] \leq (2\alpha)^{Mc (1-(4r + 1)\beta)/2} \mcom$$

which yields the following concentration bound for all $\delta > 0$,

$$ \Pr[ \card{ \oplus_{i} G_i } \leq \delta c M] <  (2\alpha )^{(1-(4r+1)\beta - 2\delta) cM/2} \leq \left( \frac{112}{\beta k} \right)^{ (1-(4r+1)\beta-2\delta) (r-1) Mc/2}$$

\end{proof}

The lemma below shows that in a simple graph formed by matchings with the evenness property, there cannot be too many more distinct vertices than distinct edges.
\restatelemma{lem:vtx-edge}
\begin{proof}
    We first prove the following claim:
    \begin{claim*}
	If each labeled edge appears with multiplicity exactly $2$, then $L \le m + c$.
    \end{claim*}
    \begin{proof}
	In this case, there are exactly $2m$ edges and exactly $2m$ vertices.
We proceed by induction on $c$ and $m$.
In the base case, we have $c = 1$ component with $2$ vertices, in which case we have at most $2$ distinct labels on the vertices, confirming the claim.

Assuming the claim for $c\ge 1$ components and $2m\ge 2$ vertices, consider an instance on $2m + 2$ vertices.
If all labels appear $\ge 2$ times, we are done, since there are $2m + 2$ vertices and thus at most $L\le m+1$ labels.
Otherwise, locate a vertex $v$ whose label has multiplicity $1$.

If $v$ is in a cycle of length $2$, remove $v$ and its neighbor from the graph, obtaining a smaller instance with $L'$ labels, $c'$ components, and $m'$ distinct edge types, with $L' + 2 \ge L$, $c' = c-1$, and $m' = m$.
By the induction hypothesis, $L' \le m' + c' = m + c-1$, and therefore $L \le m+1 + c$, as desired.

If $v$'s cycle has length $> 2$, both $v$'s vertex neighbors must have the same label in order for the edges incident on $v$ to appear twice.
We remove $v$ and identify its neighbors, obtaining an instance with $L' + 1 = L$, $m' = m$, $c' = c$.
Appealing to the induction hypothesis, we have $L' \le m' + c$, from which we conclude that $L \le m + 1 + c$, as desired.
\end{proof}
Now, we reduce our lemma to the above case.
Say an edge appears with even multiplicity $ \mu > 2$, and that the labels of the edge are $(a,b) \in [n]^2$.
We will remove the occurrences of this edge, and put the graph segments back together.
When we remove all occurrences of the edge $(a,b)$, we get $3$ kinds of graph segments: paths from $a$--$b$, paths from $a$--$a$, and paths from $b$--$b$.
Since $a$, $b$ each have to appear $\mu$ times, we can form a matching between segments of type $a$--$b$, gluing them together at the $a$ endpoint to get a $b$--$b$ segment.
Now, we make one cycle by gluing together $a$--$a$ segments, and a separate cycle by gluing together $b$--$b$ segments.
Our number of distinct edges has decreased by $1$, and our number of cycles has increased by at most $1$, since we broke up at least one cycle to remove the edge $(a,b)$.
We recursively apply this process to our instance, until we reach an instance in which there are only edges of multiplicity $2$, never increasing the quantity $m+c$.
In conjunction with our above claim, the conclusion follows.
\end{proof}

\section{Refuting Random $k$-XOR Instances}\label{sec:lin}

In this section, we give our algorithm for refuting random $k$-XOR instances.
In \pref{sec:lin-even}, we describe the algorithm for even $k$; in \pref{sec:lin-odd}, we describe the algorithm for odd $k$.
We first recall the problem:

\begin{definition}[Random $k$-XOR with density $\alpha = pn^{(k-1)/2}$]
    A \emph{random instance of $k$-XOR with density $\alpha = pn^{(k-1)/2}$} is a formula $\Phi$ on $n$ variables $x \in \{\pm 1\}^n$, sampled so that for each $S \in [n]^k$:
    \begin{compactitem}
    \item Independently with probability $p$, add constraint $C_S: \prod_{i \in S} x_i = \eta_S$, for $\eta_S$ a uniformly random Rademacher variable.
    \item Otherwise, with probability $1-p$, add no constraint.
    \end{compactitem}
    We let $m \approx pn^{k}$ be the number of constraints, and for any assignment $x \in \{\pm 1\}^n$, $P_{\Phi}(x)$ is the fraction of constraints satisfied by $x$.
\end{definition}

\begin{problem}[Strongly refuting random $k$-XOR]
    Given a random $k$-XOR instance $\Phi$, certify with high probability over the choice of $\Phi$ that for all assignments $x \in \{\pm 1\}^n$,
    \[
	P_{\Phi}(x) \le \frac{1}{2} + \delta + o(1),
    \]
for some constant $\delta \in [0,\half)$, where $P_{\Phi}(x)$ is the fraction of $\Phi$'s constraints satisfied by $x$.
\end{problem}

As described in \pref{sec:tech-over}, there is a natural random order-$k$ tensor that we can identify with any $k$-XOR instance $\Phi$.
Given a $k$-XOR instance $\Phi$ with constraints $C_1,\ldots,C_m$, form the tensor $\bT_{\Phi}$ as follows: for each constraint $C_i: \prod_{j \in S_i} x_i = \eta_{S_i}$, set the entry $\bT_{S_i} = \eta_{S_i}$; in all other entries place a $0$.
We then have that for any assignment $x \in \{\pm 1\}^n$,
\[
    \iprod{\bT_{\Phi},x^{\otimes k}} = \sum_{i \in [m]} \eta_{S_i} \cdot \prod_{j \in S_i} x_j = m \cdot \left( P_{\Phi}(x) -  \frac{1}{2}\right)
\]
That is, the inner product $\iprod{\bT_{\Phi}, x^{\otimes k}}$ gives the difference between the number of constraints $x$ satisfies and the number of constraints $x$ violates.
Our strong refutation algorithm will be based on showing that
\begin{equation} \label{eq:xorref}
    \left|\Iprod{\bT_{\Phi}, x^{\otimes k}}\right| \le (\delta + o(1)) \cdot m \quad \forall x \in \{\pm 1\}^n,
\end{equation}
for a constant $\delta$ arbitrarily close to $0$.

From \pref{eq:xorref}, it is clear that a good bound on $\|\bT_{\Phi}\|_{inj}$ would give a refutation algorithm, and so we could hope that our algorithms for bounding tensor norms would suffice.
However, when the probability of sampling a constraint $p \le n^{-k/2}$, the tensor $\bT_{\Phi}$ becomes sparse enough that its norm is maximized by sparse vectors, so that $\|\bT_{\Phi}\|_{inj} \approx 1$.
We are only interested in balanced vector $x \in \{\pm 1\}^n$, and so this is a poor upper bound--it will only let us certify that $P_{\Phi}(x) \le \frac{1}{2} + \frac{n^{k/2}}{m} \ge 1$.

So our algorithm for the case of $k$-XOR is almost identical to our algorithm for bounding tensor norms, but with an additional twist to get rid of the sparse vectors.
We form our certificate as we did in the tensor norm algorithm: we flatten $\bT_{\Phi}$ to a matrix $T$, then take the $d$th Kronecker power of $T$, and we average over rows and columns corresponding to permutations of the same index set.
But now, there is one additional step: {\bf we delete any row or column indexed by a multiset $S \in [n]^{kd/2}$ which contains an element $i$ with multiplicity greater than $O(\log n)$.}

It is not difficult to see why this should help: supposing we started with a sparse vector, say the standard basis vector $e_1 \in \R^n$, this will ensure that $e_1^{\otimes kd/2}$ has $0$ projection onto our matrix.
On the other hand, the choice of $O(\log n)$ as our upper bound on the multiplicity makes sense, since we are eliminating an $o(1)$-fraction of the Frobenius norm of the certificate in this way, even when $d \ge n^{1/2}$--if we were to delete all rows and columns in which an element appears with multiplicity $\ge 2$, then once $d = n^{1/2}$ we would be deleting a constant fraction of the rows and columns, by the birthday paradox.

This introduces some technicalities in the analysis--in particular, once we delete these rows and columns, it is no longer obvious that we are working with a valid relaxation of $\iprod{\bT_{\Phi}, x^{\otimes k}}$ over $x \in \{\pm 1\}^n$.
But as before, the main theorems of this section will have to do with bounding the norm of our matrix certificate--arguing that the matrix certificate is valid will be straightforward.

We begin by detailing our algorithm for even $k$, then give the somewhat more involved analysis for odd $k$ (the additional complication introduced by the lack of a natural matrix flattening for odd-order tensors).

\subsection{Even $k$-XOR} \label{sec:lin-even}

We begin by describing our matrix certificate for this case, and establishing an upper bound on its norm--as mentioned above, this is the main result of this section.
Later, in \pref{sec:valid-cert}, we will show how to use the matrix to get a valid certificate.

\begin{algorithm}[Even $k$-XOR Certificate at level $d$]\label{alg:lin-sub}~\\
    {\bf Input:} A $k$-XOR instance $\Phi$ for even $k$ on $n$ variables and $m$ clauses.
    Parameters $d \in \N$.
    \begin{compactenum}
    \item Form the tensor $\bT_{\Phi}$ from $\Phi$ as described above (see \pref{eq:xorref}).
    \item Take the natural $n^{k/2} \times n^{k/2}$ matrix flattening $T$ of $\bT_{inj}$, and take the Kronecker power $T^{\otimes d}$
	\item Letting $\hat \cS_{dk/2}$ be the set of all permutation
		matrices that perform the permutations
		corresponding to $\cS_{dk/2}$ on the rows and columns of $T$, form
	\[
	    C_{(d)}\defeq \E_{\Pi,\Sigma \in \hat\cS_{dk/2}}\left[\Pi T^{\tensor d} \Sigma \right]\mper
	\]
    \item Zero out any row or column of $C_{(d)}$ indexed by a multiset in $[n]^{kd/2}$ containing more than $10 \log n$ copies of any $i \in [n]$.
    \end{compactenum}
    {\bf Output:} The value $\|C_{(d)}\|$.
\end{algorithm}
The following theorem gives a bound on the value output by \pref{alg:lin-sub}.
\begin{theorem}
    \torestate{\label{thm:lin-upper}
    Let $k,n,d,k \in \N$, so that $d\log n \ll n$, $k$ is even.
    Let $\Phi$ be a random instance of $k$-XOR on $n$ variables with $\Theta(p n^k)$ clauses (so each constraint is sampled uniformly and independently with probability $p$).
    Let $C_{(d)}$ be the matrix formed from the instance $\Phi$ as described in \pref{alg:lin-sub}.
    Then if $p \cdot d^{(k/2-1)}n^{k/2} > 1$, there is a constant $c_k$ depending on $k$ such that with high probability,
    \[
	\|C_{(d)}\|^{1/d} \le
	\left( c_k \log^{2k} n \cdot \frac{p^{1/2}n^{k/4}}{d^{(k-2)/4}}\right) \mper
    \]
}
\end{theorem}

We will prove the theorem below, in \pref{sec:lin-norm-sec}.
First, we will see how to use this certificate, with the deleted high-multiplicity rows and columns, to strongly refute $k$-XOR instances.

\subsubsection{Validity of certificate with deleted rows and columns}\label{sec:valid-cert}
When we zero out the high-multiplicity rows and columns in \pref{alg:lin-sub},
\[
    \Iprod{\bT_{\Phi}, x^{\otimes k}}^d  = (x^{\otimes dk/2})^\top \left( C_{(d)} + C_{\ge} \right) x^{\otimes dk/2},
\]
where $C_{\ge}$ is a matrix containing only the zeroed out rows and columns.
So our upper bound on $\|C_{(d)}\|$ from \pref{thm:lin-upper} is not enough.
It is not hard to bound the $\ell_1$-norm of $C_{\ge}$.
However, because for our values of $p$, $\|C_{(d)}\|$ is close to $0$, the $\ell_1$-norm bound is too costly when we try to bound $\iprod{\bT_{\Phi}, x^{\otimes k}}$.
For this reason, we will work with $P_{\Phi}(x)$, the fraction of satisfied constraints, which is bounded away from $0$.
We will relate $(P_{\Phi}(x))^d$ to the matrix norms of $C_{1},\ldots C_{d}$.

Lets write
\[
    P_{\Phi}(x) = \E_{i \sim [m]} \left[P_i(x) \right] = \E_{i\sim[m]} \left[\frac{1}{2} (1+C_i(x))\right],
\]
where $P_1(x),\ldots,P_m(x)$ are the $0-1$ valued predicates of the instance $\Phi$, and $C_1(x),\ldots,C_m(x)$ are the $\pm 1$-valued predicates of the instance $\Phi$.
We have that
\[
    (P_{\Phi}(x))^d = \E_{i_1,\ldots,i_d \sim [m]^d}\left[\prod_{\ell = 1}^d P_{i_\ell}(x)\right]\mper
\]
We will prove that the quantity above is not changed very much if we remove sets $i_1,\ldots,i_d$ corresponding to high-multiplicity rows and columns.

\begin{proposition}\label{prop:formula}
    Let $\Phi$ be a random $k$-XOR formula in which each clause is sampled independently with probability $p$.

    Let $\cC^d_{low} \subset [m]^d$ be the set of all ordered multisets of clauses $C_{i_1},\ldots,C_{i_d}$ from $\Phi$ with the property that if we form two multisets of variables $I,J \in [n]^{dk/2}$ with $I$ containing the first $k/2$ variables of each $C_{i_\ell}$ and $J$ containing the last $k/2$ variables of each $C_{i_\ell}$, then $I,J$ are both low-multiplicity multisets, in that both have no element of $[n]$ with multiplicity $\ge 100\log n$.

    Suppose that no variable appears in more than $m_{\max}$ clauses.
    Then if $d \ll n$ and
$d k m_{\max} < 200 \epsilon m \log n$,
    \[
	P_{\Phi}(x) \le \left(\E_{i_1,\ldots,i_d \sim \cC_{low}^d}\left[\prod_{\ell=1}^d P_{i_\ell}(x)\right]\right)^{1/d} + \epsilon
    \]
for all $x \in \{\pm 1\}^n$ with high probability.
    Furthermore when $p \ge 200 \frac{\log n}{n^{k-1}}$, we have that $\epsilon = o(1)$ with high probability.
\end{proposition}
\begin{proof}
    Let $m_{\max}$ be an upper bound on the number of clauses any variable $x_i$ appears in the instance $\Phi$.
    We sample a uniform element  $\cC\sim\cC_{low}^d$, $\cC = C_1,\ldots,C_d$ in the following way:
    \begin{compactitem}
	\item For $t = 1,\ldots,d$:
	    Let $\cA_t \subset \Phi$ be the set of clauses such that for any $C' \in \cA$, the multiset $C_1,\ldots,C_{t-1}, C'$ is not excluded from $\cC_{low}^{t}$.
		Choose a uniformly random $C \sim \cA_t$ and set $C_t := C$, adding $C$ to $\cC$.
    \end{compactitem}
	This sampling process clearly gives a uniformly random element of $\cC^d_{low}$.
    \begin{claim}\label{claim:exclusions}
	At step $t+1$ there are at least $m - t\frac{k\cdot m_{\max}}{200\log n}$ clauses that can be added.
    \end{claim}
    \begin{proof}
	In order to exclude any variable $i \in [n]$, we must add at least $200\log n$ copies of $i$.
	Further, to exclude $\ell$ distinct variables in $[n]$, at least $200\log n$ copies of each variable, for a total of $200\ell\log n$ variables, which requires adding at least $200\ell\log n/k$ clauses.
	If $\ell$ distinct variables are excluded, then at most $\ell \cdot m_{\max}$ clauses are excluded.
	The claim now follows.
    \end{proof}

    Now, define the random variable $X_{t} = \prod_{j=1}^t P_{j}(x)$ to be the value of $x$ on $\cC_t$.
    We apply \pref{claim:exclusions}, along with the observation that the total number of satisfied clauses can only drop by $1$ for each clause removed regardless of the assignment $x$, to conclude that
    \begin{align*}
	\E[X_{t+1}|C_1,\ldots,C_t]
	\ge \left(P_{\Phi}(x) -t \frac{k\cdot m_{\max}}{m \cdot 200 \log n}\right) \cdot X_{t}
    \end{align*}
    From this we have that
    $ \E[X_t]
    \ge \left(P_{\Phi}(x) -  t\frac{k m_{\max}}{200 m\log n}\right) \cdot \E[X_{t-1}]$
    from which we have that as long as $d k m_{\max} \le \epsilon 200 m\log n$,
    \[
	\E[X_d] \ge \prod_{t=1}^{d} \left(P_{\Phi}(x) - t\frac{k m_{\max}}{200 m\log n}\right)
	\ge \left(P_{\Phi}(x) - \epsilon \right)^d \mper
    \]
Which by definition of $X_d$ gives us our first result.

    Now, we can establish that $dk m_{\max} \le \epsilon \cdot 200 m \log n$ with high probability.
    A Chernoff bound implies that when $p \ge 200\log n/n^{k-1}$, $ 2 pn^k \ge m \ge pn^k/2$ with probability at least $1-2\exp(-pn^k/8)$, and that $m_{\max} \le 2pn^{k-1}$ with probability at least $1 - \exp(-pn^{k-1}/2)$, and so by a union bound and using the assumption that $p n^{k-1} \ge 200\log n$, we have our result by taking $\epsilon = \Theta( 1/\log n)$.
\end{proof}

Now, we will relate the right-hand-side of \pref{prop:formula} to the matrices from \pref{alg:lin-sub}.
We recall that given a $k$-XOR instance $\Phi$, $\cC^d_{low} \subset [m]^d$ is the set of all ordered multisets of clauses $C_{i_1},\ldots,C_{i_d}$ from $\Phi$ with the property that if we form two multisets of variables $I,J \in [n]^{dk/2}$ with $I$ containing the first $k/2$ variables of each $C_{i_\ell}$ and $J$ containing the last $k/2$ variables of each $C_{i_\ell}$, then $I,J$ are both low-multiplicity multisets, in that both have no element of $[n]$ with multiplicity $\ge 100\log n$.
By \pref{prop:formula},
\begin{align}
    \left(P_{\Phi}(x) - o(1)\right)^d
    &\le \E_{i_1,\ldots,i_d \sim \cC_{low}^d}\left[\prod_{\ell=1}^d P_{i_\ell}(x)\right]
    ~=~ \E_{i_1,\ldots,i_d \sim \cC_{low}^d}\left[\prod_{\ell=1}^d \frac{1}{2}(1+C_{i_\ell}(x))\right] \nonumber
    \intertext{and expanding the product on the right and applying the symmetry of the uniform distribution on $\cC_{low}^{d}$,}
    &= \left(\frac{1}{2}\right)^d \sum_{j = 0}^d \binom{d}{j} \E_{i_1,\ldots,i_j \sim \cC_{low}^j}\left[\prod_{\ell=1}^d C_{i_\ell}(x)\right] \nonumber
    \intertext{and by definition, for any assignment $x \in \{\pm 1\}^n$,}
    &= \left(\frac{1}{2}\right)^d \sum_{j = 0}^d \binom{d}{j} \frac{1}{|\cC_{low}^j|} \cdot (x^{\otimes jk/2})^{\top} C_{(j)} (x^{\otimes jk/2}),\nonumber
    \intertext{
	where $C_{(j)}$ is the matrix output by \pref{alg:lin-sub} when $d \leftarrow j$. We won't get a good bound on $\|C_{(j)}\|$ when $j$ is too small, but we can take}
    &=
    \left(\frac{1}{2}\right)^d \sum_{j = 0}^t \binom{d}{j}
    ~+~ \left(\frac{1}{2}\right)^d \sum_{j = t+1}^d \binom{d}{j} \frac{1}{|\cC_{low}^j|} \cdot (x^{\otimes jk/2})^{\top} C_{(j)} (x^{\otimes jk/2})\mper \label{eq:upbd}
\end{align}
We'll take $t = \alpha \cdot d$ for some small constant $\alpha$, so that the sum on the left is small, and the sum on the right we will bound by applying \pref{thm:lin-upper}, our upper bound on $\|C_{(j)}\|$.
We will also need a bound on $|\cC_{low}^j|$, which we can easily get by modifying our proof of \pref{prop:formula}:
\begin{lemma}\label{lem:sizeofmultilinear}
    If $\Phi$ is a random $k$-XOR instance on $n$ variables with $m$ clauses such that no variable participates in more than $m_{\max}$ clauses, then so long as $d\ll n$ and $d k m_{\max} \le 200 \epsilon m \log n$,
 \[
     |\cC_{low}^d| \ge (1-\epsilon)^dm^d\mper
 \]
 Furthermore, when $p \ge 200\frac{\log n}{n^{k-1}}$, we can take $\epsilon = o(1)$ with high probability.
\end{lemma}
The proof proceeds exactly as the proof of \pref{prop:formula}, but instead of bounding the decrease in the value as each clause is added, one bounds the probability that a clause is chosen which will make the multiplicity of some index too high.

We are now ready to prove that computing the norm of $O(d)$ matrices $C_{(\alpha d)},\ldots,C_{(d)}$ will give us a strong refutation algorithm for random $k$-XOR.
This concludes the proof of the refutation theorem, modulo the proof of the $C_{(d)}$ matrix norm bound from \pref{thm:lin-upper}, which we give in the next subsection.
\begin{theorem}\label{thm:refutation}
    Let $k$ be even, and let $d \ll n$.
    Then there is an algorithm that certifies with high probability that a random $k$-XOR instance has value at most $\frac{1}{2} + \gamma +o(1)$ for any constant $\gamma >0$ at clause density $m/n = \tilde O\left(\frac{n^{k/2-1}}{d^{(k/2-1)}}\right)$ (where the $\tilde O$ hides a dependence on $\gamma$ and $k$) in time $n^{O(d)}$.
\end{theorem}

\begin{proof}
    Define $\beta := \frac{dkm_{\max}}{200m\log n}$, where $m_{\max}$ is the maximum number of clauses any variable participates in.
    By \prettyref{prop:formula} and the proceeding calculations culminating in \pref{eq:upbd},  with high probability over the
	choice of the instance $\Phi$, for any $x \in \{\pm 1\}^n$,
    \begin{align*}
	\left(P_{\Phi}(x) - \beta\right)^d
	&\leq
    \frac{1}{2^d} \sum_{j = 0}^t \binom{d}{j}
    ~+~ \frac{1}{2^d} \sum_{j = t+1}^d \binom{d}{j} \frac{1}{|\cC_{low}^j|} \cdot (x^{\otimes jk/2})^{\top} C_{(j)} (x^{\otimes jk/2})\mper
	\intertext{Setting $t = \delta d$ for some $\delta < 1$,}
	&\leq \frac{1}{2^d} \cdot \left(\sum_{j = 1}^{\delta d} \binom{d}{j}\right)
	+ \sum_{j = \delta d+1}^{d} \frac{\binom{d}{j}}{2^d} \cdot \frac{1}{|\cC_{low}^j|} \cdot (x^{\otimes jk/2})^{\top} C_{(j)} (x^{\otimes jk/2})
    \end{align*}
	For the terms in the right-hand sum, we can apply our bound on $|\cC_{low}^j|$ from \pref{lem:sizeofmultilinear} and the fact that $\|x\| = n^{1/2}$, to conclude that
    \begin{align*}
	\frac{1}{|\cC_{low}^j|} \cdot (x^{\otimes jk/2})^{\top} C_{(j)} (x^{\otimes jk/2})
    &\le \frac{\|x\|^{jk}}{|\cC^j_{low}|} \cdot \|C_{(j)}\|
    = \frac{n^{jk/2}}{|\cC^j_{low}|}\cdot \|C_{(j)}\|\\
    & \leq   \frac{n^{jk/2}}{(1-\beta)^j m^j} \cdot \left(\frac{n^{k/2}p \cdot c_k \log^{2k}{n}}{j^{(k/2-1)}}\right)^{j/2}   \qquad \text{ (by \prettyref{lem:sizeofmultilinear} and \prettyref{thm:lin-upper})} \\
    & \leq   \frac{n^{jk/2}}{(1-\beta)^j (0.1 pn^{k})^j} \cdot \left(\frac{n^{k/2}p \cdot c_k \log^{2k}{n}}{j^{(k/2-1)}}\right)^{j/2}   \qquad \text{ (since $m = \Theta(pn^k)$ w.h.p.)} \\
	& \leq  \left(\frac{c_k' \log^{2k}{n}}{(1-\beta)^2pn^{k/2} \cdot j^{k/2-1}} \right)^{j/2}\mper
  \end{align*}
	  for some constant $c'_d$, where the second inequality holds with high probability from the conditions of \pref{thm:lin-upper}, and so also holds with high probability simultaneously for all $j\in[\delta d,d]$ by a union bound.

The term comprised of the sum of binomial coefficients is at most
\[
    \frac{1}{2^d}\sum_{j=0}^{\delta d} \binom{d}{j}
    \le 2^{(H(\delta)-1)d},
\]
where $H(\cdot)$ is the binary entropy function, $H(\delta) = -\delta \log_2 \delta - (1-\delta)\log_2(1-\delta)$.

    Since the coefficients of the $C_{(j)}$ terms sum to $< 1$, we have that for some $\alpha \in [\delta , 1]$,
\[
    P_{\Phi}(x) -\beta
    \le \left(2^{(H(\delta) - 1)d} + \left(c_k'\frac{\log^{2k}n}{(1-\beta)^2pn^{k/2} (\alpha d)^{k/2-1}}\right)^{\alpha d/2}\right)^{1/d}
\]
Now, for $p\gg \tilde O(n^{-k/2}d^{-(k/2-1)})$, where the $\tilde O$ hides a dependence on $k$ and $\alpha > \delta$, $\beta = o(1)$ and the latter quantity is $o(1)$.
	Thus, for sufficiently large $n$ the term in the parenthesis is at most $(1+o(1))2^{H(\delta) - 1}$, which we can take to be a constant arbitrarily close to $\frac{1}{2}$ by choosing sufficiently small constant $\delta$.
	  We can certify this bound in time $n^{O(d)}\cdot d$, by running \pref{alg:lin-sub} to compute $\|C_{(j)}\|$ for each $j\in[\delta d,d]$.
\end{proof}

\subsubsection{Bounding the even certificate spectral norm}\label{sec:lin-norm-sec}

Here, we prove the norm bound on the matrix $\|C_{(d)}\|$ given in \pref{thm:lin-upper}, the main theorem of this section.

\restatetheorem{thm:lin-upper}

The proof is similar to that of \pref{thm:D-upper}, except that, because the moments of the entries of $\bT_{\Phi}$ depend on $p$, and because we rely on getting an accurate bound in terms of $p$, our counting arguments have to be much more precise.
So we require stricter, specialized analogues of our even simple graphs count (\pref{prop:simple-graph-bound-1}) and our even hypergraph sampling probability (\pref{lem:matching-columns}).
\begin{proof}
    We will apply the trace power method (\pref{prop:tpm}) to $C_{(d)}$, for which it suffices to obtain an upper bound on $\E[\Tr((C_{(d)} C_{(d)}^\top)^\ell)]$.
    We recall from \pref{sec:tech-over} our interpretation of the $(S,T)$th entry of $C_{(d)}$ as the average over all $k$-hypergraph matchings between two multisets $S,T \in [n]^{dk/2}$; additionally, now by construction we can restrict our attention to $S,T$ which do not have more than $R \defeq 100\log n$ copies of any one vertex (since those rows/columns are zeroed out).
    For convenience, we say such sets are $R$-multilinear.

    We also recall that the trace gives us a sum over all $R$-multilinear vertex configurations consisting of sets $S_1,\ldots,S_{2\ell} \in [n]^{dk/2}$, and for each vertex configuration an average over all choices of sequences of hypergraph matchings.
    Let the set of all valid $R$-multilinear vertex configurations be denoted $\cV_R$, and let the set of all hyperedge matching sequences be denoted $\cH$.
    For $H\in\cH$ and $V \in \cV_R$, denote by $(V,H)$ the hypergraph given by the hyperedges $H$ on the vertex configuration $V$.
    Applying the above observations, and recalling that we have assembled $C_{(d)}$ from the random tensor $\bT := \bT_{\Phi}$, we have that
    \begin{align*}
	\E[\Tr(C_{(d)} C_{(d)}^\top)^{\ell}]
	&= \sum_{V \in \cV_R} \E_{H\in\cH} \left[ \prod_{(i_1,\ldots,i_d) \in (V,H)} \bT_{i_1,\ldots,i_d}\right]\mcom
    \end{align*}
    The expectation of each product is $0$ if any hyperedge in $(V,H)$ appears with odd multiplicity, and is $p^M$ if exactly $M$ distinct hyperedges  appear in $(V,H)$.
    Thus,
    \begin{align}
	\E[\Tr(C_{(d)} C_{(d)}^\top)^{\ell}]
	&\le \sum_{V \in \cV_R} \sum_{M=1}^{d\ell} p^M \cdot \E_{H \in \cH}\left[\Ind((V,H) ~\text{even}) \cdot \Ind((V,H)~\text{has}~M~\text{hyperedges})\right]\nonumber\\
	&= \sum_{V \in \cV_R} \sum_{M=1}^{d\ell} p^M \cdot \Pr_{H \in \cH}\left[(V,H) ~\text{even with}~M~\text{hyperedges})\right]\mper \label{eq:prbdM}
    \end{align}

    To bound this probability, we will again sample uniformly $H\sim \cH$ in a two step process.
    \begin{compactenum}
	\item Sample a uniformly random perfect matching (with $2$-edges rather than hyperedges) between each set $S_i, S_{i+1} \in \cV_R$--call the edge set sampled in this manner $E$, so that we now have the graph $(V,E)$.
	\item Sample hyperedge matching configuration from $E$ by choosing a uniform random grouping of the edges between $S_i,S_{i+1}$ into groups of $k/2 = \kappa$ edges.
	\end{compactenum}
	We invoke the following lemma, which is a very slight embellishment upon \pref{lem:graph-from-hgraph-nocount}:
\begin{lemma}
    \torestate{
    \label{lem:graph-from-hgraph}
    Let $\hght,\wdth,\kappa,\ty,\tau \in \N$.
    Let $V \in \cV_R$ be a vertex configuration with $R$-multilinear vertex sets $S_1,\ldots,S_{\wdth} \in [n]^{\kappa\hght}$.
    Let $H \in \cH$ be a hypergraph configuration with $\wdth$~ $2\kappa$-uniform hypergraph matchings between the sets $S_i, S_{i+1} \forall i \in [\wdth]$, with $\kappa$ vertices from $S_i$ and $\kappa$ vertices from $S_{i+1}$ in each hypergraph matching.

    Suppose that $(V,H)$ has $\tau$ distinct labeled hyperedges and the evenness property, where hyperedges on the same vertex set but with a different partition into $S_i,S_{i+1}$ count as distinct.
Suppose that we sampled $H$ by first choosing a set of simple-edge perfect matchings $E$ on $V$, then grouping them into hyperedges.
Then
\[
    \Pr(~(V,E) ~\text{even with}~\ty \le \kappa\tau ~\text{edges}~|~ (V,H)~ \text{even with}~\tau ~\text{edges}) \ge \left(\frac{1}{\kappa!}\right)^{\wdth \hght}\mper
\]
}
\end{lemma}
\begin{proof}
    The proof is almost identical to that of \pref{lem:graph-from-hgraph-nocount}--choosing a random matching within each hyperedge gives a uniformly random $E$ from which $H$ is sampled, and that with probability at least $(\kappa!)^{-\wdth \hght}$ we choose the same matching in every copy of every hyperedge.
    We need only add that if a hyperedge $h_i \in (V,H)$ has multiplicity $a$, then if we chose the same matching in every copy of $h_i$, all $\kappa$ of the simple edges making up $h_i$ will have multiplicity at least $a$, so if $\ty$ is the total number of distinct edges in $(V,E)$, we have $\ty \le \kappa\tau$ and also the evenness property.
\end{proof}

    Letting $\cE_{H}^M$ be the event that $(V,H)$ is even with $M$ distinct hyperedges and letting $\cE_{E}^{kM/2}$ be the event that $(V,E)$ is even with at most $kM/2$ distinct hyperedges, \pref{lem:graph-from-hgraph} (and the asymmetry of $\bT_{\Phi}$) with $\wdth \leftarrow 2\ell$, $\hght \leftarrow d$, $\tau \leftarrow M$, $\kappa \leftarrow k/2$ implies that
\begin{align*}
    \Pr_{H\in\cH}(~(V,H)~\text{even with}~ M~ \text{edges})
    = \frac{ \Pr(\cE_{H}^M,\cE_{E}^{kM/2})}{\Pr(\cE_{E}^{kM/2}~|~\cE_{H}^{M})}
    &\le \left(\frac{k}{2}!\right)^{2d\ell} \Pr(\cE_{H}^{M},\cE_{E}^{kM/2}) \qquad (\text{ by \pref{lem:graph-from-hgraph}})\\
    &\le \left(\frac{k}{2}\right)^{dk\ell} \Pr(\cE_{E}^M) \cdot \Pr(\cE_H^{M}~|~\cE_{E}^{kM/2})\mper
\end{align*}

Therefore, from \pref{eq:prbdM} we have
\begin{align*}
    \E[\Tr(C_{(d)} C_{(d)}^\top)^{\ell}]
    &\le \left(\frac{k}{2}\right)^{dk\ell} \sum_{V \in \cV_R} \sum_{M=1}^{d\ell} \Pr(\cE_E^{kM/2}) \cdot \Pr(\cE_H^M | \cE_E^{kM/2})  \cdot p^M\mcom
\end{align*}
Now, we use a lemma to bound the conditional probability of sampling an even hyperedge matching with $M$ hyperedges, given that we sampled an even matching with at most $kM/2$ edges:
\begin{lemma}
    \torestate{\label{lem:matching-probability-2}
	Suppose $\hght,\wdth,\kappa,n,\tau \in \N$.
	Let $G =(V,E)$ be a graph consisting of $\wdth$ sets of $\kappa \hght$ vertices each with $R$-multilinear labels from $[n]$, where
	$E$ is a set of $\wdth$ perfect matchings $M_1,\ldots,M_{\wdth}$, so that $M_i$ is a perfect matching between $S_i$ and $S_{i+1}$,
	and $\alpha = a_1,\ldots, a_{\ty}$ is a list of even edge multiplicities of $E$ on the labeled vertex set $V$, so that $\sum a_i = \kappa\wdth\hght$.

	Suppose we sample a hyperedge matching configuration $H$ from $E$ by uniformly grouping the edges in each matching from $S_i$ to $S_{i+1}$ into hyperedges of order $2\kappa$, and let $\tau$ be a number of distinct hyperedges that is possible to sample from $(V,E)$ in this way.
    Then,
\[
    \Pr((V,H)\text{ even with } \tau\text{ edges} ~|~ (V,E) \text{ even })
    \le \frac{(2e^\kappa \kappa^\kappa R^{\kappa+1} \wdth )^{\wdth \hght} }{(\kappa\hght)^{(\kappa-1)(\wdth\hght - \tau)}}\mper
\]}
\end{lemma}
We'll prove \pref{lem:matching-probability-2} below in \pref{sec:pr-hgraph}.
For now, we apply \pref{lem:matching-probability-2} with $\wdth \leftarrow 2\ell$, $\hght \leftarrow d$, $\kappa \leftarrow k/2$ and $\tau \leftarrow M$, which for $R$-multilinear $V \in \cV_R$ implies that
\[
\Pr(\cE_H^M | \cE_E^{kM/2}) \le
\frac{(4e^{k/2} R^{k/2+1} (k/2)^{k/2}\ell)^{2d\ell}}{(dk/2)^{(k/2-1)(2d\ell - M)}}\mper
\]
Combining this with the above and letting $c_1 := 4e^{k/2}(k/2)^{k/2}$ for convenience,
\begin{align}
    \E[\Tr(C C^\top)^{\ell}]
    &\le \left(\frac{k}{2}\right)^{dk\ell} \sum_{V \in \cV_R} \sum_{M=1}^{d\ell} \Pr(\cE_E^{kM/2}) \cdot \frac{(c_1R^{k/2+1}\ell)^{2d\ell}}{(dk/2)^{(k/2-1)(2d\ell - M)}}  \cdot p^M\mper \label{eq:sum-over-all}
\end{align}
It remains for us to bound $\Pr((V,E)~\text{even with}~\le kM/2~\text{edges})$.
We now interchange the order of the summation, and bound the sum over $V$ for a fixed value of $M$.
Letting $\cM$ be the set of all possible edge configurations $E$, we have
\begin{align}
    \sum_{V \in \cV_R} \Pr_{E}(\cE_E^{kM/2})
    &= \sum_{V \in \cV_R} \frac{\left|\{ E ~|~(V,E) ~\text{ even with} \le kM/2 ~\text{edges}\}\right| }{|\cM|}\nonumber\\
    &= \frac{\left|\right\{E , V  ~|~ (V,E) ~\text{ even with} \le kM/2 ~\text{edges}\left\}\right| }{|\cM|}\mper\label{eq:pr-to-count}
\end{align}
We will bound this quantity with the following proposition, which counts the number of $V \in \cV_R$ that yield and even graph with at most $\ty$ edges on a fixed $E \in \cM$.
\begin{proposition}
    \torestate{
	\label{prop:multilin-graph-bound}
	Let $\wdth, \hght, n \in \N$.
	Let $\alpha = \alpha_1,\ldots,\alpha_\ty$ be a sequence of $\ty$ even numbers so that $\sum_{i=1}^\ty \alpha_i = \wdth \cdot \hght$.
	Let $E = M_1,\ldots,M_\wdth$ be a sequence of perfect matchings between two sets of size $\hght$.

	Let $\cG_{\wdth \times \hght}^{\alpha,E}$ be the set of all graphs which have a vertex set comprised of $\wdth$ $R$-multilinear multisets $S_1,\ldots, S_\wdth \in [n]^{\hght}$, and have edges forming the perfect matching $M_i$ between $S_i, S_{i+1}$ (where the indexing is modulo $\wdth$), so that the labels in $[n]$ assigned to the vertices induce exactly $\ty$ distinct labelings for the edges, and the labeled edges have multiplicities $\alpha_1,\ldots,\alpha_\ty$.
    In words, $\cG_{\wdth\times\hght}^{\alpha,E}$ is the set of $\wdth \times \hght$ matching cycles with matchings specified by $E$ that have edge multiplicities $\alpha$ when labeled with $R$-multilinear labels from $[n]$.

    If $\wdth\cdot\hght \le n$,
    \[
	|\cG^{\alpha,E}_{\wdth,\hght}| \le
	(5R\wdth)^{\wdth\hght} (\wdth\hght)^3\cdot n^{\ty + \hght}\mper
    \]
}
\end{proposition}
We will prove this proposition below, in \pref{sec:pr-hgraph}.
Applying \pref{prop:multilin-graph-bound} with $\hght \leftarrow kd/2$, $\wdth \leftarrow 2\ell$, and $\ty \leftarrow m$, we have that for a fixed $E \in \cM$ and for a fixed list of edge multiplicities $a_1,\ldots,a_{m}$,
\[
    \left|\{ V ~|~ (V,E)~\text{has}~ m ~\text{edges with multiplicities}~a_1,\ldots,a_m\}\right| \le
	(10R\ell)^{dk\ell} \cdot (dk\ell)^2 \cdot n^{m + dk/2}.
\]
where we have used the assumption that $dk\ell \ll n$ to meet the requirements of \pref{prop:multilin-graph-bound}.
The number of possible edge multiplicity lists $a_1,\ldots,a_m$ for a given value of $m$ is at most $\binom{m + dk\ell-1}{m-1} \le 2^{2dk\ell}$.
Thus, applying \pref{eq:pr-to-count}  and noting that there are $|\cM|$ choices for $E$ for each $V$,
\begin{align*}
    \sum_{V \in \cV_R} \Pr_{E}(\cE_E^{kM/2})
     &\le \frac{1}{|\cM|} \cdot \sum_{m=1}^{kM/2} |\cM| \cdot\sum_{a_1,\ldots,a_m}
	\left|\{ V ~|~ (V,E)~\text{has}~ m ~\text{edges with even mult.s}~a_1,\ldots,a_m\}\right| \\
	 &\le \sum_{m = 1}^{kM/2} 2^{2dk\ell} \cdot (10R\ell)^{dk\ell} \cdot (dk\ell)^2 \cdot n^{m + dk/2}
	~\le~  (40 R\ell)^{dk\ell} \cdot (dk\ell)^2 \cdot n^{kM/2+ dk/2+1}\mper
\end{align*}

Combining \pref{eq:sum-over-all} and the above, there is a constant $c_2$ depending on $k$ so that
\begin{align*}
    \E[\Tr(C_{(d)} C_{(d)}^\top)^{\ell}]
    &\le \left(\frac{k}{2}\right)^{dk\ell} \sum_{M = 1}^{d\ell} p^M \cdot \frac{(c_1 R^{k/2+1}\ell)^{2d\ell}}{(dk/2)^{(k/2-1)(2d\ell - M)}} \sum_{V \in \cV_R} \Pr(\cE_E^{kM/2})\\
    &\le \left(\frac{c_2 R^{2k + 2} \ell^{k+2} }{d^{2(k/2-1)}}\right)^{d\ell} \cdot(dk\ell)^2 n^{dk/2+1}
   \cdot \sum_{M=1}^{d\ell}  \left(p d^{(k/2-1)}n^{k/2}\right)^M
    \intertext{By assumption, $p\cdot (d^{k/2-1}n^{k/2}) \ge 1$, so the term $M = d\ell$ dominates:}
    &\le \left(\frac{c_2 R^{2k + 2} \ell^{k+2} }{d^{2(k/2-1)}}\right)^{d\ell} \cdot(d\ell)^2 n^{dk/2+1}
    \cdot d\ell \left(d^{(k/2-1)}pn^{k/2}\right)^{d\ell} \\
    &= \left(\frac{c_2 R^{2k + 2} \ell^{k+2} \cdot p n^{k/2} }{d^{k/2-1}}\right)^{d\ell} \cdot(d\ell)^3 n^{dk/2+1}\mper
\end{align*}
Choosing $\ell = O(d\log n)$, recalling that $R = 100\log n$, and invoking \pref{prop:tpm}, we have that with probability $1 - n^{-100}$, for some constant $c_k:= c(k)$,
\[
    \|C_{(d)}\|^{1/d}
\le c_k \log^{2k} n \cdot \left(\frac{p^{1/2} n^{k/4}}{d^{(k-2)/4}}\right)\mper
\]
The conclusion follows.
\end{proof}

%=================

\subsection{Odd $k$-XOR}\label{sec:lin-odd}

In this section, we modify our algorithm for refuting random even $k$-XOR instances to handle odd $k$-XOR instances.
The odd $k$-XOR algorithm is extremely similar to the algorithm for even $k$-XOR, save for complications introduced by the fact that an odd-order tensor has no natural matrix flattening.

The solution is to apply the Cauchy-Schwarz inequality to the objective value.
Let $k = 2\kappa + 1$ for some integer $\kappa$.
For the tensor $\bT_{\Phi}$ formed by the constraints of $\Phi$ and for its $n^{\kappa} \times n^{\kappa}$ slices $T_i ~\forall i \in [n]$, we have that
\[
    \Iprod{x^{\tensor k}, \bT_{\Phi}}^2
    \le \left(\sum_{i\in[n]}  x_i^2\right)\left( \sum_{i\in[n]}\left((x^{\tensor \kappa})^\top T_i x^{\tensor \kappa}\right)^2\right)
    = n \cdot (x^{\tensor 2\kappa})^\top \left(\sum_{i \in [n]} T_i \tensor T_i \right)x^{\tensor 2\kappa}\mper
\]
Now, the first technicality arises--since the entries $(T_i \tensor T_i)_{(ab),(cd)} = \bT_{a,c,i}\cdot \bT_{b,d,i}$ are always squares when $a=b$ and $c=d$, we must subtract them from the matrix $\sum_i T_i \tensor T_i$, as otherwise they contribute too much to the norm.
Thus, using $\squares(\cdot)$ to refer to the part of the matrix for which $a=b$ and $c = d$, we instead will use that the number of constraints $m = (x^{\tensor 2\kappa})^\top \squares\left(\sum_{i \in [n]} T_i \tensor T_i \right)x^{\tensor 2\kappa}$, and that
\begin{equation}
    \Iprod{x^{\tensor k}, \bT_{\Phi}}^2 - mn
    \le n \cdot (x^{\tensor 2d})^\top \left(\sum_{i \in [n]} T_i \tensor T_i -\squares(T_i \tensor T_i) \right)x^{\tensor 2d} \mper\label{eq:psisat}
\end{equation}
We can also view this as doing one step of resolution, so that we have gotten a $4\kappa$-XOR instance starting from a $(2\kappa + 1)$-XOR instance.
That is how we will treat our new instance from now on.
\medskip

Suppose $\Phi$ has $\pm 1$-constraint predicates $C_1,\ldots,C_m$, so that $C_a(x) = \eta_a \cdot \prod_{j \in S_a} x_j$.
We create a new $2(k-1)$-XOR instance $\Psi$ as follows.
For each $a,b \in [n]$, $a \neq b$:
if $C_a$ and $C_b$ both contain the variable $i$ in the $k$th position, add the $\pm 1$ constraint predicate $C'_{ab}(x) = \eta_a \cdot \eta_b \cdot \left(\prod_{j \in S_a} x_j \right)\left(\prod_{j \in S_a} x_j \right)$ to $\Psi$.
Let $m'$ be the number of clauses in $\Psi$.

The right-hand side of \pref{eq:psisat} is $n \cdot \sum_{ab} C'_{ab}(x) = n \cdot 2m' \cdot (P_{\Psi}(x) - \tfrac{1}{2})$, where $P_{\Psi}(x)$ is the fraction of clauses of $\Psi$ satisfied by $x$.
Combining this with the above calculations,
\begin{align}
    \left(2m\left(P_{\Phi}(x) - \frac{1}{2}\right)\right)^2
    &\le nm + n \cdot 2m' \cdot \left(P_{\Psi}(x) - \frac{1}{2}\right),\nonumber\\
    P_{\Phi}(x)
    &\le \frac{1}{2} + \frac{1}{2m}\sqrt{ nm  + 2nm' \cdot \left(P_{\Psi}(x)-\frac{1}{2}\right) }\mper\label{eq:psi-upper}
\end{align}

Now, we will essentially apply our even-$k$-XOR strategy to $\Psi$.
The only issue is that the clauses of $\Psi$ are not independent, so we will need to zero out not only rows and columns indexed by high-multiplicity subsets of $[n]^{2\kappa}$, but also get rid of terms that contain the same slice with too high a multiplicity.
So, instead of taking the $d$th tensor power of the matrix $\sum_{i} T_i \tensor T_i - \squares(T_i\tensor T_i)$, we omit the cross-products in which $T_i \tensor T_i$ appears more than $100\log n$ times for any $i \in [n]$.

Formalizing this, we introduce our matrix certificate for the odd case:

\begin{algorithm}[Odd $k$-XOR certificate at level $d$]\label{alg:lin-sub-odd}~\\
    {\bf Input:} A $k$-XOR instance for odd $k = 2\kappa+1$ on $n$ variables
    with $m$ clauses $C_1,\ldots,C_m$, where $C_a(x) = \eta_a \cdot \prod_{j \in S_a} x_j$ for
    $S_j \in [n]^{k}$ and $\eta_a \in \{\pm 1\}$.
    \begin{compactenum}
    \item Form the tensor $\bT := \bT_{\Phi}$ by setting $\bT_{S_a} = b_a$ for all $a \in [m]$, and setting all other entries to $0$.
	  \item Initialize an empty $n^{2d\kappa} \times n^{2d\kappa}$ matrix $\Gamma$.
	  \item For each ordered multiset $U \in [n]^{d}$ in which no entry appears with multiplicity $ >100\log n$:
	      \begin{compactenum}
	      \item Add the squared tensor of the slices of $\bT_{\Phi}$ corresponding to the indices in $U$:
		  \[
		      \Gamma := \Gamma + \bigotimes_{i \in U} \left( T_i \tensor T_i - \squares(T_i \tensor T_i)\right)
		  \]
		  where $\squares(\cdot)$ is the restriction to entries $(I,J),(K,L)$ such that $(I,K) = (J,L)$ as ordered multisets.
	\end{compactenum}
	\item Letting $\hat \cS_{2d\kappa}$ be the set of all permutation
		matrices that perform the index permutations
		corresponding to $\cS_{2d\kappa}$ on the rows and columns of $\Gamma$, form
	\[
	    \Gamma_{(d)}:= \E_{\Pi,\Sigma \in \hat\cS_{2d\kappa}}\left[\Pi \Gamma \Sigma \right]\mper
	\]
    \item Set to zero all rows and columns of $\Gamma_{(d)}$ indexed by multisets $S \in [n]^{2d\kappa}$ which contain some element of $[n]$ with multiplicity $> 100\log n$.
    \end{compactenum}
    {\bf Output:} The value $\|\Gamma_{(d)}\|$.
\end{algorithm}
The following theorem, which is the main theorem of this section, gives a bound on the value output by \pref{alg:lin-sub}.
\begin{theorem}
    \torestate{
	\label{thm:odd-lin-upper}
    Let $k,n,d \in \N$, so that $d\log n \ll n$, and furthermore let $k$ be odd so that $k = 2\kappa + 1$.
    Let $\Phi$ be a random instance of $k$-XOR on $n$ variables $x_1,\ldots,x_n$, with $\Theta(p n^k)$ clauses (so each constraint is sampled uniformly and independently with probability $p$).
    Let $\Gamma_{(d)}$ be the matrix formed from the instance $\Phi$ as described in \pref{alg:lin-sub-odd}.
    Then if $d^{(k-2)/2}n^{k/2}p > 1$, there exists a constant $c_k$ depending on $k$ such that with high probability over the choice of $\Phi$,
    \[
	\|\Gamma_{(d)}\|^{1/d} \le
    \tilde O \left(\frac{pn^{k/2}}{d^{(k-2)/2}}\right)\mper
    \]
}
\end{theorem}
We will prove the theorem below, in \pref{sec:odd-lin-norm-sec}, and we will now show how to take this matrix and acquire a certificate from it.

\subsubsection{Validity of the odd certificate}

Again, our strategy will be to work with the polynomial $(P_{\Psi}(x))^d$, which is not much altered by removing terms corresponding to the high-multiplicity rows, columns, or slice cross-products.

\begin{proposition}\label{prop:formula-odd}
    Let $\Phi$ be a random $k$-XOR formula in which each clause is sampled independently with probability $p$, and let $\Psi$ be the $2(k-1)$-XOR instance obtained from $\Phi$ as described above, where $\Psi$ has $m'$ clauses $\{C_{ab}\}_{ab}$ corresponding to pairs of clauses from $\Phi$ sharing the same final variable.

    Let $\hat\cC^d_{low} \subset [m']^d$ be the set of all ordered multisets of clauses $C_{a_1b_1},\ldots,C_{a_db_d}$ from $\Psi$ with the property that if we form three multisets of variables, $I,J \in [n]^{d(k-1)}$ and $S \in [n]^d$, with $I$ containing the first $(k-1)/2$ variables of each $C_{a_\ell}, C_{b_{\ell}}$, $J$ containing the next $(k-1)/2$ variables of each $C_{a_\ell},C_{b_{\ell}}$, and $S$ containing the last (shared) variable of $C_{a_\ell}$ and $C_{b_\ell}$, then $I,J$ are both low-multiplicity multisets, in that both have no element of $[n]$ with multiplicity $> 100\log n$.

    Let $o_{\max}$ be the maximum number of clauses of $\Psi$ any variable appears in.
    Then if $d\ll n$ and $d (2k-1) o_{\max} < 200 \epsilon m' \log n$,
    \[
	P_{\Psi}(x) \le \left(\E_{a_1 b_1,\ldots,a_d b_d \sim \hat\cC^d_{low}}\left[\prod_{\ell=1}^d \frac{1}{2}\left(1 - C_{a_\ell}(x)C_{b_\ell}(x)\right)\right]\right)^{1/d} + \epsilon
    \]
for all $x \in \{\pm 1\}^n$ with high probability.
    When $p \ge 200 \frac{\log n}{n^{k-1}}$, we have that $\epsilon = o(1)$ with high probability.
\end{proposition}
\begin{proof}
    The proof is very similar to that of \pref{prop:formula-odd}.
    First, let $m'$ be the number of clauses in $\Psi$, and let $o_{\max}$ be the maximum number of clauses of $\Psi$ that any variable $i \in [n]$ appears in (even if it the shared variable and is included with multiplicity 2).

    By definition, we have that $P_{\Psi}(x)$ gives the proportion of satisfied clauses, so
    \begin{equation}
	(P_{\Psi}(x))^d = \E_{a_1 b_1,\ldots,a_d b_d \sim [m']^d}\left[\prod_{\ell=1}^d \frac{1}{2}\left(1 - C_{a_\ell}(x)C_{b_\ell}(x)\right)\right].
    \end{equation}
    Since only a $o(1)$ fraction of the multisets of indices, $[n]^{dk}$, will not contain any item with multiplicity more than $100\log n$, we will be able to prove that those terms contribute negligibly.

    We sample a uniform element  $\cC\sim\hat \cC^d_{low}$, $\cC = (C_{a_1},C_{b_2}),\ldots,(C_{a_{d}},C_{b_{d}})$ in the following way.
    For $t = 1,\ldots, d$:
    \begin{compactitem}
	    \item Let $\cA_t \subset \cI$ be the set of pairs of clauses such that for any $(C',C'') \in \cA$, $(C_{a_1},C_{b_1}),\ldots,(C_{b_{t-1}},C_{b_{t-1}}), (C',C'')\in \hat \cC^t_{low}$, that is, the set of clauses from $\Psi$ that maintain the low-multiplicity conditions.
	\item Choose a uniformly random $(C',C'') \sim \cA_t$ and set $C_{a_t},C_{b_t} := (C',C'')$, adding $(C',C'')$ to $\cC$.
    \end{compactitem}
	This sampling process clearly gives a uniformly random element of $\hat \cC^t_{low}$.
    \begin{claim}\label{claim:exclusions-odd}
	At step $t+1$ there are at least $m' - t\frac{(2k-1)\cdot o_{\max}}{R}$ clauses that can be added.
    \end{claim}
    \begin{proof}
	In order to exclude any variable $i \in [n]$, we must add at least $100\log n$ copies of $i$..
	Further, to exclude $\ell$ distinct variables in $[n]$, we must add at least $100\log n$ copies of each variable, for a total of $100\ell\log n$ variables, which requires adding at least $100\ell\log n/(2k-1)$ pairs (since each pair contains $2k-1$ variables).
	If $\ell$ distinct variables are excluded, then at most $\ell \cdot o_{\max}$ pairs of clauses are excluded.
	The claim now follows.
    \end{proof}
    Now, define the random variable $X_{t} = \prod_{j=1}^t \frac{1}{2}(1-C_{a_j}(x)C_{b_j}(x))$--this is the $0$-$1$ value of $x$ on $\cC_t$.
    We apply \pref{claim:exclusions-odd}, along with the observation that $P_{\Psi}(x)$ can only drop by $1/m'$ for each clause pair removed, to conclude that
    \begin{align*}
	\E[X_{t+1}|C_{a_1b_1},\ldots,C_{a_tb_t}]
	\ge \left(P_{\Psi}(x) - t\frac{ (2k-1)\cdot o_{\max}}{100 m'\log n}\right) \cdot X_{t}\mper
    \end{align*}
    From this we have that
    $ \E[X_t]
    \ge \left(P_{\Psi}(x) -  t\frac{(2k-1)o_{\max}}{100m'\log n}\right) \cdot \E[X_{t-1}]$
    from which we have that as long as $d(2k-1) o_{\max} \le \epsilon 100m'\log n $,
    \[
	\E[X_k] \ge \prod_{t=1}^{d} \left(P_{\Psi}(x) - t\frac{(2k-1) o_{\max}}{100m'\log n}\right)
	\ge \left(P_{\Psi}(x) - \epsilon \right)^d \mper
    \]
So taking $\epsilon = 1/\log n$, if we can establish that the inequality $d(2k-1) o_{\max} \le 100 m'$ with high probability when $p\ge \Omega(\log n / n^{k-1})$, then we are done.

    A Chernoff bound implies that $pn^{k}/2\le m \le 2pn^k$ with probability at least $1-\exp(-\Omega(pn^k))$, and that each variable's degree $m_i$ is $pn^{k-1}/2 \le m_i \le 2 pn^{k-1}$ with probability at least $1 - \exp(-\Omega(pn^{k-1}))$.
    We have that $m' = (\sum_i m_i^2) - m$, and so by a union bound and using the assumption that $p n^{k-1} \ge \Omega(\log n)$, we have that
    \[
	m' \ge  p^2 n^{2k-1}/4.
    \]

    Let $o_i$ be the degree of variable $i$ in $\Psi$.
    To bound $o_{max}$, we observe that $o_i$ is made up of occurrences of pairs in which $i$ is the shared variable, and of pairs in which $i$ is not the shared variable.
    The contribution of the first category is $m_i^2$, and with high probability by our union bound $m_i^2 \le (2pn^{k-1})^2$.
    In the second category, we have $\sum_{j} m_{j} \cdot m_{ij}$, where $m_{ij}$ is the number of clauses containing $i$ and $j$.
    By our previous assumption regarding the concentration of the $m_i$, we have that $\sum_j m_j \cdot m_{ij} \le 2pn^{k-1}\sum_j m_{ij}$.
    The quantity $\sum_j m_{ij} = m_i$, and so we can conclude that $o_{\max} \le 4 p^2 n^{2k-2}$, so that $o_{\max}/ m' \le 16/n $, yielding our result.
\end{proof}

The proof above can be modified to give the following lemma, which gives a lower bound on the number of low-multiplicity terms.
\begin{lemma}\label{lem:sizeofmultilinear-odd}
    If $\Phi$ is a random $k$-XOR instance, then so long as $d \ll n$ and $d(2k-1)o_{max} < 200\epsilon m'\log n$,
 \[
     |\hat\cC_{low}^d| \ge ((1-\epsilon)m')^d \cdot\mper
 \]
 Furthermore, $\epsilon = o(1)$ with high probability when $p \ge \Omega(n^{-k+1}\log n )$.
\end{lemma}
The proof proceeds exactly as the proof of \pref{prop:formula-odd}, but instead of bounding the decrease in the value as each clause is added, one bounds the probability that a clause the multiplicity restriction is chosen.

Now, we have that
\begin{align}
    \E_{a_1b_1,\ldots,a_db_d \in \hat \cC^d_{low}}\left[ \prod_{\ell = 1}^d \frac{1}{2}(1 + C_{a_\ell}(x)C_{b_{\ell}}(x))\right]
    &=
    \frac{1}{2^d} \sum_{S \subseteq [d]} \E_{a_1b_1,\ldots,a_db_d \in \hat \cC^d_{low}}\left[\prod_{\ell \in S} C_{a_\ell}(x)C_{b_\ell}(x)\right]\nonumber
    \intertext{and by the symmetry of the uniform distribution over $\hat\cC^d_{low}$,}
    &= \frac{1}{2^d} \sum_{j=0}^d \binom{d}{j} \cdot \E_{a_1b_1,\ldots,a_jb_j \in \hat \cC^j_{low}}\left[\prod_{\ell=1}^j C_{a_\ell}(x)C_{b_\ell}(x)\right]\nonumber \\
    &\le \frac{1}{2^d} \sum_{j=0}^t \binom{d}{j}  + \frac{1}{2^d}\sum_{j=t+1}^d \E_{a_1b_1,\ldots,a_jb_j \in \hat \cC^j_{low}}\left[\prod_{\ell=1}^j C_{a_\ell}(x)C_{b_\ell}(x)\right]\nonumber \\
    &=\frac{1}{2^d} \sum_{j=0}^t \binom{d}{j}  +\frac{1}{2^d} \sum_{j=t+1}^d \binom{d}{j} \cdot \frac{1}{|\hat \cC_{low}^j|} (x^{\otimes 2j\kappa})^{\top}\Gamma_{(j)}x^{\otimes 2j\kappa}\mper\label{eq:expansion}
\end{align}

Now, we can use the spectral norm of $\Gamma_{(j)}$ as a certificate, for values of $j \in [\delta d, d]$--we stitch together the details below.
The following concludes the proof of the refutation theorem, modulo the proof of the $\Gamma_{(d)}$ matrix norm bound from \pref{thm:odd-lin-upper}, which we give in the next subsection.
\begin{theorem}\label{thm:refutation-odd}
    Let $k = 2\kappa +1$ be odd, and let $d \ll n$.
    Then for sufficiently large $n$ there is an algorithm that with high probability certifies that a random $k$-XOR instance has value at most $\frac{1}{2} + \gamma + o(1)$ for any constant $\gamma > 0$ at clause density $m/n = \tilde O\left(\frac{n^{(k-2)/2}}{d^{(k-2)/2}}\right)$ (where the $\tilde O$ hides a dependence on $\gamma$ and $k$) in time $n^{O(d)}$.
\end{theorem}

\begin{proof}
    As argued in the proof of \pref{prop:formula-odd}, we have that $m' = \Theta(p^2 n^{2k-1}) $ and $m= \Theta(pn^{k})$ with high probability, as long as $p \ge \Omega(\frac{\log n}{n^{k-1}})$.
    Suppose that no variable appears in more than $o_{\max}$ clauses in $\Psi$.
    Then for $\beta = \frac{d(2k-1)o_{\max}}{200 m' \log n}$,
from \pref{prop:formula-odd} and \pref{eq:expansion},
	\begin{align*}
	    (P_{\Psi}(x)-\beta)^d
	    &\le
	    \frac{1}{2^d} \sum_{j=0}^t \binom{d}{j}  +\sum_{j=t+1}^d \frac{\binom{d}{j}}{2^d} \cdot \frac{1}{|\hat \cC_{low}^j|} (x^{\otimes 2j\kappa})^{\top}\Gamma_{(j)}x^{\otimes 2j\kappa}\mper\label{eq:expansion}
	\end{align*}
	If we choose $t = \delta d$ for some constant $\delta > 0$, then we can bound the $j$th term in the second summation by
    \begin{align*}
	\frac{1}{|\hat \cC_{low}^j|} (x^{\otimes 2j\kappa})^{\top}\Gamma_{(d)}x^{\otimes 2j\kappa}
	&= \frac{\|x\|^{2j(k-1)}}{|\hat \cC^j_{low}|} \cdot \|\Gamma_{(j)}\|
	~=~ \frac{n^{j(k-1)}}{|\hat\cC^j_{low}|}\cdot \|\Gamma_{(j)}\|\\
	& \leq   \frac{n^{j(k-1)}}{(1-\beta))^j (m')^j} \cdot
	\tilde O \left(\frac{pn^{k/2}}{j^{(k-2)/2}}\right)^j\mcom
	\qquad \text{ (by \prettyref{lem:sizeofmultilinear-odd} and
	\prettyref{thm:odd-lin-upper})} \\
	&\le \tilde O \left(\frac{1}{(1-\beta) p n^{k/2}  j^{(k-2)/2}}\right)^j\mcom
    \qquad \text{(using that $m' = \Theta(p^2n^{2k-1})$ w.h.p.)}
	  \end{align*}
	  where the inequality holds with high probability from the conditions of \pref{thm:odd-lin-upper}, and therefore also holds with high probability simultaneously for all $j\in[\delta k,k]$ by a union bound.

	  The term comprised of the sum of binomial coefficients is at most
	\[
	    \frac{1}{2^d}\sum_{j=0}^{\delta d} \binom{d}{j}
	    \le 2^{(H(\delta)-1) d},
	\]
	Where $H(\cdot)$ is the binary entropy function, $H(\delta) = -\delta \log \delta - (1-\delta) \log (1-\delta)$.
	Also, $\beta = o(1)$ with high probability.
	Therefore, for some $\alpha \in [\delta , 1]$,
	\[
	    P_{\Psi}(x)
	    \le \left(2^{(H(\delta)-1)d} + \tilde O\left(\frac{1}{pn^{k/2} (\alpha d)^{(k-2)/2}}\right)^{\alpha d}\right)^{1/d} + o(1).
	\]
	Now, for $p\ge \tilde O(n^{-(k/2)}d^{-((k-2)/2)})$, the latter quantity  is $o(1)$, where the $\tilde O$ hides a $\polylog n$ and a dependence on $\delta$ and $k$.
	Thus, for $n$ sufficiently large the full term is at most $(1+o(1))2^{H(\delta) - 1}$.
	We can choose $\delta$ sufficiently small so as to bound $P_{\Psi}(x) \le \frac{1}{2} + \epsilon$ for any constant $\epsilon>0$.

	Using the fact that $c$, $m' \le 4 p^2  n^{2k-1}$ with high probability for sufficiently large $p \ge \tilde O(n^{-k+1})$ (see the proof of \pref{prop:formula-odd}), and that $m \ge pn^{k}/2$ with high probability, combining with \pref{eq:psi-upper} we have that
	\begin{align*}
    P_{\Phi}(x)
	    &\le \frac{1}{2} + \frac{1}{2m}\sqrt{nm  + 2nm' \cdot \left(P_{\Psi}(x)- \frac{1}{2}\right)}\\
	    &\le \frac{1}{2} + \sqrt{\Theta\left(\frac{1}{pn^{k-1}}\right) +\frac{\epsilon n m'}{2m^2}}\\
	    &\le \frac{1}{2} + \sqrt{o(1) +\frac{\epsilon 4 p^2n^{2k}}{2(\frac{1}{2})^2p^2n^{2k}}}
	    ~=~ \frac{1}{2} + \sqrt{\Theta\left(\frac{1}{pn^{k-1}}\right) +8\epsilon }
	\end{align*}
	and we can take the quantity within the square root to be an arbitrarily small constant by choosing a constant $\epsilon$ sufficiently small.

	  We can certify this bound in time $n^{O(d)}\cdot d$, by running \pref{alg:lin-sub-odd} to compute the top eigenvalue of $\Gamma_{(j)}$ for each $j\in[\delta d,d]$.
\end{proof}

\subsubsection{Bounding the odd certificate spectral norm}\label{sec:odd-lin-norm-sec}

Now, we bound $\|\Gamma_{(d)}\|$ for the matrices $\Gamma_{(d)}$ defined above in \pref{alg:lin-sub-odd}.
Before stating our theorem, we describe the hypergraphs corresponding to entries of $\Gamma_{(d)}$, and to $((\Gamma_{(d)})(\Gamma_{(d)})^\top)^{\ell}$.

We obtain $\Gamma_{(d)}$ by averaging over row and column symmetries of the matrix
\[
    \sum_{i\in[n]} \sum_{\substack{U \in [n]^d\\U ~ \text{low-mult} }} \bigoplus_{u \in U} T_u \tensor T_u - \squares(T_u \tensor T_u)\mcom
\]
then setting rows and columns indexed by high-multiplicity multisets to $0$.
Ignoring the subtracted squares for now, this can in turn be understood as the low-multiplicity restriction of the matrix
\[
    \left(\sum_{u} T_u \tensor T_u\right)^{\tensor d},
\]
where the low-multiplicity restriction is occurring on the Cauchy-Schwarz'd mode $u$, as well as on the rows and columns.
We begin with the hypergraph interpretation of the matrix $(\sum_{u} T_u \tensor T_u)^{\tensor k}$, from which the interpretation for $\Gamma_{(d)}$ will follow by our understanding of symmetrization over $\cS_{2d\kappa}$ and of low-multiplicity restrictions.
Let $M:=(\sum_{u} T_u \tensor T_u)^{\tensor d}$ for convenience.
We have that the $(A,B),(C,D)$th entry of $M$ (for $A,B,C,D \in [n]^{d\kappa}$ with $A = a_1,\ldots,a_{d}$ with $a_i \in [n]^\kappa$, and with similar decompositions defined for $B,C,D$) has value
\begin{align*}
    M_{(A,B),(C,D)}
    &=  \prod_{i \in [d]} \left(\sum_{u\in[n]}\bT_{a_i,c_i,u} \cdot \bT_{b_i,d_i,u}\right)
    =  \sum_{U \in [n]^d} \prod_{i \in [d]} \left(\bT_{a_i,c_i,u_i} \cdot \bT_{b_i,d_i,u_i}\right)\mper
\end{align*}
Interpreting the variables $\bT_{a_i,c_i,u_i}$ as $k = (2\kappa +1)$-uniform hyperedges, we have that each entry is a sum over hypergraphs indexed by $U \in [n]^d$.
For each $U \in [n]^d$, we have a hypergraph on the following vertex configuration: on the left, we have the vertices from the multiset $A,B$.
On the right, we have the vertices from the multiset $C,D$.
In the center, we have the vertices from $U$.
On this vertex set, we have $2d$ hyperedges.
Of these hyperedges, $d$ form a tripartite matching on the vertices in $A,C,U$, with $\kappa$ vertices from each of $A,C$ and one vertex in $U$.
The other $d$ form a similar tripartite matching on the vertices in $B,D,U$.
Every hyperedge on $A,C,U$ shares exactly one vertex in $U$ with exactly one hyperedge from $B,D,U$.
See \pref{fig:oddten} for an illustration.

Now, we detail the impact of subtracting the squares, and of removing high-multiplicity rows, columns, and Kronecker powers.
\begin{compactitem}
\item The subtraction of the square terms $\squares(T_u \tensor T_u)$ forces us to never have two hyperedges sharing a vertex in $U$ if they contain vertices of the same type in $[n]$: that is, we can never have $(a_i,c_i) = (b_i,d_i)$ as ordered multisets.
\item The deletion of high-multiplicity indices, both in the Cauchy-Schwarz'd mode and in the rows and columns, forces us to exclude hypergraphs with $(A,B)$, $(C,D)$, or $U$ containing more than $100\log n$ repetitions of any one vertex type.
\item The averaging operation $\E_{\Pi,\Sigma\in\hat \cS_{2d\kappa}}$ takes each such entry to an average over all allowed hyperedge configurations on the vertex set $(A,B), (C,D), U$.
\end{compactitem}
When we take $\Tr(\Gamma_{(d)}\Gamma_{(d)}^\top)^\ell$, we are taking a sum over all ``cycles'' of length $2\ell$ in such hypergraphs, where the vertices in the cycle are given  by the $(A,B)$ multisets, and the edges are given by the average hyperedge configuration between $(A,B)$ and the next $(C,D)$, with the $U$ vertices in between.

We are now ready to prove our upper bound on $\Gamma_{(d)}$.

\restatetheorem{thm:odd-lin-upper}
\begin{proof}
We fix $d$ and take $\Gamma := \Gamma_{(d)}$ for convenience.
    Also, fix $R := 100\log n$, and call a multiset $S \in [n]^m$ $R$-multilinear if no element of $[n]$ appears with multiplicity more than $R$ in $S$.
    We will apply the trace power method (\pref{prop:tpm}) to $\Gamma$, for which it suffices to obtain an upper bound on
    \[
	\E\left[\Tr\left(\Gamma_{(d)} \Gamma_{(d)}^\top)^\ell\right)\right]\mper
    \]
As described above, this amounts to bounding the number of hypergraph cycles of length $2\ell$, where each ``vertex'' of the cycle is comprised of an $R$-multilinear multiset $(A,B) \in [n]^{2d\kappa}$, and the ``edges'' in the cycle between $(A,B)$ and $(C,D)$ are the sum over all $R$-multilinear $U \in [n]^d$ of the average over all possible hypergraphs (because of the symmetries) that contain a tripartite hypergraph matching with $2d$ edges between $(A,B), U, (C,D)$ in which every hyperedge contains $d$ vertices from $(A,B)$, $d$ vertices from $(C,D)$, and one vertex from $U$.
    Hypergraphs that contain two identical hyperedges sharing a vertex from $U$ have contribution $0$ to the sum (due to the subtraction of the ``squares'').

    Let the set of all valid $R$-multilinear vertex configurations $V$ comprising $V = S_1,\ldots,S_{2\ell} \in [n]^{2d\kappa}$  be denoted $\cV_R$.
    Let the set of all $R$-multilinear center vertex configurations $U = U_1,\ldots,U_{2\ell} \in [n]^d$ be denoted $\cU_R$.
    Let the set of all hyperedge matching sequences $H = H_1,\ldots,H_{2\ell}$ with $H_i$ a matching between $S_i,U_i,S_{i+1}$ be denoted $\cH$.
    For $H\in\cH, V \in \cV_R, U \in \cU_R$, denote by $(V,U,H)$ the hypergraph given by the hyperedges $H$ on the vertex configuration $V$.
    We think of the elements in the sum $\Tr(\Gamma \Gamma^\top)^\ell$ as being indexed by $\E_{H}[(V,U,H)]$, where the expectation over $H$ is a result of our symmetrization/averaging operation.

    Applying the above observations, and recalling that we have assembled $\Gamma$ from the random tensor $\bT$, we have that
    \begin{align*}
	\E_{\bT}[\Tr(\Gamma \Gamma^\top)^{\ell}]
	&= \sum_{U \in \cU_R}\sum_{V \in \cV_R} \E_{H\in\cH} \left[\E_{\bT}\left[ \prod_{(i_1,\ldots,i_d) \in (V,U,H)} \bT_{i_1,\ldots,i_d}\right]\right]\mcom
    \end{align*}
    Because $\bT_{S} \neq \bT_{\pi(S)}$, our hyperedges are ordered, and so two hyperedge variables are not identical unless the vertices appear in the same order (in particular, the partition into $a_i, b_i, u_i$ and the order within each should be the same).
    The expectation over $\bT$ of a term is $0$ if any ordered hyperedge in $(V,U,H)$ appears with odd multiplicity or if two identical ordered hyperedges share the same vertex in some $U_i$, and is $p^M$ if exactly $M$ distinct hyperedges  appear in $(V,U, H)$.
    Thus,
    \begin{align}
	\E_{\bT}[\Tr(\Gamma \Gamma^\top)^{\ell}]
	&\le \sum_{U\in \cU_R}\sum_{V \in \cV_R} \sum_{M=1}^{d\ell} p^M \cdot \E_{H \in \cH}\left[\Ind((V,U,H) ~\text{even, nonsharing}) \cdot \Ind((V,U,H)~\text{has}~M~\text{hyperedges})\right]\nonumber\\
	&= \sum_{U\in\cU_R}\sum_{V \in \cV_R} \sum_{M=1}^{d\ell} p^M \cdot \Pr_{H \in \cH}\left[(V,U,H) ~\text{even, nonsharing with}~M~\text{hyperedges})\right]\mper \label{eq:prbdModd}
    \end{align}

    To bound this probability, we will sample uniformly $U \in \cU_R$ and $H\sim \cH$ in a three-step process.
    \begin{compactenum}
	\item Fix $V \in \cV_R$.
	\item Sample a uniformly random perfect matching (with $2$-edges rather than hyperedges) between each set $S_i, S_{i+1} \in \cV_R$--call the edge set sampled in this manner $E$, so that we now have the graph $(V,E)$.\label{step:matching}
	\item Sample a hyperedge matching configuration $G$ from $E$ by choosing a uniform random grouping of the edges between $S_i,S_{i+1}$ into groups of $d$ edges, to obtain the hypergraph $(V,G)$ (when $k = 3 \implies \kappa = 1$, this step is skipped).\label{step:hyp}
	\item Sample a pairing $H$ of the hyperedges in $G$, a center vertex for each pair in $H$ and an order on the center vertices to form $U$, to obtain the hypergraph $(V,U,H)$.\label{step:pair}
	\end{compactenum}
	For \pref{step:matching} and \ref{step:hyp}, we will employ the same \pref{prop:multilin-graph-bound} and \pref{lem:matching-probability-2} that we used in the proof of the even case (\pref{thm:lin-upper}) to bound the probability that we sample a $(V,E)$ and $(V,G)$ with a certain edge multiplicity and the evenness property.
	For \pref{step:pair}, we will need another lemma along the same lines.
	\medskip

	We note that if $(V,U,H)$ has every hyperedge appearing an even number of times and there are $M$ distinct edges, then even if all center vertices are removed to obtain a $2\kappa$-hypergraph $(V,G)$, every ordered hyperedge must still appear with even multiplicity, and there can only be at most $M$ distinct hyperedges.
	Therefore, letting $\cE^{M}_H$ be the event that $(V,U,H)$ is even with $M$ edges and no square/sharing hyperedges, letting $\cE^{m}_G$ be the event that $(V,G)$ is even with at most $m$ edges, we have that
	\begin{align}
	    \Pr[\cE^M_H]
	    = \sum_{m \le M} \Pr[\cE^M_H, \cE^m_G]
	    = \sum_{m \le M} \Pr[\cE^M_H | \cE^m_G] \cdot \Pr[\cE^m_G]\mper\label{eq:HG}
	\end{align}
	Now, let $\cE^{\le m'}_E$ be the event that $(V,E)$ is a simple graph with the evenness property and at most $m'$ distinct edges.
	We use the asymmetry of $\bT$ to invoke \pref{lem:graph-from-hgraph} with $\hght \leftarrow 2d$, $\wdth \leftarrow 2\ell$, $\tau \leftarrow M$ and $\kappa \leftarrow \kappa$ which gives us that $\Pr(\cE^m_G ~|~ \cE^{\le dm}_{E}) \ge (\kappa!)^{-4d\ell}$.
	From this, we have
	\begin{align}
	    \Pr[\cE^M_H]
	    &= \sum_{m \le M} \Pr[\cE^M_H | \cE^m_G] \cdot \Pr[\cE^m_G] \qquad\text{(by \pref{eq:HG})}\nonumber \\
	    &= \sum_{m \le M} \Pr[\cE^M_H | \cE^m_G] \cdot \frac{\Pr[\cE^m_G,\cE_E^{\le \kappa m}]}{\Pr[\cE_E^{\le \kappa m}~|~\cE^m_G]}\nonumber \\
	    &\le (\kappa!)^{4d\ell} \sum_{m \le M} \Pr[\cE^M_H | \cE^m_G] \cdot\Pr[\cE^m_G,\cE_E^{\le \kappa m}] \qquad (\text{by \pref{lem:graph-from-hgraph}})\nonumber \\
	    &\le \kappa^{4d\ell\kappa}\sum_{m \le M} \Pr[\cE^M_H | \cE^m_G] \cdot\Pr[\cE^m_G~|~\cE_E^{\le \kappa m}]\cdot \Pr[\cE_E^{\le \kappa m}]\mper \label{eq:conc}
	\end{align}

	We will bound $\Pr(\cE_{G}^{m}|\cE_E^{\le \kappa m})$, using \pref{lem:matching-probability-2} with $\hght \leftarrow 2d$,$\wdth \leftarrow 2\ell$, $\tau \leftarrow m$, which gives us that
	\[
    \Pr(\cE_{G}^{m}|\cE_E^{\le dm})\le \Pr((V,G)\text{ even with } m\text{ edges} ~|~ (V,E) \text{ even })
    \le
    \frac{( 4e^\kappa \kappa^\kappa R^{\kappa+1} \ell)^{4d\ell} }{(2d\kappa)^{(\kappa-1)(4d\ell - m)}}\mper
	\]

	And so now, combining with \pref{eq:conc}, we have that for some constant $c_1$ depending on $\kappa$,
	\begin{align}
	    \Pr[\cE^M_H]
	    &\le \kappa^{4d\ell\kappa} \sum_{m \le M} \Pr[\cE^M_H | \cE^m_G] \cdot\left(\frac{( 4e^\kappa \kappa^\kappa R^{\kappa+1} \ell)^{4d\ell} }{(2d\kappa)^{(\kappa-1)(4d\ell - m)}}\right)\cdot \Pr[\cE_E^{\le \kappa m}]\nonumber\\
	    &\le \frac{(c_1 R^{d+1}\ell)^{4d\ell}}{d^{(\kappa-1)4d\ell}} \sum_{m \le M} d^{(\kappa-1)m}\cdot \Pr[\cE^M_H | \cE^m_G] \cdot \Pr[\cE_E^{\le \kappa m}]
	\end{align}

    And therefore, with \pref{eq:prbdModd},
    \begin{align}
	\E[\Tr(\Gamma \Gamma^\top)^{\ell}]
	&= \frac{(c_1 R^{\kappa+1} \ell)^{4d\ell}}{d^{(\kappa-1)4d\ell}} \cdot  \sum_{M=1}^{d\ell} p^M \cdot\sum_{m \le M} \sum_{\substack{U \in [n]^d\\R-\text{multi}}}d^{(\kappa-1)m}\cdot \Pr[\cE^M_H | \cE^m_G]\sum_{V \in \cV_R}  \cdot \Pr[\cE_E^{\le \kappa m}]
    \end{align}

    We now bound $\Pr(\cE_{E}^{\le \kappa m})$.
If we interchange the order of summation, sum over these probabilities for a fixed value of $m$, letting $\cM$ be the set of all possible edge configurations $E$, we have
\begin{align}
    \sum_{V \in \cV_R} \Pr_{E}(\cE_E^{\kappa m})
    &= \sum_{V \in \cV_R} \frac{\left|\{ E ~|~(V,E) ~\text{ even with} \le \kappa m ~\text{edges}\}\right| }{|\cM|}\nonumber\\
    &= \frac{\left|\right\{E , V  ~|~ (V,E) ~\text{ even with} \le \kappa m ~\text{edges}\left\}\right| }{|\cM|}\mper\label{eq:pr-to-count}
\end{align}
We will bound this quantity with \pref{prop:multilin-graph-bound}, which counts the number of $V \in \cV_R$ that yield and even graph with at most $m$ edges on a fixed $E \in \cM$.
From \pref{prop:multilin-graph-bound} with $\wdth \leftarrow 2\ell$, $\hght \leftarrow 2d$, $\ty \leftarrow m'$, we have that
for a fixed $E \in \cM$ and for a fixed list of edge multiplicities $a_1,\ldots,a_{m'}$,
\begin{align*}
    \left|\{ V ~|~ (V,E)~\text{has}~ m' ~\text{edges with multiplicities}~a_1,\ldots,a_{m'}\}\right|
    &\le (c_2 R\ell)^{4d\kappa \ell} \cdot (d\ell)^2 \cdot n^{m' + 2d\kappa}
    \end{align*}
    for some constant $c_2$ depending on $k$, where we have used the assumption that $d \ll n$ to meet the requirements of \pref{prop:multilin-graph-bound}.
The number of possible edge multiplicity lists $a_1,\ldots,a_{m'}$ for a given value of $m'$ is at most $\binom{m' + 4d\kappa\ell-1}{m'-1} \le 2^{4dk\ell}$.
Thus,
\begin{align*}
    &\frac{\left|\right\{E, V ~|~ (V,E) ~\text{ even with} \le \kappa m ~\text{edges}\left\}\right| }{|\cM|}\\
    &\qquad\qquad\le \frac{1}{|\cM|} \cdot \sum_{m'=1}^{\kappa m} \sum_{a_1,\ldots,a_{m'}} |\cM| \cdot
    \left|\{ V ~|~ (V,E)~\text{has}~ m' ~\text{edges with even mult.s}~a_1,\ldots,a_{m'}\}\right| \\
	 &\qquad\qquad\le \sum_{m' = 1}^{\kappa m} 2^{4dk\ell} \cdot (c_2 R\ell)^{4d\kappa\ell} \cdot (d\ell)^2 \cdot n^{m' + 2d\kappa}
  ~\le~  (c_3 R\ell)^{4dk\ell} \cdot (d\ell)^2 \cdot n^{\kappa m+ 2d\kappa +1}\mcom
\end{align*}
for some constant $c_3$ depending on $k$.
So there is a constant $c_4$ depending on $k$ so that,
\begin{align*}
	\E[\Tr(\Gamma \Gamma^\top)^{\ell}]
	&\le \frac{(c_2 R^{\kappa+1}\ell)^{4d\ell}}{d^{(\kappa-1)4d\ell}} \cdot  \sum_{M=1}^{d\ell} p^M \cdot\sum_{m \le M} \sum_{\substack{U \in [n]^k\\R-\text{multi}}}d^{(\kappa-1)m}\cdot \Pr[\cE^M_H | \cE^m_G]\left((c_3R\ell)^{4d\kappa\ell} (d\ell)^2 n^{\kappa m+2d\kappa+1}\right)\\
	&= \frac{(c_4 R\ell)^{4d\ell(\kappa +1)}}{d^{(\kappa-1)4d\ell}}(d\ell)^2 n^{2d\kappa +1} \cdot  \sum_{M=1}^{k\ell} p^M \cdot\sum_{m \le M} d^{(\kappa-1)m}\cdot n^{\kappa m}\sum_{\substack{U \in [n]^k\\R-\text{multi}}} \Pr[\cE^M_H | \cE^m_G]\mper
\end{align*}

For a given hypergraph $(V,G)$ and a fixed $U$, there are $\left(d! \cdot \left(\frac{(2d)!}{d!2^d}\right)\right)^{2\ell}$ hyperedge groupings from which we can sample $(V,U,H)$--first we choose a matching of the $2d$ hyperedges in each column, and then we choose an ordering on the $d$ vertices of $U$ to determine which belong to each hyperedge pair.
We now appeal to the following lemma, which bounds the number of pairings and choices of $U$ that can result in the evenness property for a given $(V,G)$:
\begin{lemma}
    \torestate{
\label{lem:odd-match}
    Suppose $R,\wdth,\hght\in \N$ such that $\hght \le n$ and $\hght$ is even.
    Let $G$ be an even $R$-multilinear hypergraph with $\hght$ hyperedges per column and $\wdth$ columns, and hyperedge multiplicity profile $\alpha = a_1,\ldots,a_{\ty}$.
Let $H$ be a hypgergraph we sample from $G$ by matching edges in a column to each other, then adding a vertex in between with a label from the set $[n]$, with the additional constraint that the columns of center labels be $R$-multilinear, and that no two identical ordered hyperedges from $G$ are matched to the same center vertex.
Let $\tau$ be a valid number of distinct hyperedges sampleable from $G$.
Then
\[
    \left(\#~H \text{ even with }\le \tau \text{ edges} ~|~ G\right)
    \le (\frac{\hght}{2}!)^{\wdth}\cdot (2\hght \wdth)^2 (4\hght R^2)^{\hght\wdth} \cdot (\hght n)^{\tau/2}\mper
\]
}
\end{lemma}

Applying \pref{lem:odd-match} with $\wdth \leftarrow 2\ell$, $\hght \leftarrow 2d$,$\tau \leftarrow M$, we have that since all of our $U$-configurations are $R$-multilinear, and since we forbid two identical ordered edges to share a center vertex, we sum over all possible hyperedge multiplicity profiles (at most $2^{4d\ell}$ ) we divide by the number of possible hyperedge groupings and get that for some constant $c_5$ depending on $k$,
\[
    \sum_{U \in \cU_R} \sum_{\alpha} \Pr[\cE_{H}^{M} ~|~ \cE_G^{\le m}]
    \le 2^{8d\ell}\cdot \frac{(d!)^{2\ell}(8d\ell)^2 (4\ell R^2)^{4d\ell} \cdot (2d n)^{M/2}}{\left(d!\cdot \frac{(2d)!}{2^d d!}\right)^{2\ell}}
    \le (c_5R^2\ell)^{4d\ell}\frac{(d\ell)^2 \cdot (d n)^{M/2}}{d^{2d\ell}}
\]
And combining this with the above, there exists a constant $c_6$ depending on $k$ so that
\begin{align*}
	\E[\Tr(\Gamma \Gamma^\top)^{\ell}]
	&\le \frac{(c_4 R\ell)^{4d\ell(\kappa+1)}}{d^{(\kappa-1)4d\ell}}(d\ell)^2n^{2d\kappa +1} \cdot  \sum_{M=1}^{2d\ell} p^M \cdot\sum_{m \le M} d^{(\kappa-1)m}\cdot n^{\kappa m}\left(\frac{(d\ell)^2 (c_5\ell R^2)^{4d\ell} (dn)^{M/2}}{d^{2d\ell}}\right)\\
	&\le \frac{(c_6 R\ell)^{4d\ell(\kappa + 8)}}{d^{(\kappa-1)4d\ell}}\cdot \frac{1}{d^{2d\ell}} (d\ell)^4\cdot n^{2d\kappa+1} \sum_{M=1}^{2d\ell} p^M (dn)^{M/2} \cdot\sum_{m \le M} d^{(\kappa-1)m}\cdot n^{\kappa m}\\
	&\le \frac{(c_6 R\ell)^{4d\ell(\kappa + 8)}}{d^{(\kappa-1/2)4d\ell}} \cdot (d\ell)^4 n^{2d\kappa+1} (2d\ell)\sum_{M=1}^{2d\ell} p^M \cdot d^{(\kappa-1/2)M}\cdot n^{(\kappa+1/2)M}\\
	\intertext{And so long as $pk^{(\kappa-1/2)}n^{(\kappa+1/2)} \ge 1$,}
	&\le \frac{(c_6 R\ell)^{4d\ell(\kappa+8)}}{d^{(\kappa-1/2)4d\ell}} \cdot n^{2d\kappa+1} (d\ell)^6 \cdot p^{2d\ell} \cdot d^{(\kappa-1/2)2d\ell}\cdot n^{(\kappa+1/2)2d\ell}\\
	&\le n^{2d\kappa+1} (d\ell)^6 \cdot\left(c_6 (R\ell)^{2(\kappa+8)}\cdot\frac{pn^{\kappa+1/2}}{d^{\kappa-1/2}}\right)^{2d\ell},
\end{align*}
and now taking $\ell = O(\log n)$ and recalling that $R = 100\log n$, with high probability by \pref{prop:tpm} we have that
\[
    \|\Gamma\|^{1/d} \le \tilde O \left(\frac{pn^{k/2}}{d^{(k-2)/2}}\right)\mper
\]
\end{proof}
This concludes our bound on $\|\Gamma \|$--in the next subsection, we prove the bounds on the sampling probabilities that we relied upon in the proofs of \pref{thm:lin-upper} and \pref{thm:odd-lin-upper}.

\subsection{Bounding probabilities of sampling even hypergraphs}
\label{sec:pr-hgraph}

Our first proposition counts the number of vertex configurations with the evenness property and a given set of edge multiplicities for a fixed edge set $E$.
\restateprop{prop:multilin-graph-bound}

We remark that this proposition resembles a lemma used in establishing exact bounds on the order of the deviation of the second eigenvalue of a Wigner matrix, in the work of \cite{FK81}.
Unfortunately, their statement does not directly imply the bounds we require, as they work in a slightly different setting and wished to precisely bound the constant.
Our proof is similar to the exposition of \cite{FK81} in \cite{tao}.
\begin{proof}
    We bound the number of such graphs by encoding each graph as a unique string.
    Since $E$ is known, it suffices to encode enough information to recover the labels of the vertices.

    We will call $M_i$ (the matching in $E$ between $S_i$ and $S_{i+1}$) the $i$th \emph{column} of edges.
    We choose an ordering on the edges of $E$: we order them first by column, and within each column arbitrarily.
    Given a $G \in \cG_{\wdth\times\hght}^{\alpha,E}$, we will process the edges one at a time in this pre-specified order, and for each edge we will record either the labels of its incident vertices, or enough information to recover the labels from what we have previously recorded.
    To reduce the amount of recorded information, it will be helpful to specify several edge types:
    \begin{compactitem}
    \item Edges we see for the first time:
	\begin{compactitem}
	\item \emph{new-endpoint edges}: never-before seen edges $e_i$ with $a_i = 2$ which take us to a vertex with a label we have not seen before.
	    Let there be $\nsurp$ such edges.
	\item \emph{reused-endpoint edges}: never-before seen edges $e_i$ with $a_i = 2$ which take us to a vertex with a label we have already seen.
	\end{compactitem}
	\item Edges we see for the second (and last) time:
	    \begin{compactitem}
	\item \emph{return edges}: edges $e_i$ with $a_i = 2$ which we see for the second (and last) time.
	\item \emph{unforced edges}: return edges $e_i$ with $a_i = 2$ which are not the only possible labeled edge we can use from the endpoint vertex of the previous edge.
	\end{compactitem}
    \item Edges we see more than twice:
	\begin{compactitem}
	\item \emph{high-multiplicity edges}: edges $e_i$ with $a_i > 2$.
	    \end{compactitem}
    \end{compactitem}
    Suppose there are $\nsurp$ new-endpoint edges, $\nrecyc$ reused-endpoint edges, $\nunf$ forced edges, and $\nhigh$ high-multiplicity edges.
    As we process the edges in our pre-specified order, we record:
    \begin{compactitem}
    \item The labels of each vertex belonging to the first column set $S_1$ (at most $n^{|S_1|} = n^{h}$ choices).
	\item The edge type of every edge in the graph:
	    whether it is a new-endpoint edge,
	    a reused-endpoint edge,
	    a high-multiplicity edge,
	    or a forced or unforced return edge
	    (at most $5^{|E|}=5^{\wdth\hght}$ choices).
	\item For each new-endpoint edge, we record the label of its second endpoint (at most $n^{\nsurp}$ choices).
	\item For each reused-endpoint edge, we record the location of the first appearance of its second endpoint (at most $|V|^{\nrecyc} = (\wdth\hght)^{\nrecyc}$ choices).
	\item For each unforced return edge, we record the column index of the first appearance of that edge, and the index within that column of the label involved in the edge (at most $(R\wdth)^{\nunf}$ choices).
	\item For each high multiplicity edge $e_i$ we record the second endpoint of its first appearance (at most $n^{\nhigh}$ choices).
	For each consequent appearance of $e_i$, we record the column index of the first appearance of $e_i$, and the index within that column of the label involved in the high multiplicity edge (in total, at most $(R\wdth)^{\sum_{a_i>2} a_i}$ choices).
    \end{compactitem}

    \medskip
    \begin{claim*}
    The recorded information, along with $E$ and $\alpha$, suffices to reconstruct $G$.
    \end{claim*}
\begin{proof}
    We will prove this by induction.
    Our inductive claim is that in the $i$th step, we can reconstruct the labels of the $i$th column of vertices, $S_i$.

    For the first column, we have recorded all of the labels, so we have $S_1$.

    In any subsequent column, assuming we know $S_i$, we will process the edges in order, and determine their endpoint in $S_{i+1}$.
    For each edge, we can the determine from our edge information whether it is a new-endpoint, reused-endpoint, high-multiplicity, unforced return, or forced return edge.
    Depending on the type of edge we use different information to discern the label of its endpoint in $S_{i+1}$:

    \begin{compactitem}
    \item If we traversed a new-endpoint edge: we have recorded the label of the endpoint in $S_{i+1}$.
	\item If we traversed a reused-endpoint edge: we have recorded the position of the first appearance of the vertex's label, and we can look it up.
	\item If we traversed an unforced return edge: we have recorded the column index of the first appearance of the edge, as well as the index $I$ within that column of the label corresponding to this edge's endpoint in $S_i$.
	    We go to the column, and then choose the label in $S_{i+1}$ by finding the $I$th edge in that column with a label matching our edge's known label from $S_i$.
	\item If we traversed a forced return edge: there is only one choice for the second endpoint.
	\item If we traversed a high-multiplicity edge: we have recorded the column index of the other appearances of this edge.
	    If this is the first appearance of the edge, we have recorded the label of the second endpoint.
	    Otherwise, we have recorded the column in which this edge first appeared, as well as the index $I$ within that column of the label corresponding to this edge.
	    We go to the column, and then choose the label in $S_{i+1}$ by finding the $I$th edge in that column with a label matching our edge's known label from $S_i$.
    \end{compactitem}
    This proves the inductive claim.
\end{proof}

Thus,
    \begin{equation}
	|\cG^{\alpha,E}_{\wdth\times\hght}|
	\le \sum_{\substack{\nrecyc\\\nsurp\\\nunf}} 5^{\wdth\hght} \cdot  (\wdth\hght)^{\nrecyc}\cdot (R\wdth)^{\sum_{a_i>2} a_i + \nunf}\cdot n^{\nsurp + \nhigh + \hght}.
	\label{eq:rough-bound}
    \end{equation}
    All that remains is for us to translate between the above quantity, which is in terms of edge types, to our desired quantity in terms of the parameters $\ty,\wdth,\hght$.
    We do this by observing that
	\[
	    \nsurp = \ty - \nhigh - \nrecyc\mper
	\]
	This is because there are a total of $\ty$ distinct labeled edges, and of those, $\nhigh$ are high-multiplicity, and $\nrecyc$ do not introduce new labels.

	We use this observation to simplify $n$'s exponent, and we use the fact that $\sum_{\alpha_i>2} a_i +\nunf + \nrecyc \le \sum_{i} \alpha_i = \wdth\hght$ to simplify $\wdth$ and $R$'s exponents, giving us that
    \begin{equation}\label{eq:sub-bound}
	|\hat \cG_{\alpha}^\ty| \le
    (5R\wdth)^{\wdth\hght}\cdot n^{\hght} \sum_{\substack{\nrecyc\\\nsurp\\\nunf}}
	\left(\frac{\wdth\hght}{n}\right)^{\nrecyc} \cdot n^{\ty}.
    \end{equation}
    The number of possible combinations of values for $\nsurp$, $\nunf$, and $\nrecyc$ is at most $(\wdth\hght)^3$.
    Furthermore, because $\wdth\hght \le n$, from \pref{eq:sub-bound} we may conclude,
    \[
	|\cG^{\alpha,E}_{\wdth\times\hght}|
	\le (5R\wdth)^{\wdth\hght} (\wdth\hght)^3 \cdot n^{\ty + \hght}
    \]
    as desired.
\end{proof}

\medskip

%===================================

We now prove \pref{lem:matching-probability-2}, which gives us a bound on the probability that we sample an even hypergraph cycle with the correct edge multiplicity from a simple edge cycle.
Again, the proof of \pref{lem:matching-probability-2} is different from the proof of \pref{lem:matching-columns}, and is more similar to the proof of \pref{prop:multilin-graph-bound} (although it is already quite different from the proof of \cite{FK81}).

    \restatelemma{lem:matching-probability-2}

\begin{proof}
    We will count the number of possible even $H$ one can sample from $E$ with at most $\tau$ unique hyperedges by encoding each such $H$ uniquely as a string, then counting the number of strings.

    First, we fix an ordering on the edges of $E$, ordering them first by column, then arbitrarily within each column.
    Now, define a \emph{new} hyperedge to be a labeled hyperedge which we have never seen before, and define an \emph{old} hyperedge to be a hyperedge which has already been seen.
    Let $H$ be some even hypergraph sampled from $G$ with at most $M$ unique labeled hyperedges.
    We encode $H$ in a string as follows.
    We will process the hyperedges of $H$ one at a time, ordering hyperedges by the first simple edge they contain.
    \begin{compactitem}
	\item For every hyperedge encountered in $H$, record whether it is new or old ($2^{|H|} = 2^{\wdth\hght}$ options).
	\item For every new hyperedge encountered: record the indices of the $2,...,d$ simple edges that it contains ($\left((d\hght)^{(d-1)}\right)^{\tau}$ options).
	    The identity of the first edge will be obvious from the edge ordering.
	\item For every old hyperedge encountered, record the column index $j$ of its first appearance ($\wdth^{\wdth\hght-\tau}$ choices).
	If any of the simple edges $e_{i_1},\ldots,e_{i_d}$ appear with multiplicity $>1$ in the old ($j$th) column, record the indices of those edge within the identical edges in that column (at most $(R^{d})^{\wdth\hght -\tau}$ choices, since no simple edge can appear more than $R$ times in a column).
    \end{compactitem}
    We claim that given $V,E$, and this encoding, we can uniquely recover $H$.
    We process the simple edges in our specified order.
    If the edge is contained in a new hyperedge, we can deduce the other simple edges belonging with it from what we recorded.
    If the edge is contained in an old hyperedge, we know the column in which the hyperedge first appears, and we can determine the other edges in the group by looking up the edge in the previous column--if there are multiple copies of the edge in the old column, we have recorded which copy to look up.
    Furthermore, if the grouping is insufficient due to multiplicities within this column, we have recorded the relative indices of the relevant edges.

    The number of such strings is at most
    \begin{equation}
	(\text{\# encodings})
	\le 2^{\wdth\hght} \cdot (d\hght)^{(d-1)\tau} \cdot (R^{d}\wdth)^{\wdth\hght - \tau}\label{eq:blah1}
    \end{equation}
    giving an upper bound on the number of $H$ we can sample from $(V,E)$ with at most $\tau$ distinct edges.
    There are a total of
    \begin{equation}
	(\text{\# sampleable graphs})
	= \left(\frac{1}{\hght!}\prod_{j=0}^{\hght-1} \binom{d(\hght -j)}{d}\right)^{\wdth}
     = \left(\frac{(d\hght)!}{(d!)^{\hght}\hght!}\right)^{\wdth} \label{eq:blah2}
    \end{equation}
    possible hyperedge graphs sampleable from $(V,E)$, and from this we have that
    \begin{align*}
	\Pr(H\text{ has }\tau \text{ edges } | (V,E) \text{ even })
    \le \frac{(\text{\# encodings})}{(\text{\# sampleable graphs})}
	\le \frac{(2d^{d} e^d\wdth R^{d})^{\hght\wdth}}{(d\hght)^{(d-1)(\hght\wdth - \tau)}}
    \end{align*}
where we have combined \pref{eq:blah1} with \pref{eq:blah2} and applied Stirling's inequality.
Our conclusion follows.
\end{proof}

%==========================================================

Now we prove a lemma that bounds the number of even $k$-hyperedge configurations with $\tau$ hyperedges, each paired and sharing a center vertex, sampleable from an even $2d$-hyperedge configuration by pairing and labeling.
\restatelemma{lem:odd-match}
\begin{proof}
We will count the number of such $H$ by encoding each instance uniquely as a string.
Fix an order on the hyperedges, first in column order.
\begin{compactitem}
    \item A \emph{new} hyperedge is a hyperedge which introduces a new center vertex.
    \item A \emph{reuse} hyperedge is a hyperedge which we see for the first time, and whose center vertex is the first of its type in its column, but which reuses a center vertex from a previous column.
    \item A \emph{sharing} hyperedge is a hyperedge which we see for the first time, but which shares a center vertex with another hyperedge in its own column.
    \item A \emph{return} hyperedge is a hyperedge which we see for the second or later time.
\end{compactitem}
    Our encoding is as follows:
    \begin{compactitem}
	\item In the first $\hght\wdth$ positions, we record the type of every hyperedge we see ($4^{\hght\wdth}$ choices).
	\item In the next $\nsurp$ positions, we record the labels of new hyperedges ($n^{\nsurp}$ choices).
	\item In the next $\nrecyc$ positions, we record the position of the first appearance of the reused label ($(\hght\wdth/2)^{\nrecyc}$ choices)
	\item In the next $\nshare$ positions, we record the partner of the sharing edge within the column ($(\hght/2)^{\nshare}$ choices, since there cannot be more than $\hght$ new labels in a column).
	\item In the next $\nret$ positions, we record the column of the first appearance of the hyperedge (a total of $\wdth^{\nret}$ choices), the index of the previous occurrence of the hyperedge among hyperedges with the same labels in the previous column (a total of $R^{\nret}$ choices), and the index of the hyperedge's center vertex among center vertices of the same label within the current column (a total of $R^{\nret}$ choices).
	\item We record the permutation of the middle labels in each column ($(\frac{\hght}{2}!)^\wdth$ choices).
    \end{compactitem}
    Given this information, we can uniquely reconstruct an $H$ from $G$.
    For every surprise hyperedge we encounter, we have recorded the center label.
    For every recycle hyperedge we encounter, we can determine the center label by looking at the previous occurrence.
    For every sharing hyperedge, we can determine its partner.
    For every return hyperedge, we can determine the label of the center vertex by looking at the previous occurrence, and we can determine partnership by knowing the index of the center label's occurrence within the column.

    We thus have
    \begin{align*}
	(\# H )
	&\le \sum_{\substack{\nsurp\\\nrecyc}} \left(\frac{\hght}{2}!\right)^{\wdth}(4\wdth)^{\wdth\hght}\cdot R^{2\nret} \cdot n^{\nsurp}\cdot \hght^{\nrecyc + \nshare}\mper
    \end{align*}
    We use some observations about these quantities to simplify the above expression.
    We have that
    \[
	\tau = \nsurp + \nshare + \nrecyc\mcom
    \]
    since every hyperedge must appear for the first time.
    Furthermore, we have that $\nsurp \le \nshare$, since every surprising hyperedge must be paired with a sharing hyperedge, since it introduces a new vertex label which hasn't been seen before, so it's partner must be a sharing edge since we forbid two hyperedges with the same labels to share a center vertex.
    It follows that
    \[
	\nsurp \le \tau/2\mper
    \]

    Putting these together,
    \begin{align*}
	(\# H )
	&\le \left(\frac{\hght}{2}!\right)^{\wdth}(4\wdth R^2)^{\hght\wdth}\sum_{\substack{\nsurp\\\nrecyc}}  n^{\nsurp}\cdot \hght^{\tau - \nsurp}\\
	&= \left(\frac{\hght}{2}!\right)^{\wdth}(4\wdth R^2 )^{\hght\wdth} \hght^{\tau} \cdot \sum_{\substack{\nsurp\\\nrecyc}}  \left(\frac{n}{\hght}\right)^{\nsurp}\\
	&\le \left(\frac{\hght}{2}!\right)^{\wdth}(2\hght \wdth)^2 (4\hght R^2)^{\hght\wdth} \cdot (\hght n)^{\tau/2}\mcom
    \end{align*}
    where in the last line we have assumed that $\hght < n$, and have used that $\nsurp$ and $\nrecyc$ take on at most $\tau < \wdth\hght/2$ values.
    The conclusion follows.
\end{proof}

\section{Strong Refutation for All CSPs}\label{sec:sat}

In this section, we consider the problem of refuting Boolean CSP's with arbitrary predicates.

\begin{problem}[Refuting CSP's with predicate $P$]
    Let $P:\{\pm 1\}^k \to \{0,1\}$ be a predicate on $k$ variables.
    Then we sample a \emph{random instance of CSP-$P$}, $\Phi$, with clauses $C_1,\ldots,C_m$, as follows:
    for each $I \in [n]^k$:
	\begin{compactitem}
	    \item With probability $p$, sample a uniformly random $\sigma \in \{\pm 1\}^k$ and add the constraint $P(x_I \oplus \sigma) = 1$ to $\Phi$ as clause $C_I$, where $\oplus$ denotes the entry-wise product and $x_I$ denotes the ordered subset of variables $x_i$ for each $i \in I$.
\end{compactitem}
The problem of \emph{strongly refuting CSP-$P$} is to devise an
algorithm that given an instance $\Phi$ sampled as above, with high
probability over $\Phi$, outputs a certificate that
\[
    \opt(\Phi) \leq 1 - \gamma
\]
for an absolute constant $\gamma>0$.
\end{problem}

Our result is the following:
\begin{theorem}\label{thm:csp-refutation}
    Let $\Phi$ be a random instance of a $k$-CSP with predicate $P$, with clause density $m/n \ge \tilde O(n^{(k/2 - 1)(1-\delta)})$.
    Then with high probability over the choice of $\Phi$, there is a spectral algorithm which strongly refutes $\Phi$ in time $\exp(\tilde O(n^{\delta}))$, certifying that
    \[
	\opt(\Phi) \le \E_{x \sim \{\pm 1\}^k}[P(x)] + \epsilon
    \]
for any constant $\epsilon > 0$.
    Furthermore, the degree-$O(n^{\delta})$ \sos relaxation also certifies this bound.
\end{theorem}

We will employ the framework of Allen et al. \cite{AOW15} to prove that we can strongly refute any $k$-CSP with predicate $P$ at densities as low as $\tilde O(1)$, given sufficient time.
The strategy is as follows: given some random $k$-CSP on $x\in\{\pm 1\}^n$ with clauses $C_1,\ldots,C_m$, $C_i:\{\pm 1\}^k \to \{\pm 1\}$:
\begin{compactitem}
\item Expand $C_1(x),\ldots,C_m(x)$ using the Fourier expansion.
\item Split the Fourier expansions of $C_1,\ldots,C_m$ into XOR instances.
\item Refute each XOR instance.
\end{compactitem}

Because a $k$-CSP predicate has a Fourier expansion of degree at most $k$, the above strategy in combination with our $k$-XOR refutation results will allow us to tightly strongly refute any $k$-CSP in time $\exp(n^{\delta})$ at densities $\ge \tilde O(n^{(k/2-1)(1-\delta)})$.
However, as a result of the work of \cite{AOW15}, we are able to show that for predicates satisfying some additional properties, it is possible to strongly refute more quickly at lower densities.
We elaborate further:
\begin{definition}
    Let $1 \le t \le k$.
    A predicate $P:\{\pm 1\}^k \to \{0,1\}$ is $\delta$-far from $t$-wise supporting if every distribution $\cD$ on $\{\pm 1\}^k$ which has uniform marginals on all subsets of $t$ variables only satisfies $P$ with probability at most $1-\delta$, i.e.
    \[
	\E_{x\sim \cD}[P(x)] \le 1 - \delta.
    \]
\end{definition}

Allen et al. give the following characterization of $\delta$-far from $t$-wise supporting predicates:
\begin{theorem}[Lemma 3.16 and Theorems 4.9 and 6.6 of \cite{AOW15}]\label{thm:od}
Let the predicate $P:\{\pm 1\}^k \to \{0,1\}$ be $\delta$-far from $t$-wise supporting, for $0 \le t \le k$.
    Then there exists a multilinear polynomial $Q:\{\pm 1\}^t \to \R$ such that $\E_{x\sim \{\pm\}^t}[Q(x)] = 0$ and $P(x) \le (1-\delta) + Q(x)$ for any $x \in \{\pm 1\}^k$.
Furthermore, $Q$ can be obtained by solving a linear program of constant size, and there is a degree-$k$ \sos proof that $\pE[P(x)] \preceq (1-\delta) + \pE[Q(x)]$.
\end{theorem}
We will use this theorem to extend the results of \cite{AOW15} for $\delta$-far from $t$-wise independent predicates below the spectral threshold.
\begin{theorem}\label{thm:twise}
Let the predicate $P:\{\pm 1\}^k \to \{0,1\}$ be $\delta$-far from $t$-wise supporting, for $0 \le t \le k$ and a constant $\delta > 0$.
    Let $\Phi$ be a random instance of a $k$-CSP with predicate $P$, with clause density $m/n \ge \tilde O(n^{(t/2 - 1)(1-\delta)})$.
    Then with high probability over the choice of $\Phi$, there is a spectral algorithm which strongly refutes $\Phi$ in time $\exp(\tilde O(n^{\delta}))$, certifying that
    \[
	\opt(\Phi) \le 1-\delta + \epsilon
    \]
for any constant $\epsilon > 0$.
    Furthermore, the degree-$O(n^{\delta})$ \sos relaxation also certifies this bound.
\end{theorem}

We now prove \pref{thm:csp-refutation}, and then below we will describe the mild changes needed to prove \pref{thm:twise}.
We will utilize our own \pref{thm:maindlin}, as well as the following theorem which has appeared in \cite{AOW15} (and also partially in \cite{BM15}).
We cite the exact form of the theorem given in \cite{AOW15}.
\begin{theorem}[\cite{AOW15}, Theorem 4.1]\label{thm:AOW}
    For $k \ge 2$, $q \ge n^{-k/2}$, let $\{w_S\}_{S \in [n]^k}$ be independent random variables such that for each $S \in [n]^k$,
    \begin{align*}
	\E[w_S] = 0,\qquad
	\Pr[w_S\neq 0] \le q,\qquad\text{and}\qquad
	|w_S| \le 1.
    \end{align*}
    Then there is an efficient algorithm certifying that
    \[
	\left|\sum_{S \in [n]^k} w_S \prod_{i \in S} x_i\right| \le 2^{O(k)} \sqrt{q} n^{3k/4} \log^{3/2} n
    \]
    for all $x$ with $\|x\|_{\infty} \le 1$ with high probability.
\end{theorem}
\begin{remark}
    In \cite{AOW15}, the theorem appears without the absolute value--however  the statement for the absolute value is implied by the fact that the negated variables $-w_S$ also satisfy all of the constraints.
\end{remark}

Given the above theorem and our results for refuting XOR instances (\pref{thm:maindlin}), the result for arbitrary binary CSPs follows easily.
\begin{proof}[Proof of \pref{thm:csp-refutation}]
    Let the Fourier expansion of $P$ on $y\in\{\pm 1\}^k$ be $P(y) = \sum_{S \subseteq [k]} \hat P(S) \cdot \chi_S(y)$.
    Let $\Phi$ have constraints $C_1,\ldots,C_m$, chosen independently on each $I \in [n]^k$ with probability $p$, so that the constraint $C_I$ asserts that $P(x_I \oplus \sigma^I)=1$ for a uniformly chosen signing $\sigma^I \in \{\pm 1\}^k$.
    We have that
    \begin{align*}
	P_{\Phi}(x)
	=  \sum_{I \in [n]^k} \Ind(C_I \in \Phi)\cdot P(x_I \oplus \sigma^I)
	&= \sum_{I \in [n]^k} \Ind(C_I \in \Phi) \cdot \sum_{S \subseteq [k]} \hat P(S) \cdot \prod_{i \in I} x_i \cdot \prod_{i \in S(I)} \sigma_i^{I}\mcom
	\intertext{where we have used $\oplus$ to denote the entry-wise product and $S(I)$ to denote the entries of $I$ corresponding to the subset $S$ of $[k]$.
	We will move the sum over ordered subsets $S \subseteq k$ outwards, then simplify further}
	&= \sum_{S \subseteq k} \sum_{I \in [n]^k} \Ind(C_I \in \Phi) \cdot \hat P(S) \cdot \prod_{i \in S(I)} x_i \cdot \prod_{i \in S(I)} \sigma^I_i\\
	&= \hat P (\emptyset)  + \sum_{\substack{S \subseteq k\\|S| \ge 1}} \hat P(S) \sum_{I \in [n]^k} \Ind(C_I \in \Phi) \cdot  \prod_{i \in S(I)} x_i \cdot \prod_{i \in S(I)} \sigma^I_i\mper
    \end{align*}
    Now, we will see that for each $S \subseteq [k]$, $|S| \ge 1$, we have a random weighted XOR instance $\Psi_S$ on $|S|$ variables.
    Letting $b_L^I \defeq \prod_{\ell \in L}\sigma^I_{\ell}$, we define
    \begin{align*}
	\Psi_S(x)
	&\defeq \sum_{I \in [n]^k} \Ind(C_I \in \Phi) \cdot  \prod_{i \in S(I)} x_i \cdot \prod_{i \in S(I)} \sigma^I_i\\
	&= \sum_{L \in [n]^S} \prod_{i\in L} x_i \cdot \left(\sum_{J \in [n]^{k\setminus S}} \Ind(C_{J \cup L} \in \Phi) \cdot  b_L^I\right)\mcom
    \end{align*}
where we have abused notation by allowing $J \cup L$ to denote the ordered multiset with $L$ in the exact positions corresponding to $S$ and $J$ in the positions corresponding to $k \setminus S$.
    Furthermore, by definition,
    \[
    \frac{P_{\Phi}(x)}{m}
    = \hat P(\emptyset) + \sum_{\substack{S \subseteq k\\|S| \ge 1}} \hat P(S) \cdot \frac{\Psi_S(x)}{m}
    \le \hat P(\emptyset) + \sum_{\substack{S \subseteq k\\|S| \ge 1}} \hat P(S) \cdot \frac{|\Psi_S(x)|}{m}
    \]
    and so bounding the values of the $\Psi_S$ suffices to get a bound on the value of $\Phi$.

We list some properties of  $\Psi_S$.
For convenience, we now use the notation $\chi_L(x) \defeq \prod_{\ell \in L} x_\ell$ and the notation $x_L$ to denote a string of elements of $x$ indexed by $L$.
For each $S\subseteq k, S \neq \emptyset$, $\Psi_S$ is an instance of $|S|$-XOR with independent constraints on each $\chi_L(x)$ for $L \in [n]^{|S|}$--the independence is because the constraints for $\chi_L(x)$ depend only on the presence of clauses in $C_I \in \Phi$ for $I\in[n]^k$ including $x_L$ in the positions corresponding to $S$.
The weight on each $\chi_L(x)$ is distributed according to a sum of $n^{k-|S|}$ independent random variables, each of which is $0$ with probability $1-p$, and uniformly $\pm 1$ with probability $p$.
For convenience, call the distribution over such sums $\hat k(n^{k-|S|},p)$, so that the coefficient $c_L$ of $\chi_L(x)$ is distributed according to $c_L \sim \hat k(n^{k-|S|},p)$.

We now perform a case analysis on $p$ and $|S|$, which allows us to bound the contributions of each $\Psi_S(x)$ individually.

    \paragraph{Case 1: $pn^{k-|S|} \ge 1$}
If $pn^{k-|S|} \ge 1$ then with high probability, every constraint $c_L$ has $|c_L| \le O(\sqrt{p n^{k - |S|} \log n})$ (where we have combined \pref{lem:dhat} with a union bound).
	Furthermore the $c_L$ are distributed symmetrically about $0$, and they are nonzero with probability at most $1 \le pn^{k-|S|}$.
	\begin{itemize}[parsep=2pt,partopsep=2pt, topsep=5pt,itemsep=2pt,leftmargin=1em]
	\item If $|S| = 1$, we have that with high probability,
	    \[
		\frac{|\Phi(x)|}{m}
		\le \frac{n \cdot \max_{\ell \in [n]} |c_\ell|}{m}
		\le \frac{n \cdot O(\sqrt{pn^{k-1}\log n})}{\Theta(pn^k)}
		\le O\left(\sqrt{\frac{\log n}{pn^{k-1}}}\right)\mcom
	    \]
	    where we have applied a Chernoff bound to use that $m = \Theta(pn^k)$.
	    By our assumption on the clause density, $p \ge n^{-(k-1)}\cdot \polylog n$, and therefore it follows that $|\Psi_S(x)|/m = o(1)$.
	\item Otherwise, if $|S| \ge 2$, we can divide each $c_L$ by $\beta = O(\sqrt{pn^{k-|S|}\log n})$ to obtain a polynomial with coefficients bounded in absolute value by $1$, with independent symmetrically distributed coefficients and probability at most $1 \le pn^{k-|S|}$ of being nonzero.
	    We can thus apply \pref{thm:AOW} to get that with high probability we can certify in polynomial time that,
	\[
	    \frac{|\Psi_S(x)|}{\beta}
	    \le O(n^{3|S|/4}\log^{3/2}n),
	\]
	Implying that with high probability,
	\begin{align*}
	    \frac{|\Psi_S(x)|}{m}
	    \le \frac{\beta \cdot O(n^{3|S|/4}\log^{3/2} n)}{\Theta(pn^k)}
	    &\le \frac{O(\sqrt{pn^{k-|S|}\log n}) \cdot O(n^{3|S|/4}\log^{3/2} n)}{\Theta(pn^k)}\\
	    &\le O\left( \frac{\log^{2} n}{p^{1/2}n^{k/2 - |S|/4}}\right)
	    \le O\left(\frac{\log^{2} n}{n^{|S|/4}}\right)\mper
	\end{align*}
	where the last inequality follows by the assumption that $pn^{k-|S|} \ge 1$.
    \end{itemize}

	\paragraph{Case 2: $pn^{k-|S|} < 1$}
If $pn^{k-|S|} < 1$, then with high probability all  $|c_L| \le O(\log n)$ (where we have combined \pref{lem:dhat} with a union bound).
    We now split into cases in which we can apply \pref{thm:AOW} and cases in which we must apply \pref{thm:maindlin}.
	\begin{itemize}[parsep=2pt,partopsep=2pt, topsep=5pt,itemsep=2pt,leftmargin=1em]
	\item If $|S| < k$ and  $p \ge n^{-|S|/2}$, then again letting $\beta = O(\log n)$, we can divide $\Psi_S$ by $\beta$ to obtain a polynomial with coefficients that are symmetrically distributed about $0$, bounded by $1$ in absolute value, and are nonzero with probability at most $pn^{k-|S|}$.
	By \pref{thm:AOW} and by a Chernoff bound on $m$, it follows that we can certify in polynomial time that
	\begin{align*}
	    \frac{|\Psi_S(x)|}{m}
	\le \frac{\beta \cdot O(\sqrt{pn^{k-|S|}} \cdot n^{3|S|/4}\log^{3/2} n)}{m}
    &\le \frac{O(\log n) \cdot O(\sqrt{pn^{k-|S|}} \cdot n^{3|S|/4}\log^{3/2} n)}{\Theta(pn^k)}\\
	    &\le O\left( \frac{\log^{3/2} n}{p^{1/2}n^{k/2 - |S|/4}}\right)
	    \le O\left(\frac{\log^{3/2} n}{n^{(k-|S|)/2}}\right)
	= o(1),
	\end{align*}
	where the second-to-last inequality follows by our assumption that $pn^{|S|/2}\ge 1$.
	\item If $|S| = k$, then $\Psi_S$ is a $k$-XOR instance with each constraint present with probability $p$.
	By \pref{thm:maindlin}, we can certify in time $\exp(\tilde O(n^{\delta}))$ that $\frac{|\Psi_S|}{m}\le \gamma$ for any constant $\gamma > 0$.
\item
    If $|S| < k$ and  $p < n^{-|S|/2}$, we must apply \pref{thm:maindlin}.
    We must modify the instances slightly first, since \pref{thm:maindlin} applies to unweighted instances.

    To obtain an unweighted instance, we split $\Psi_S$ further into $r = \log^2 n$ instances, $\Psi_S^{(1)},\ldots,\Psi_S^{(r)}$.
    We split as follows:
    let $c_L^{(i)}$ denote the coefficient of $\chi_L(x)$ in $\Psi_S^{(i)}$.
    For each nonzero $c_L$, we choose $|c_L|$ uniformly random indices $i_1,\ldots, i_{|c_L|} \in [r]$, and assign $c_L^{(i_j)} = \frac{c_L}{|c_L|}$ for each $j = 1,\ldots, |c_L|$
    (recall that with high probability $|c_L|\le O(\log n) < r$).
    Let $m_i$ be the number of constraints in $\Phi^{(i)}_S$, so that we have $m \ge \sum_{i} m_i$ (since $c_L$ may be a sum of negative and positive constraints from the full instance $\Phi$).

    First, note that
    \[
	\frac{|\Psi_S(x)|}{m} \le \frac{\left|\sum_{i = 1}^r \Psi_S^{(i)}(x)\right|}{\sum_{i} m_i} \le \max_{i\in[r]} \frac{|\Psi_S^{(i)}(x)|}{m_i}\mper
    \]

    It remains to argue that each instance $\Psi_S^{(i)}$ has bounded value with high probability.
    Towards this, consider the properties of $\Psi_S^{(i)}$.
    First, we note that the constraints of $\Psi_S^{(i)}$ are independent of one another and are distributed symmetrically about zero--this is because the $c_L$ are independent of one another and symmetrically distributed about zero.
    Furthermore, we have that each $c_L^{(i)}$ is nonzero with probability $\hat q$:
    \begin{align*}
	\hat q
	\defeq \Pr[c_{L}^{(i)} \neq 0]
	&= \sum_{j=1}^{r}\Pr[|c_L| = j,~ i ~\text{chosen}]\\
	&= \sum_{j=1}^{r} \Pr[i ~\text{chosen} ~|~ |c_L| = j] \cdot \Pr[|c_L| = j]\\
	&= \sum_{j=1}^{r} \frac{j}{r} \cdot \Pr[|c_L| = j]\mcom
    \end{align*}
    where we have taken the sum up to $r$ because we implicitly condition on $|c_L| \le O(\log n)$ (which occurs with high probability).
    We thus have that
    \[
	\hat q
	= \sum_{j=1}^{r} \frac{j}{r} \cdot \Pr[|c_L| = j]
	\le \Pr[|c_L| > 0 ] \leq p n^{k-|S|},
    \]
    and also that
    \[
	\hat q
	= \sum_{j=1}^{r} \frac{j}{r} \cdot \Pr[|c_L| = j]
	\ge \frac{1}{r} \cdot \Pr[|c_L| > 0]  \geq \frac{1}{\log^2 n} \cdot \frac{pn^{k-|S|}}{2} \mper
    \]
The last inequality follows by observing that with probability at
least $pn^{k-|S|}$, at least one of the chosen constraints in $\Phi$ contributes
to $c_L$, and conditioned on this event, the contributions to $c_L$ sum to zero with probability at most $1/2$.
    Therefore, $\hat q = \delta \cdot pn^{k-|S|}$ for some $\delta \in [\frac{1}{2\log^2 n}, 1]$.

    Thus, each ${\Phi^{(i)}_S}'$ is a random $|S|$-XOR instance in which each clause is revealed with probability $\hat q$.

    From \pref{thm:maindlin}, with high probability we can refute a random $|S|$-XOR instance in which each clause is present with probability $\hat q$ in time $\exp(\tilde O(n^{\delta}))$ so long as $\hat qn^{|S|-1} \ge \tilde O(n^{(|S|/2-1)(1-\delta)})$, certifying that the instance satisfies at most $\frac{1}{2} + \gamma + o(1)$ clauses for any constant $\gamma > 0$.
This condition on $\hat qn^{|S|-1}$ holds by our assumption that $p \ge \tilde O(n^{-k(1+\delta)/2 +\delta})$ (as long as we make the correct adjustments of logarithmic factors on $p$ to account for the value of $\hat q$).
We can also certify with high probability that the fraction of satisfied constraints is at least $\frac{1}{2} - \gamma$ for any constant $\gamma$, by applying the same argument with the negations of the $c_L^{(i)}$.

Thus, in time $\exp(\tilde O(n^{\delta}))$, with high probability we can certify that $\frac{| {\Phi^{(i)}_S(x)}|}{m_j} \le \gamma $, implying by a union bound over $i\in[r]$ that $\frac{|\Psi_S(x)|}{m} \le \gamma$.
    \end{itemize}
    Using this case analysis, we have that
    \[
    \frac{P_{\Phi}(x)}{m}
    \le \hat P(\emptyset) + \gamma \cdot \sum_{\substack{S \subseteq k\\|S| \ge 1}} \hat P(S)
    \]
    and since $\hat P(S)$ can depend only on $k$, for large enough $n$, with high probability over $\Phi$ we can certify that $\frac{P_{\Phi}(x)}{m} \le \hat P(\emptyset) + \gamma'$ for any constant $\gamma'$ in time $\exp(\tilde O(n^{\delta}))$ when $m/n \ge \tilde O(n^{(k/2-1)(1-\delta)})$.

    The same conclusion holds in the degree-$O(n^{\delta})$ \sos relaxation, as every step of this proof holds within the \sos proof system (because \pref{thm:AOW} and \pref{thm:maindlin} hold within the \sos proof system).
\end{proof}

The proof of \pref{thm:twise} proceeds almost identically, except that instead of using the Fourier expansion of the predicate $P$, we use the degree-$t$ polynomial $Q(x)$ given by the work of Allen et al. \cite{AOW15} (\pref{thm:od}).
Because $P(x) \le (1-\delta) + Q(x)$, and since $Q(x)$ has no constant term, the proof we applied to the degree $\ge 1$ terms of the Fourier expansion of $P$ applies to $Q(x)$, and this completes the proof.

\begin{lemma}\label{lem:dhat}
    For any $0\le q\le1$, define $\hat k(N,q)$ to be the distribution over scalars such that $X\sim \hat k(N,q)$ is a sum of $N$ independent variables, each $0$ with probability $1-q$, $-1$ with probability $q/2$, and $1$ with probability $q/2$.
    Then if $Nq \ge 1$, for any constant $c$, there exists a constant $c'$ such that
    \[
	\Pr_{X\sim \hat k(N,q)}[|X| \le \sqrt{c' Nq \log N}] \ge 1- N^{-c}\mcom
    \]
    and if $Nq < 1$, for any constant $c$ there exists a constant $c'$ such that
    \[
	\Pr_{X\sim \hat k(N,q)}[|X| \le c'\log N] \ge 1- N^{-c}\mper
    \]
\end{lemma}
\begin{proof}
    By definition, for $X\sim \hat k(N,q)$, $X = \sum_{i=1}^N x_i$ for $x_i$ distributed according to $\hat k(1,q)$.

    It is easy to see that $\E[X] = 0$, and to calculate that $\Var(X) = qN$, and we note also that $|x_i| \le 1$.
   Therefore when $Nq \ge 1$, from a Bernstein inequality we have that
   \[
	\Pr[|X| \ge t] \le \exp\left(\frac{-t^2/2}{t/3 + qN}\right),
   \]
   and taking $t = \sqrt{4cqN\log N}$, and using that $qN \ge 1$, we have the desired result.

   When $qN < 1$, we apply the same bound with $t = 4c\log N$ to obtain our second result.
\end{proof}

\section{Sum-of-Squares Algorithms}\label{sec:sos}
In this section, we use our spectral algorithms to certify \sos upper bounds.
\subsection{Background}

Suppose we wish to maximize an $n$-variate polynomial $f_{obj}(x)$ over $x \in \cC \subset \R^{n}$, where $\cC$ is some subset of $\R^n$ defined by polynomial constraints.
This problem is clearly \np-hard in general.
The sum-of-squares hierarchy is a hierarchy of semidefinite programming relaxations for such polynomial optimization problems.
The $d$-round \sos relaxation is a program with variables $X_S$ for each $S \in \{\emptyset \cup [n]\}^{2d}$.
For each $S \in\{\emptyset \cup [n]\}^{2k}$, the variable $X_S$ is a relaxation for the monomial $\prod_{i\in S} x_i$.
\begin{definition}[standard degree-$d$ \sos constraints]\label{def:sosconst}
The basic constraints for the $d$-round sum-of-squares relaxation are:
\begin{align}
     X_{\emptyset} &= 1 \qquad \text{representing the constant term} \\
    X_{A,B} &= X_{C,D} \qquad  \forall A,B,C,D \in \{\emptyset \cup [n]\}^{k}~ \text{if}~ (A,B) = (C,D) ~ \text{as unordered multisets} \label{eq:symconst}\\
    \cX &\succeq 0 \label{eq:psdness}
\end{align}
where $\cX$ is a $\{\emptyset \cup [n]\}^d \times \{\emptyset \cup [n]\}^d$ matrix whose $(A,B)$th entry contains $X_{A,B}$.
We refer to this set of constraints as $\sos_d$.
    If there are additional polynomial constraints $g_1(x) = 0,\ldots,g_m(x) = 0$, then we also add the constraints
    \[
	X_{S} \circ g_j(X) = 0 \quad \forall j, S: \deg(g_j) +|S| \le 2d,
    \]
where the notation $\circ$ is used to mean replacing each variable $X_{T}$ appearing in $g_j$ with the variable $X_{S,T}$.
\end{definition}
A useful alternate characterization of \pref{eq:psdness} is that for any polynomial $f$ of degree at most $d$, we have that $f^2(X) \ge 0$, where $f^2(X)$ is the function given by evaluating the coefficients of $f^2$ at the stand-in monomials given by the variables of the program.
These are all constraints that any true polynomial solution satisfies.

We define the linear operator $\pE:\R[x]^{\le 2d} \to \R$, which maps any monomial of degree at most $2d$ to the \sos variable identified with it.
It is sometimes instructive to think of the variable $X_{S}$ as a \emph{pseudoexpectation} or a \emph{pseudomoment} of the monomial $\prod_{i\in S} x_i$ over feasible solutions which maximize the objective function:
\[
    X_{S} = \pE_{x~\text{maximizing}~f_{obj}}\left[\prod_{i \in S} x_i\right].
\]
Intuitively, the constraints of the SDP force the solution to behave somewhat like the moments of a probability distribution over feasible maximizing solutions, although they needn't correspond to the moments of a true distribution, hence the term \emph{pseudomoment}.
See e.g. \cite{Barak14} for more background.

\subsection{Relaxations for tensor norm and $k$-XOR}
The natural \sos relaxations for computing the tensor norm and for maximizing $k$-CSPs are very similar to each other.
Both correspond to polynomial maximization problems, where the constraint is that the maximizing solution $x \in \R^n$ lie on the unit sphere or on the Boolean hypercube.
Save for these ``normalization'' constraints and the natural SDP constrains $\sos_d$, there are no other constraints.

\begin{definition}[$d$-round \sos relaxation for tensor norm]
    Given an order-$k$ tensor $\bT$, for any $d \ge \lceil k/2\rceil$, the $d$-round \sos relaxation for the injective tensor norm is
\begin{align}
    & \max~ \pE\left[\iprod{\bT,x^{\tensor k}}\right]\nonumber\\
    &\quad s.t. \quad
    \pE\left[\sum_{i \in [n]} x_i^2\right] = \sum_{i \in [n]} X_{i,i} = 1\mcom \label{eq:normconst}
\end{align}
With the addition of the standard $d$-round \sos constraints.
\end{definition}

\begin{definition}[$d$-round \sos relaxation for $k$-XOR]
    Given an instance $\Phi$ of $d$-XOR with constraint tensor $\bT_{\Phi}$ defined as described in \pref{sec:lin}, for any $d \ge \lceil k/2\rceil$, the $d$-round \sos relaxation is given by
\begin{align}
    &\max~\pE\left[\iprod{\bT_{\Phi},x^{\tensor k}}\right]\nonumber \\
    &\quad s.t. \quad
    \pE\left[x_i^2\right] = 1 \qquad \forall i \in [n]\mcom \label{eq:booleanity}
\end{align}
With the addition of the standard $d$-round \sos constraints.
\end{definition}

\subsection{Bounds for tensor norm}
In this subsection, we show how bound the objective value of the $d$-round \sos relaxation for a polynomial optimization problem when $\pE(\sum_{i} x_i^2)$ is known, in terms of the operator norm of a specific matrix.
We will use $\succeq$ and $\preceq$ to denote inequalities that are sum-of-squares identities.

We will require the use of the following lemma, which is standard in \sos-proofs.
\begin{lemma}[\sos matrix inner product]\label{lem:sos-matrix}
    Let $M$ be an $[n]^{d} \times [n]^d$ matrix, and let $\pE$ be a degree-$2d$ pseudoexpectation.
    Then
    \[
	\pE\iprod{x^{\tensor 2d}, M} \preceq \pE[\|x\|^{2d}] \cdot \|M\|.
    \]
\end{lemma}
\begin{proof}
    By assumption, $\lambda \cdot I - M \succeq 0$, and therefore the expression $I - M$ can be written as a sum-of-squares of degree at most $2d$.
    We thus have that
    \begin{align*}
	0
	& \preceq \pE \left[ \iprod{x^{\tensor 2d}, \lambda \cdot I - M}\right]\\
\pE \left[ \iprod{x^{\tensor 2d}, M}\right]
	& \preceq
	\lambda \cdot \pE \left[ \iprod{x^{\tensor 2d}, I}\right]
	~= \lambda \cdot \pE[\|x\|^{2d}]\mcom
    \end{align*}
    as desired.
\end{proof}

We will also make use of standard \sos versions of the Cauchy-Schwarz Inequality, and H\"older's Inequality, proofs of which can be found in \cite{BBHKSZ12}, for example.
Additionally we will use the following fact, which can be proven by induction:
\begin{fact}[\sos-Convexity]
    For any $d \in \N$, if $\pE$ is a degree-$2dk$ pseudodistribution and $f(x)$ is a polynomial of degree at most $k$, then
    \[
	\pE[f(x)]^{2d} \preceq \pE[f(x)^{2d}]\mper
    \]
\end{fact}

\begin{proposition}\label{prop:sos-even}
    Let $\bT$ be an order-$k$ tensor for even $k = 2\kappa$.
    Let $R,d\in\N$ such that $R \le dk$ and $k$ is even.
     Consider the $dk$-round \sos relaxation for the problem $\cP$,
    \begin{align}
	\max~ \pE\left[\iprod{\bT, x^{\tensor k}}\right] \qquad s.t. \quad \pE\left[ \|x\|^2_2\right] = \alpha \label{eq:normalpha}\mcom
    \end{align}
    Furthermore, let $T$ be the natural flattening of $\bT$ to an $n^{k/2}\times n^{k/2}$ tensor, let $\hat \cS_{dk/2}$ be the set of matrices that permute rows and columns of matrices in $[n]^{dk/2} \times [n]^{dk/2}$ according to actions of $\cS_{dk}$ on the coordinates in $[n]$
    Then in the $dk$-round \sos relaxation,
    \[
	\pE\left[\Iprod{\bT,x^{\otimes k}}\right] \le \alpha^{k/2} \cdot \left\|\E_{\Pi,\Sigma\in \hat\cS_{dk/2}}\left[\Pi T^{\otimes d} \Sigma\right]\right\|^{1/d}\mper
    \]
\end{proposition}
\begin{proof}
    We have that
    \begin{align*}
	 \left(\pE\left[\iprod{\bT, x^{\tensor k}}\right]\right)^{d}
	&\preceq \pE\left[ \left(\iprod{\bT, x^{\tensor k}}\right)^{d} \right]\qquad \text{(by \sos-convexity)}\\
	&= \pE\left[ \iprod{\bT^{\tensor d}, x^{\tensor dk}} \right] \qquad \text{(by the symmetry constraints \pref{eq:symconst})}\\
	&= \pE\left[\Iprod{\left(\E_{\Pi,\Sigma \in \hat \cS_{dk}}\left[ \Pi( \bT^{\tensor k})\Sigma\right]\right),  x^{\tensor dk}} \right] \qquad \text{(by \pref{eq:symconst})}\\
	 &\preceq \left\|\E_{\Pi,\Sigma \in \hat \cS_{dk/2}}\left[ \Pi( \bT^{\tensor k})\Sigma\right]\right\| \cdot \pE\left[ \|x\|^{dk}\right]  \qquad \text{(by \pref{lem:sos-matrix})}
\end{align*}
The conclusion follows from \pref{eq:normalpha}.
\end{proof}

As an immediate corollary of the above and of \pref{thm:D-upper}, we have \pref{thm:inj-tensor-informal} for even $k$.
To get \pref{thm:inj-tensor-informal} for odd $k$, we can apply Cauchy-Schwarz before applying \sos-convexity, so that we are working with
\begin{align*}
    \iprod{\bT,x^{\otimes k}}^2
    &\preceq \left(\sum_{i \in [n]} x_i^2\right)\cdot \Iprod{\sum_{i\in[n]} T_i \tensor T_i, x^{\otimes 2(k-1)}}\\
    &= \left(\sum_{i \in [n]} x_i^2\right)\cdot \left(\Iprod{\sum_{i\in[n]} T_i \tensor T_i - \squares(T_i \tensor T_i), x^{\otimes 2(k-1)}} + \sum_{i\in[n]}\sum_{A [n]^{k-1}} T_{A,i}^2 \prod_{j\in A} x_j^2 \right),
\end{align*}
where $T_i$ is the $i$th slice of $\bT$, and $\squares(T_i\tensor T_i)$ corresponds to the entries of $T_i \tensor T_i$ which are squares of the base variables $\bT$.
The right-hand term is bounded by obtaining a high-probability bound of $O(\log n)$ on the maximum coefficient $T_{iA}^2$, and the left-hand term is bounded by following the same steps as in the proof of \pref{prop:sos-even}, then applying \pref{thm:odd-tensor}.

\subsection{Bounds for $k$-XOR}

For the case of $k$-XOR, the proof is a bit more complicated than for the case of tensor norms, because the matrix certificates we used have certain rows and columns deleted.
Still, the arguments are similar to our proof from \pref{sec:lin}.
All steps in the proofs from \pref{sec:lin-norm-sec} and \pref{sec:odd-lin-norm-sec} we can make into \sos proofs in an analogous way to the tensor norm \sos proofs above, except for the steps in which the high-multiplicity rows and columns are deleted.
This too is not difficult to see, and we will prove it for the even case.
We will require the following \sos fact:
\begin{claim} \label{claim:l1normsos}
Suppose $\pE$ is a degree $2d$ pseudoexpectation functional with
Boolean constraints, i.e., $\pE[x_i^2 \cdot r(x)] = \pE[r(x)]$ for all
$r(x)$ with $\deg(r) \leq 2d-2$.

Let $q = \sum_{\sigma} \hat{q}_{\sigma} x_\sigma$ and $r$ be
polynomials such that $\deg(q r^2) \leq 2d$.  Then,
    \[
	\pE[ q(x) r^2(x)] \leq \pE[ \norm{\hat{q}}_1 \cdot r^2]\mper
	\]
\end{claim}
\begin{proof}
    Note that for each monomial $x_\sigma$,  $\pE[(1-x_{\sigma})] = \pE[
		(1-x_{\sigma})^2] \geq 0$.  Using this inequality for
		each of the monomials in $q$, the claim follows
		immediately.
\end{proof}

Now, we prove an \sos analogue of \pref{prop:formula}, which allows us to use the low-multiplicity restrictions of our certificate matrices to get our upper bounds.

\begin{proposition}\label{prop:formula-sos}
    Let $\Phi$ be a random $k$-XOR formula in which each clause is sampled independently with probability $p$.
    Let $\cC^d_{low} \subset [m]^d$ be the set of all ordered multisets of clauses $C_{i_1},\ldots,C_{i_d}$ from $\Phi$ with the property that if we form two multisets of variables $I,J \in [n]^{dk/2}$ with $I$ containing the first $k/2$ variables of each $C_{i_\ell}$ and $J$ containing the last $k/2$ variables of each $C_{i_\ell}$, then $I,J$ are both low-multipicity multisets, in that both have no element of $[n]$ with multiplicity $\ge 100\log n$.

    Then if $p \ge 200 \frac{\log n}{n^{k-1}}$ and $d \ll n $,
    and if $\pE$ is a pseudoexpectation of degree at least $2dk$, then
    \[
	\pE[P_{\Phi}(x)] \le \left(\E_{i_1,\ldots,i_d \sim \cC_{low}^d}\left[\prod_{\ell=1}^d P_{i_\ell}(x)\right]\right)^{1/d} + o\left(1\right)
    \]
for all $x \in \{\pm 1\}^n$ with high probability.
\end{proposition}
\begin{proof}

    We sample a uniform element  $\cC\sim\cC_{low}^d$, $\cC = C_1,\ldots,C_d$ in the following way:
    \begin{compactitem}
	\item For $t = 1,\ldots,d$:
	    Let $\cA_t \subset \cI$ be the set of clauses such that for any $C' \in \cA$, $C_1,\ldots,C_{t-1}, C' \in \cC_{low}^{t}$.
		Choose a uniformly random $C \sim \cA_t$ and set $C_t := C$, adding $C$ to $\cC$.
    \end{compactitem}
	This sampling process clearly gives a uniformly random element of $\cC_{low}^d$.

	Let $P_i(x)$ be the $0-1$ predicate corresponding to whether $x$ satisfies the clause $C_i$.
	Let $m_{\max}$ be the maximum number of clauses any variable in $\Phi$ participates in.
	Because $\pE\left[(P_i^2(x)-P_i(x))r(x)\right] = 0~ \forall r(x), \deg(r) \leq 2d$ we can write,
\begin{align*}
	\E_{C_1,\ldots,C_d \sim \cC_{low}^d} \pE\left[\prod_{i \in [d]} P_i(x)\right]
& =	\E_{C_1,\ldots,C_d \sim \cC_{low}^d} \pE\left[\prod_{i \in [d]}
    P^2_i(x)\right]  \\
    & =	\E_{C_1,\ldots,C_{d-1}} \pE\left[\left(\prod_{i \in [d-1]}
    P_i^2(x)\right) \cdot \left(P_{\Phi}(x) + \Delta_{C_1.\ldots,C_{d-1}}(x)
\right)\right]
\end{align*}
where $\Delta_{C_1,\ldots,C_{d-1}}(x) \defeq \E\left[ P_d^2(x) |
	C_1,\ldots,C_{d-1}\right] - \E[ P_d^2(x)] = \E\left[ P_d^2(x) |
	C_1,\ldots,C_{d-1}\right] - P_{\Phi}(x)$.  By definition, the
		$\ell_1$-norm of the coefficients of the polynomial
		$\Delta_{C_1,\ldots,C_{k-1}}$ is at most $d
		m_{max}/100 m \log n < o(1)$ with high probability, by concentration argument for $m$ and $m_{\max}$ (see the proof of \pref{prop:formula}) and by our requirement that $d \ll n$.
		Using \prettyref{claim:l1normsos}, this implies that
\begin{align*}
\E_{C_1,\ldots,C_{d-1}} \pE\left[\prod_{i \in [d-1]} P^2_i(x) \cdot \Delta_{C_1.\ldots,C_{d-1}}(x)\right]
& \preceq \E_{C_1,\ldots,C_{d-1}} \pE\left[\prod_{i \in [d-1]} P^2_i(x) \cdot
\norm{\Delta_{C_1.\ldots,C_{d-1}}(x)}_1\right]\\
    & \preceq o(1) \cdot \E_{C_1,\ldots,C_{d-1}} \pE\left[\prod_{i \in [d-1]}
P^2_i(x) \right]
\end{align*}
Therefore,
\begin{align*}
\E_{C_1,\ldots,C_d \sim \cC_{low}^d} \pE\left[\prod_{i \in [d]} P^2_i(x)\right]
    & \succeq  \E_{C_1,\ldots,C_{d-1}} \pE\left[ (P_{\Phi}(x) - o(1)) \prod_{i \in [d-1]} \cdot P_i^2(x)\right]
\end{align*}
Repeating the argument $d$ times, we can conclude that for even $d$,
\begin{align*}
\E_{C_1,\ldots,C_d \sim \cC_{low}^d} \pE\left[\prod_{i \in [d]} P^2_i(x)\right] & \succeq
    \pE \left[(P_{\Phi}(x) - o(1))^d\right]
     \geq \left(\pE\left[ P_{\Phi}(x) - o(1)\right]\right)^d
\end{align*}
Where the last inequality follows from \sos convexity.
This concludes the argument.
\end{proof}

This proposition, plugged into the argument from \pref{sec:lin-norm-sec} along with the \sos-ifying steps used for the tensor norm upper bound, gives \pref{thm:maindlin} for the even $k$ case.
The odd $k$ case can be obtained in a similar way.

\section*{Acknowledgements}
T.S. thanks Sam Hopkins for helpful conversations, and Jonah Brown-Cohen for helpful comments in the preparation of this manuscript.
\addreferencesection
\bibliographystyle{amsalpha}
\bibliography{writeup,pc}

\appendix
\section{Useful matrix concentration facts}\label{app:useful}

\restateprop{prop:tpm}
\begin{proof}
    For a positive semidefinite matrix $P$, $\|P\| \le \Tr(P)$.
    We apply this along with Markov's inequality:
    \begin{align*}
	\Pr[\|M\| \ge t]
	~=~ \Pr[\|(MM^\top)^\ell\| \ge t^{2\ell}]
	~\le~ \Pr[\Tr((MM^\top)^\ell) \ge t^{2\ell}]
	~\le~ \frac{1}{t^{2\ell}}\E[\Tr((MM^\top)^\ell)]
	~\le~ \frac{\beta}{t^{2\ell}}\mcom
    \end{align*}
    and the conclusion follows from taking $t = c \beta^{1/2\ell}$
\end{proof}

\subsection{Bound on the norm of a Rademacher matrix}
Here, we prove an upper bound on the norm of a Rademacher matrix.
Although tighter bounds are known (see e.g. \cite{NKV02}, we are off by a constant factor), we include this simpler, looser proof here in an effort to be self-contained.

The following lemma gives a bound on the size of an epsilon net needed to cover the unit sphere.
\begin{lemma}[see Lemma 5.2 in \cite{vershynin}]\label{lem:eps-net}
    For every $\epsilon > 0$, the unit Euclidean sphere $S^{n-1}$ equipped with the Euclidean metric has an $\epsilon$-net with volume at most $\left(1+\frac{2}{\epsilon}\right)^n$.
\end{lemma}

We are ready to prove our bound.
\begin{theorem}\label{thm:rademacher}
Let $A$ be an $n \times n$ symmetric matrix with i.i.d. Rademacher entries.
Then for all $s \ge 0$,
\[
    \Pr\left( |\|A\| - 12\sqrt{n}| \ge s\right) \le \exp(-t^2/16).
\]
\end{theorem}
\begin{proof}
    Let $\Lambda$ be an $\epsilon$-net over $\cS_{n-1}$,  with $\epsilon$ to be chosen later.
    By \pref{lem:eps-net}, we can choose $|\Lambda| \le (1 + \frac{2}{\epsilon})^n$.
    For any fixed $x \in \Lambda$,
    \[
	x^\top A x = \sum_{i<j} 2x_i x_j A_{ij} + \sum_{i} x_i^2 A_{ii}\mper
    \]
    Each $x_ix_j A_{ij}$ is an independent random variable.
    We have absolute bounds on the values of each variable, so we can apply a Hoeffding bound to this sum,
    \begin{align*}
	\Pr\left(x^\top A x \ge t\right)
	&\le \exp\left(\frac{-2t^2}{\sum_{i<j} (4x_i x_j)^2  + \sum_{i} (2 x_i^2)^2}\right)
	= \exp\left(\frac{-2t^2}{8\|x\|_2^4 - 4\|x\|_4^4}\right)
	\le \exp\left(\frac{-t^2}{4}\right)\mper
    \end{align*}
    Taking a union bound over $\Lambda$, we have
    \begin{align*}
	\Pr\left(\max_{x \in \Lambda} x^\top A x \ge t\right)
	\le (1+ \frac{2}{\epsilon})^n \cdot \exp\left(\frac{-t^2}{4}\right)\mper
    \end{align*}
    To extend the bound to any point $y \in \cS_{n-1}$, let $y$ be the maximizer of $y^\top A y$.
    We note that there must exist some $x \in \Lambda$ so that $\|x - y\| \le \epsilon$ by assumption.
    We have
    \[
	\left|y^\top Ay - x^\top A x\right|
	= \left|y^\top A(y-x) - x^\top A (x-y)\right|
	\le 2\epsilon \|A\|\mcom
    \]
    which by the triangle inequality implies
    \[
	\|A\| = y^\top A y \le \max_{x \in \Lambda} \{x^\top A x \}+ 2\epsilon \|A\|\qquad \implies \qquad
	\|A\| \le \frac{1}{1-2\epsilon} \max_{x\in\Lambda}\{ x^\top A x\}.
    \]
    Taking $\epsilon = 1/4$ and  $t = 2 \sqrt{ n \log(1 + \frac{2}{\epsilon})} + s$ concludes the proof.
\end{proof}

\end{document}